\documentclass{lmcs} %%% last changed 2014-08-20
\pdfoutput=1
\usepackage[utf8]{inputenc}

\usepackage{lastpage}
\lmcsdoi{19}{1}{14}
\lmcsheading{}{\pageref{LastPage}}{}{}%
{Jan.~11,~2022}{Feb.~28,~2023}{}


\usepackage{microtype}
\usepackage{soul}

\usepackage{amsmath}
\usepackage{amsthm}
\usepackage{amssymb}
\usepackage{mathtools}
\usepackage{environments}
\input{enforcer.sty}
\input{monitor.sty}
\input{logic.sty}
\input{shorthand.sty}
\input{pic.sty}
\input{normalization.sty}

\usepackage{enumitem}
\usepackage{graphicx}
\usepackage{mathpartir}
\usepackage{semantic}
\usepackage{pifont}
\usepackage{hyperref}
\usepackage{cleveref}

\usepackage{float}

\usepackage{subfig}

\usepackage{tikz}
\usetikzlibrary{calc}
\usetikzlibrary{shapes,arrows,positioning}
\usetikzlibrary{decorations.markings}
\usetikzlibrary{shapes.geometric}
\usetikzlibrary{decorations.pathmorphing}

\pgfdeclarelayer{edgelayer}
\pgfdeclarelayer{nodelayer}
\pgfsetlayers{edgelayer,nodelayer,main}

\tikzset{cross it/.style={dashed,decoration={markings, mark=at position 0.5 with {
				\draw [-,solid] ++ (-1mm,-1mm) -- (1mm,1mm);
				\draw [-,solid] ++ (-1mm,1mm) -- (1mm,-1mm);
			}}
			,postaction={decorate}}}
\tikzset{snake it/.style={decorate, decoration={snake,amplitude=.4mm,segment length=2mm,post length=1mm}}}
\tikzset{
	doubleline/.style args={#1 colored by #2 and #3}{
		-stealth,line width=#1,#2, % first arrow
		postaction={draw,-stealth,#3,line width=(#1)/3,
			shorten <=(#1)/3,shorten >=2*(#1)/3}, % second arrow
	}
}
\tikzstyle{innerWhite} = [semithick, white,line width=1.4pt, shorten >= 4.5pt]
\tikzstyle{roundnode} = [rounded rectangle, draw=black, thick, minimum size=7mm]
\tikzstyle{dottednode} = [circle, dotted, very thick, minimum size=7mm]

\usepackage{color}

\newcommand{\ans}{\actN{ans}\xspace}
\newcommand{\cls}{\actN{cls}\xspace}
\newcommand{\loggTuple}[2]{(\actN{log},#1,#2)\xspace} 
\newcommand{\logg}{\loggTuple{\dvV}{\dvVV}\xspace} 

\newcommand{\actLoggTuple}[2]{\actOut{\pidVV}{\loggTuple{#1}{#2}}}

\newcommand{\actReq}{\actIn{\pidV}{\dvV}\xspace}
\newcommand{\actAns}{\actOut{\pidV}{\dvVV}\xspace}
\newcommand{\actCls}{\actIn{\pidVV}{\dvVVV}\xspace}
\newcommand{\actLog}{\actOut{\pidVV}{\logg}\xspace}

\newcommand{\pVg} {\ensuremath{\pV_\textbf{a}}\xspace}
\newcommand{\pVbo} {\ensuremath{\pV_\textbf{b}}\xspace}
\newcommand{\pVbi} {\ensuremath{\pV_\textbf{c}}\xspace}
\newcommand{\pVcls} {\ensuremath{\pV_\textbf{cls}}\xspace}

\newcommand{\pVclsdef}{(\prf{\actCls}\match{\dvVVV}{\cls}{\nil}{\rV})}
\newcommand{\pVgdef}{\rec{\rV}{(\ch{(\prf{\actReq}\prf{\assign{\dvVV}{\ans(\dvV)}}{\prf{\actAns}}{\prf{\actLog}}{\rV})}{\pVcls})}}
\newcommand{\pVbidef}{\prf{\actIn{\pidV}{\dvVV}}{\pVg}}
\newcommand{\pVbodef}{\rec{\rV}{(\ch{(\prf{\actReq}\prf{\assign{\dvVV}{\ans(\dvV)}}{\prf{\actAns}}
			(\ch{
				{\prf{\underline{\actAns}}}{\prf{\actLog}}{\pVg}
			}{
				{\prf{\actLog}}{\rV}
			}))}{\pVcls})}}

\newcommand{\vVA}{\ensuremath{\vV_1}\xspace}
\newcommand{\vVB}{\ensuremath{\vV_2}\xspace}

\newcommand{\vVdef}{\ensuremath{\vV_\text{def}}\xspace}
\newcommand{\vVVA}{\ensuremath{\vVV_1}\xspace}
\newcommand{\vVVB}{\ensuremath{\vVV_2}\xspace}
\newcommand{\vVVC}{\ensuremath{\vVV_3}\xspace}

\newcommand{\actInPidVvVA}{\actIn{\pidV}{\vVA}\xspace}
\newcommand{\actInPidVvVB}{\actIn{\pidV}{\vVB}\xspace}
\newcommand{\actAnsPidVvVA}{\actOut{\pidV}{\ans(\vVA)}\xspace}

\newcommand{\trnsInIdS}{\actSN{\actIn{\dbV}{\dbVVA}}{\dvV{\neq}\pidVV}\!\xspace}
\newcommand{\trnsOutIdS}{\actSet{\actOut{\dvV}{\dbVVB}\!}\xspace}
\newcommand{\trnsOutSupS}{\actSN{\actOut{\dvV}{\dbVdc}}{\!\actdot}\xspace}
\newcommand{\trnsLogIdSFull}{\actSN{\actOut{\pidVV}{\dbVVC}}{\dvVVC{=}\loggTuple{\dvVVA}{\dvVVB}}\xspace}

\newcommand{\trnsInInsS}{\!\actSN{\actdot}{\actIn{\dvV}{\vVdef}}\xspace}

\newcommand{\pidVVOutClogPat}{\actSN{\actOut{\pidVV}{\dbVVC}}{\dvVVC{=}\loggTuple{\dvVVA}{\dvVVB}}\xspace} 

\newcommand{\dbVVVInAPat}{\actSN{\actIn{\dbV}{\dbVVA}}{\dvV{\neq}\pidVV}\xspace} 
\newcommand{\dvVVVInBPat}{\actSet{\actIn{\dvV}{\dbVdc}}\xspace} 
\newcommand{\dvVVVOutAPat}{\actSet{\actOut{\dvV}{\dbVVB}}\xspace} 
\newcommand{\dvVVVOutBPat}{\actSet{\actOut{\dvV}{\dbVdc}}\xspace} 

\newcommand{\hVdef}{\hMaxX{\hNec{\dbVVVInAPat}(\hAnd{\hNec{\dvVVVInBPat}\hFls\,}{\,\hNec{\dvVVVOutAPat}\hV_1'})}\xspace}
\newcommand{\hVdefAP}{(\hAnd{\hNec{\dvVVVOutBPat}\hFls\,}{\,\hNec{\pidVVOutClogPat}\hVarX})}

\newcommand{\dbVVVInAPatF}{\actSN{\actIn{\dbV}{\dbVVA}}{\dvV{\neq}\pidVV}\xspace} 
\newcommand{\dvVVVInBPatF}{\actSN{\actIn{\dbVA}{\dbVdc}}{\dvVA{=}\dvV}\xspace} 
\newcommand{\dvVVVOutAPatF}{\actSN{\actOut{\dbVB}{\dbVVB}}{\dvVB{=}\dvV}\xspace} 
\newcommand{\dvVVVOutBPatF}{\actSN{\actOut{\dbVC}{\dbVdc}}{\dvVC{=}\dvV}\xspace} 
\newcommand{\pidVVOutClogPatF}{\actSN{\actOut{\dbVD}{\dbVVC}}{\dvVD{=}\pidVV\land\dvVVC{=}\loggTuple{\dvVVA}{\dvVVB}}\xspace} 

\newcommand{\hVdefF}{\hMaxX{\hNec{\dbVVVInAPatF}(\hAnd{\hNec{\dvVVVInBPatF}\hFls\,}{\,\hNec{\dvVVVOutAPatF}\hV_1'})}\xspace}
\newcommand{\hVdefAPF}{(\hAnd{\hNec{\dvVVVOutBPatF}\hFls\,}{\,\hNec{\pidVVOutClogPatF}\hVarX})}

\newcommand{\eVa} {\ensuremath{\eV_\textbf{a}}\xspace}
\newcommand{\eVe} {\ensuremath{\eV_\textbf{e}}\xspace}
\newcommand{\eVd} {\ensuremath{\eV_\textbf{d}}\xspace}
\newcommand{\eVdt} {\ensuremath{\eV_\textbf{dt}}\xspace}
\newcommand{\eVdtP} {\ensuremath{\eV'_\textbf{dt}}\xspace}
\newcommand{\eVdet} {\ensuremath{\eV_\textbf{det}}\xspace}
\newcommand{\eVdetP} {\ensuremath{\eV'_\textbf{det}}\xspace}
\newcommand{\eVdetPP} {\ensuremath{\eV''_\textbf{det}}\xspace}

\newcommand{\eVclsdef}{\prf{\actSet{\actIn{\pidVV}{\dbVdc}}}{\eIden}}

\hypersetup{
  colorlinks,
  linkcolor=black,
  citecolor=black,
  urlcolor=black
}

\renewcommand{\eBI}[2]{\eI{#1}{#2}}
\renewcommand{\eSembiS}[1]{\ensuremath{{\llparenthesis\,#1\, \rrparenthesis}}} 
\renewcommand{\remff}[1]{\ensuremath{\llangle#1\rrangle^{\bullet}}\xspace}

\title[First-order Bidirectional Runtime Enforcement]{Bidirectional Runtime Enforcement of First-Order 
Branching-Time Properties 
}
\author[L. Aceto]{Luca Aceto\lmcsorcid{0000-0002-2197-3018}}[a,b]	
\author[I. Cassar]{Ian Cassar\lmcsorcid{0000-0002-5845-3753}}[a,c]
\author[A. Francalanza]{Adrian Francalanza\lmcsorcid{0000-0003-3829-7391}}[c]
\author[A. Ing\'{o}lfsd\'{o}ttir]{Anna Ing\'{o}lfsd\'{o}ttir\lmcsorcid{0000-0001-8362-3075}}[a]

\address{ICE-TCS, Department of Computer Science, Reykjav\'{i}k University, Iceland}
\address{Gran Sasso Science Institute, L'Aquila, Italy}	
\address{Department of Computer Science, University of Malta, Malta}

\thanks{``TheoFoMon: Theoretical Foundations for Monitorability'' (nr.163406-051) and “MoVeMnt: Mode(l)s of Verification and Monitorability” (nr.217987-051) of the Icelandic Research Fund, the EU H2020 RISE programme under the Marie Sk{\l}odowska-Curie grant agreement nr. 778233, the Italian MIUR project PRIN 2017FTXR7S IT MATTERS ``Methods and Tools for Trustworthy Smart Systems'', the Research Excellence Funds of the University of Malta, 2021 nr. I22LU01, and the Endeavour Scholarship Scheme (Malta), part-financed by the European Social Fund (ESF) - Operational Programme II – 2014-2020.}

\newcommand{\nw}[1]
{#1}
\begin{document}

\begin{abstract}
	Runtime enforcement is a dynamic analysis technique that instruments a monitor with a system in order to ensure its correctness as specified by some property. 
	This paper explores \emph{bidirectional} enforcement strategies for properties describing the input and output behaviour of a system.
	We develop an operational framework for bidirectional enforcement and use it to study the enforceability of the safety fragment of Hennessy-Milner logic with recursion (\SHML).
	We provide an automated synthesis function that generates correct 
  % and optimal 
  monitors from \SHML formulas, and show that this logic is enforceable via a specific type of bidirectional enforcement monitors called action disabling monitors.
\end{abstract}

\maketitle

%\pagestyle{headings}
% \pagenumbering{arabic}

\section{Introduction}
\label{sec:introduction}
% !TEX root = journal.tex

\emph{Runtime enforcement} (RE)~\cite{Ligatti2005,Falcone2008}
% ,Cassar2018Concur,Cassar2019RV}
 is a dynamic verification technique that uses \emph{monitors} to analyse the runtime behaviour of a system-under-scrutiny (\sus) and transform it in order to conform to some correctness \emph{specification}.
The earliest work in RE \cite{Ligatti2005,Ligatti2009,Sakarovitch:2009:Transducers,Bielova2011,Khoury2012} models the behaviour of the \sus as a \emph{trace} of \emph{abstract} actions (\eg $\acta,\actb,\ldots \in \Act$).
Importantly, it assumes that the monitor can either suppress or replace \emph{any} of these (abstract) actions and, whenever possible, insert additional actions into the trace.

This work has been used as a basis to implement \emph{unidirectional} enforcement approaches \cite{Konighofer2017,Falcone2012,Cassar2018Concur,Alur:2011:Transducers} that monitor the outputted trace of actions emitted by the \sus as illustrated by \Cref{fig:enf-setups}(a).
In this setup, the monitor is instrumented with the \sus to form a \emph{composite system} (represented by the dashed enclosure in \Cref{fig:enf-setups}) and is tasked with transforming the output behaviour of the \sus to ensure its correctness.
For instance, an erroneous output \actb of the \sus is intercepted by the monitor and 
% modified accordingly 
transformed into $\actb'$, to stop the error from propagating to the surrounding environment.

Despite its merits, unidirectional enforcement disregards the fact that not all events originate from the \sus. 
For instance, protocol specifications describing the interaction of communicating computational entities include \emph{input} actions, instigated by the environment in addition to output actions originating from the \sus.
%
% lacks the power to enforce properties involving the \emph{input} behaviour of the \sus.
%
Arguably, these properties are  harder to enforce. 
Since inputs are instigated by the environment, the \sus possesses only partial control over them and the capabilities to prevent or divert such actions can be curtailed. 
Moreover, in such a bidirectional setting, the properties to be enforced tend to be of a \emph{first-order} nature \cite{Cassar2018Concur,Havelund2018}, describing relationships between the respective payload carried by input and output events.
This means that even when the (monitored) \sus can control when certain inputs can be supplied (\eg by opening a communication port, or by reading a record from a database \etc), the environment still has control over the provided payload.
Broadly, there are two approaches to enforce bidirectional properties at runtime. 
Several bodies of work employ two monitors attached at the output side of each (diadic) interacting party~\cite{Bocchi2017,Jia2016Popl,Chen2012,FrancalanzaMT20}. 
As shown in \Cref{fig:enf-setups}(b), the extra monitor is attached to the \emph{environment} to analyse its outputs before they are passed on as \emph{inputs} to the \sus.
While this approach is effective, it assumes that a monitor can actually be attached to the environment (which is often inaccessible).

By contrast, \Cref{fig:enf-setups}(c) presents a less explored \emph{bidirectional enforcement} approach where the monitor analyses the entire behaviour of the \sus without the need to instrument the environment. 
The main downside of this alternative setup is that it enjoys limited control over the \sus's inputs.
As we already argued, the monitor may be unable to enforce a property that could be violated by an input action with an invalid payload value. 
In other cases, the monitor might need to adopt a different enforcement strategy to the ones that are conventionally used for enforcing output behaviour in a unidirectional one.

\begin{figure}
  % [h]
% [t]
	\centering
	% {\normalsize(a)}
\begin{minipage}{0.3\textwidth}
	\begin{tikzpicture}[scale=1.3, transform shape]
	%Rectangles
	\draw (0,0) rectangle (0.5,0.6);  %% System
	\draw[fill=black] (1,0) rectangle (1.5,0.3)node[pos=.5]{\textcolor{white}{$\blacktriangleright$}}; %% Monitor
	\draw (2.2,0) rectangle (2.7,0.6);  %% Environment
	\draw[densely dashed] (-0.05,-0.05) rectangle (1.6,0.65);  %% Composite System
	%Names
	\node at (0.25,-0.2) {{\scriptsize \sus}};
	\node at (1.3,-0.2) {{\scriptsize Mon}};
	\node at (2.45,-0.2) {{\scriptsize Env}};
  %Label
  \node at (1.3,-0.6) {{\scriptsize (a)}};
	%System Nodes
	\node (sa) at (0.4,0.4) {};	
	\node (sb) at (0.4,0.2) {};	
	%Mon Left Nodes
	\node (mlb) at (1.1,0.2) {};	
	%Mon Right Nodes
	\node (mrb) at (1.4,0.2) {};	
	%Env Nodes
	\node (ea) at (2.3,0.4) {};
	\node (eb) at (2.3,0.2) {};
	% Solid Arrows (Above)
	\foreach \from/\to/\val in {sb/mlb/{\actb}, mrb/eb/{$\actb'$}}
	\draw [->] (\from) -- (\to) node[pos=.5,yshift=-1.2mm] {\scriptsize\val};
	
	% \draw [->] (ea) -- (sa) node[pos=.5,xshift=4.7mm,yshift=1.2mm] {\scriptsize\acta};
	\end{tikzpicture}
\end{minipage}
\qquad \qquad\qquad
%
% {\normalsize(b)}
\begin{minipage}{0.3\textwidth}
	\begin{tikzpicture}[scale=1.3, transform shape]
	%Rectangles
	\draw (0,0) rectangle (0.5,0.6);  %% System
	\draw[fill=black] (1,0) rectangle (1.5,0.25)node[pos=.5]{\textcolor{white}{$\blacktriangleright$}}; %% Monitor OUT
	\draw[fill=black] (1,0.35) rectangle (1.5,0.6)node[pos=.5]{\textcolor{white}{$\blacktriangleleft$}}; %% Monitor IN
	\draw (2.2,0) rectangle (2.7,0.6);  %% Environment
	 
	\draw[densely dashed] (-0.05,-0.05) -- (1.55,-0.05) -- (1.55,0.28) -- (0.6,0.28) -- (0.6,0.65) -- (-0.05,0.65) -- cycle;  %% Composite System OUT
	\draw[densely dashed] (2.75,-0.05) -- (2.1,-0.05) -- (2.1,0.31) -- (0.95,0.31) -- (0.95,0.65) -- (2.75,0.65) -- cycle;  %% Composite System IN
	
	%Names
	\node at (0.25,-0.2) {{\scriptsize \sus}};
	\node at (1.3,-0.2) {{\scriptsize Mons}};
	\node at (2.45,-0.2) {{\scriptsize Env}};
  %Label
  \node at (1.3,-0.6) {{\scriptsize (b)}};
	%System Nodes
	\node (sa) at (0.4,0.4) {};	
	\node (sb) at (0.4,0.2) {};	
	%Mon Left Nodes
	\node (mla) at (1.1,0.4) {};	
	\node (mlb) at (1.1,0.2) {};	
	%Mon Right Nodes
	\node (mra) at (1.4,0.4) {};	
	\node (mrb) at (1.4,0.2) {};	
	%Env Nodes
	\node (ea) at (2.3,0.4) {};
	\node (eb) at (2.3,0.2) {};
	
	% Solid Arrows (Above)	
	\draw [<-] (sa) -- (mla) node[pos=.5,yshift=1.4mm,xshift=0.5mm] {\scriptsize$\acta'$};
	\draw [<-] (mra) -- (ea) node[pos=.5,yshift=1mm] {\scriptsize$\acta$};
		
	% Solid Arrows (Below)
	\foreach \from/\to/\val in {mlb/sb/{\actb}, eb/mrb/{$\actb'$}}
	\draw [<-] (\from) -- (\to) node[pos=.5,yshift=-1.2mm] {\scriptsize\val};
	\end{tikzpicture}
\end{minipage}
\quad\;\,
%
% {\normalsize(c)}
\begin{minipage}{0.3\textwidth}
	\begin{tikzpicture}[scale=1.3, transform shape]
	%Rectangles
	\draw (0,0) rectangle (0.5,0.6);  %% System
	\draw[fill=black] (0.95,0.1) rectangle (1.55,0.5) node[pos=.5]{\textcolor{white}{$\blacktriangleleft\blacktriangleright$}}; %% Monitor
	\draw (2.2,0) rectangle (2.7,0.6);  %% Environment
	\draw[densely dashed] (-0.05,-0.05) rectangle (1.62,0.65);  %% Composite System
	%Names
	\node at (0.25,-0.2) {{\scriptsize \sus}};
	\node at (1.3,-0.2) {{\scriptsize Mon}};
	\node at (2.45,-0.2) {{\scriptsize Env}};
  %Label
  \node at (1.3,-0.6) {{\scriptsize (c)}};
	%System Nodes
	\node (sa) at (0.4,0.4) {};	
	\node (sb) at (0.4,0.2) {};
	%Mon Left Nodes
	\node (mla) at (1.05,0.4) {};	
	\node (mlb) at (1.05,0.2) {};
	%Mon Right Nodes
	\node (mra) at (1.44,0.4) {};	
	\node (mrb) at (1.44,0.2) {};
	%Env Nodes
	\node (ea) at (2.3,0.4) {};
	\node (eb) at (2.3,0.2) {};
%	% Solid Arrows (Above)
%	\foreach \from/\to/\val in {sa/mla/{$\acta'$}, mra/ea/{\acta}}
%	\draw [<-] (\from) -- (\to) node[pos=.5,yshift=1.2mm] {\scriptsize\val};
	
	\draw [<-] (sa) -- (mla) node[pos=.5,yshift=1.4mm,xshift=0.5mm] {\scriptsize$\acta'$};
	\draw [<-] (mra) -- (ea) node[pos=.5,yshift=1mm] {\scriptsize$\acta$};

	% Solid Arrows (Below)
	\foreach \from/\to/\val in {mlb/sb/{\actb}, eb/mrb/{$\actb'$}}
	\draw [<-] (\from) -- (\to) node[pos=.5,yshift=-1.2mm] {\scriptsize\val};
	\end{tikzpicture}
\end{minipage}
  % \vspace{-2mm}
	\caption{Enforcement instrumentation setups.}
	\label{fig:enf-setups}
\end{figure}

This paper explores how  existing monitor transformations---namely, suppressions, insertions and replacements---can be repurposed to work for bidirectional enforcement, \ie the setup in \Cref{fig:enf-setups}(c).
Since inputs and outputs must be enforced differently, we find it essential to distinguish between the monitor's transformations  and their resulting effect on the visible behaviour of the composite system.
This permits us to study the enforceability of properties defined via a safety subset  of the well-studied branching-time logic \recHML~\cite{Rathke1997,Aceto2007Book,Larsen1990} (a reformulation of the modal $\mu$-calculus \cite{Kozen1983MuCalc}), called \SHML.
A crucial aspect of our investigation is the \emph{synthesis} function that maps the \emph{declarative} safety \recHML specifications to \emph{algorithmic} monitors that enforce properties concerning both the input and output behaviour of the \sus. 
Since monitors are part of the trusted computing base, it was imperative that we ensure that all synthesised monitors are correct~\cite{FrancalanzaS13,francalanza2016theory,BasinDHHMKKMRST22}.
Our contributions are:
\begin{enumerate}[$(i)$]
	\item A general instrumentation framework for bidirectional enforcement (\Cref{fig:mod-bi-re}) 
  that is parametrisable by any system whose behaviour can be modelled as a labelled transition system. 
  The framework subsumes the one presented in previous work~\cite{Cassar2018Concur} and differentiates between the enforcement of input and output actions.
	% The novelty of this framework lies in how its effect on the visible behaviour of the resulting composite system differs according to whether the transformed action is an input or an output.
	%
	%The way this instrumentation framework interprets the transformations applied by the monitor %(\ie action suppression, insertion and replacements), 
	%differs according to whether the action %that must be enabled, disabled or adapted, 
	%is an input or an output.
	%
	%
	\item A novel formalisation describing what it means for a monitor to \emph{adequately enforce} a property in a bidirectional setting (\Cref{def:enforceability,def:enforcement}).
	These definitions are parametrisable with respect to an instrumentation relation, an instance of which is given by our enforcement framework of \Cref{fig:mod-bi-re}.
	\item A new result showing that the subclass of \emph{disabling monitors}, \Cref{def:disa-mon} (the counterpart to suppression monitors in unidirectional enforcement), suffices to bidirectionally enforce safety properties expressed as  \recHML formulas (\Cref{thm:enf}).  
   A by-product of this result is a synthesis function (\Cref{def:synthesis-bi}) that generates a disabling monitor from such safety formulas. 
  %
	% We evaluate the quality of this mapping by proving that the synthesised monitors are correct and optimal (\Cref{thm:enf,thm:oenf}).
  \item A preliminary investigation on the notion of monitor \emph{optimality} 
  (\Cref{def:opt-enf-bi}). 
	Our proposed definition assesses the level of intrusiveness of the monitor and guides in the search for the least intrusive one.
  We evaluate our monitor synthesis function of \Cref{def:synthesis-bi} in terms of this optimality measure, \Cref{thm:oenf}. 
\end{enumerate}
This article is the extended version of the paper titled \emph{``On Bidirectional Runtime Enforcement''} that appeared at FORTE 2021 \cite{AcetoCFI21}.
In addition to the material presented in the conference version, this version contains extended examples, the proofs of the main results and new material on monitor optimality.  
The related work section has also been considerably expanded. 

% Full proofs and additional details can to be found at~\cite{BidirectionalTechRep20,Cassar21Thesis}.

\section{Preliminaries}
\label{sec:preliminaries}
% !TEX root = journal.tex

\emph{The Model.} 
We assume a countably infinite set of communication ports $\pidV,\pidVV,\pidVVV{\,\in\,}\Port$, a set of values $\vV,\vVV{\,\in\,}\Val$, and partition the set of actions \Act into 
\begin{itemize}
  \item \emph{inputs}, $\actInV{\,\in\,}\IAct$ \eg denoting an input by the system from the environment on port \pidV carrying payload \vV; and 
  \item \emph{outputs}, $\actOutV{\,\in\,}\OAct$ originating from the system to the interacting environment on port \pidV carrying payload \vV
\end{itemize}
where $\Act=\IAct{\,\cup\,}\OAct$.
Systems are described as \emph{labelled transition systems} (LTSs); these are
%
% An LTS consists of a 
triples $\langle\Sys,\Act\cup\sset{\actt},\rightarrow\rangle$ consisting of a set of \emph{system states}, $\pV,\pVV,\pVVV{\,\in\,}\Sys$, a set of \emph{visible actions}, $\acta,\actb{\,\in\,}\Act$, along with a distinguished silent action $\actt{\,\notin\,}\Act$ (where $\actu{\,\in\,}\Act{\,\cup\,}\sset{\actt}$), and a \emph{transition} relation, $\reduc\;\subseteq(\Sys\times(\Act\cup\sset{\actt})\times\Sys)$.
We write $\pV\traS{\actu}\pVV$ in lieu of $(\pV,\actu,\pVV) \in\, \rightarrow$, and $\pV \wtraS{\acta} \pVV$ to denote weak transitions representing $\pV (\traS{\actt})^{\ast}\cdot\traS{\acta}\pVV$ where $\pVV$ is called the \acta-derivative of $\pV$.
For convenience, we use the syntax of the regular fragment of value-passing CCS~\cite{Hennessy1996CCS} to concisely describe LTSs.
Traces $\tr,\trr\in\Act^\ast$ range over (finite) sequences of \emph{visible} actions. 
We write $\pV \wtraS{\tr} \pVV$ to denote a sequence of \emph{weak} transitions  $\pV \wtraS{\acta_1} \ldots \wtraS{\acta_n} \pVV$ where $\tr = \acta_1\ldots\acta_n$ for some $n\geq 0$; when $\tr{\,=\,}\trE$, $\pV\wtraS{\trE}\pVV$ means $\pV\traSC{\actt}\pVV$.
Additionally, we represent system runs as \emph{explicit traces} that include \actt-actions, $\txr,\txrr{\,\in\,}(\Act{\,\cup\,}\sset{\actt})^\ast$ and write $\pV\traS{\actu_1\ldots\actu_n}\pVV$ to denote a sequence of \emph{strong} transitions $\pV \traS{\actu_1} \ldots \traS{\actu_n}\pVV$.
The function \trcsys{\txr} returns a 
% canonical 
system that produces exclusively the sequence of actions defined by \txr, \nw{modulo the data carried by input actions in \txr that cannot be controlled by the receiving process}.
For instance, 
 $\trcsys{\prf{\actIn{\pidV}{3}}\prf{\actt}\actOut{\pidV}{5}}$ produces the process $\prf{\actIn{\pidV}{\dvV}}\prf{\actt}\prf{\actOut{\pidV}{5}}\nil$. 
We consider states in our system LTS modulo the classic notion of \emph{strong bisimilarity}~\cite{Hennessy1996CCS,Sangiorgi2011Bisim} and write $\pV\bisim\pVV$ when states $\pV$ and $\pVV$ are bisimilar.
% Our system LTS induces the classic notion of \emph{strong bisimilarity} \cite{Hennessy1996CCS,Sangiorgi2011Bisim}, $\pV\bisim\pVV$, which is a commonly used form of system equivalence.

\begin{figure}[t]
	\textbf{\textsf{Syntax}}
	% \vspace{-1mm}
	\!\!{\small$$\begin{array}{r@{\,}llc@{\,}llc@{\,}ll}
	\hV,\hVV \in \SHML \bnfdef&  \hTru &(\text{truth}) & \bnfsepp  &\hFls \; &(\text{falsehood})& \bnfsepp &\hBigAndU{i\in\IndSet}{\hV_i} & (\text{conjunction})  \\[1mm]
	 	\bnfsepp& \hNec{\pate,\bV}{\hV} & (\text{necessity}) &\bnfsepp & \hMaxXF & (\text{greatest fp.}) &\bnfsepp& \hVarX & (\text{fp. variable})
\end{array}$$} 
% \vspace{-2mm}
\\[1em]
	\textbf{\textsf{Semantics}}
	{
	{\small	
  % \[
	  \begin{align*}
      % {r@{\,}c@{\;}l@{\qquad}r@{\,}c@{\;}l@{\;\;}r@{\,}c@{\;}l}
		\hSemS{\hTru,\rho}  &\defEquals \Sys \\
		\hSemS{\hFls,\rho}  &\defEquals  \emptyset \\
		% &\hspace{-5mm}
    \hSemS{\hVarX,\rho} &\defEquals \rho(\hVarX)
		\\
    % [1mm]
		\hSemS{\hBigAnd{i\in\IndSet}{\hV_i},\rho} & \defEquals  \bigintersectU{i\in\IndSet}\hSemS{\hV_i,\rho} 
    %&
    \\
		% \!
    \hSemS{\hMaxXF,\rho} & \defEquals  \bigcup \Set{S \;|\;  S \subseteq \hSemS{\hV,\rho[\hVarX\mapsto S]}}
		\\
    % [1mm]
		%\hSemS{\hBigOr{i\in\IndSet}{\hV_i},\rho} & \defEquals & \bigunionU{i\in\IndSet}\hSemS{\hV_i,\rho} &
		%\!\hSemS{\,\hMinXF,\rho} & \defEquals & \bigcap \Set{S \;|\;  \hSemS{\hV,\rho[\hVarX\mapsto S]} \subseteq S\,}
		%\\[1mm]
		\hSemS{\,\hNec{\pate,\bV}{\hV},\rho}  & \defEquals
			% \multicolumn{7}{l}{
			\!\!\!\Set{\pV \;|\;  \forall\acta,\pVV,\sV\cdot (\pV  \wtra{\acta} \pVV  \text{ and } \mtch{\pate}{\acta}{=}\sV \text{ and } \ceval{\bV\sV\!}{\!\boolT}) \text{ implies } \pVV{\,\in\,}\hSemS{\hV\sV,\rho}\!}
      % }
		%	\\[1mm]
		%\hSemS{\hSuf{\actSN{\pate}{\bV}}{\hV},\rho}  & \defEquals&
		%	\multicolumn{7}{l}{
		%	\!\!\!\Set{\pV \;|\; \exists\acta,\pVV,\sV\cdot(\pV  \wtra{\acta} \pVV   \text{ and } \mtch{\pate}{\acta}{=}\sV \text{ and } \ceval{\bV\sV}{\boolT} \text{ and } \pVV \in \hSemS{\hV\sV,\rho}) \!}
		%}
	\end{align*}
	% \]
  }
	}
  % \vspace{-5mm}
	\caption{The syntax and semantics for \SHML, the safety fragment of the branching-time logic \recHML\cite{Larsen1990}.}
	\label{fig:recHML} 
  % \vspace{-4mm}
\end{figure}

\medskip
\noindent
\emph{The Logic.} The behavioral properties we consider are described using \SHML~\cite{Aceto1999TestingHML,Francalanza2017FMSD}, a subset of the value passing \recHML \cite{Rathke1997,Hennessy1995} that uses \emph{symbolic actions} of the form \actSN{\pate}{\bV} consisting of an action pattern \pate and a condition \bV. 
Symbolic actions facilitate reasoning about LTSs with infinitely many actions (\eg inputs or outputs carrying data from infinite domains).
They abstract over concrete actions using \emph{data variables} $\dvV,\dvVV,\dvVVV\in\DVars$ that occur free in the constraint \bV or as binders in the pattern \pate. 
% denoted as $\dbV$.
Patterns are subdivided into input \patInV and output \patOutV patterns where \dbV binds the information about the port on which the interaction has occurred, whereas \dbVV binds the payload;
$\bv{\pate}$ denotes the set of binding variables in \pate whereas \fv{\bV} represents the set of free variables in condition \bV.
We assume a (partial) \emph{matching function}  $\mtch{\pate}{\acta}$ that (when successful) returns the (smallest) substitution $\sV:\DVars \rightharpoonup (\Port \cup \Val)$, mapping bound variables in \pate to the corresponding values in \acta ; by replacing every occurrence \dbV in \pate with $\sV(\dvV)$ we get the matched action \acta.
The \emph{filtering condition}, \bV, 
% may refer to variables bound in \pate and it 
is evaluated \wrt the substitution returned by successful matches, written as $\ceval{\bV\sV\!}{\!v}$ where $v\in\set{\boolT,\boolF}$.

Whenever a symbolic action $\actSNV{}$ is \emph{closed}, \ie $\fv{\bV}{\,\subseteq\,}\bv{\pate}$, it denotes the \emph{set} of actions $\hSemS{\actSNV{}}{\,\defeq\,}\Setdef{\acta}{\exists\sV\cdot\mtch{\pate}{\acta}{\,=\,}\sV \andt \ceval{\bV\sV\!}{\!\boolT}}$.
For example, we can have $\hSemS{\bigl(\patOutV,\; (x = \pidV \vee x = \pidVV) \;\wedge\; y \geq 3 \bigr)} = \{ \actOut{\pidV}{3}, \actOut{\pidVV}{3}, \actOut{\pidV}{4}, \actOut{\pidVV}{4}, \actOut{\pidV}{5}, \actOut{\pidVV}{5},\actOut{\pidV}{6}, \actOut{\pidVV}{6}, \ldots  \}.$
%  and allows more adequate reasoning about LTSs with infinitely many actions (\eg inputs or outputs carrying data from infinite domains).
%
Following standard (concrete) value-passing LTS semantics~\cite{Milner1992CCS,Hennessy1996CCS}, our systems have \emph{no control} over the data values supplied via inputs. 
Accordingly,  we assume a well-formedness constraint where the condition \bV of an input symbolic action, \actSN{\patInV}{\bV}, \emph{cannot} restrict the values of binder \dvVV, \ie $\dvVV{\,\notin\,}\fv{\bV}$.
%
% Put differently, for a closed input symbolic action $\actSN{\patInV}{\bV}$, if $\sV$ and $\sV'$ are substitutions that agree on \dvV then $\ceval{\bV\sV}{\boolT}$ iff $\ceval{\bV\sV'}{\boolT}$.
%
% By contrast, as the system produces the output values itself, a closed output symbolic action, \actSN{\patOutV}{\bV}, may also describe \emph{specific} output values by restricting the possible values of \dvVV in \bV. 
%
As a shorthand, whenever a condition in a symbolic action equates a bound variable to a specific value we embed the equated value within the pattern, 
 \eg $\actSN{\patOutV}{\dvV{\,=\,}\pidV\land\dvVV{\,=\,}3}$, 
 \actSN{\patInV}{\dvV{\,=\,}\pidV} 
and 
\actSN{\patInV}{\dvV{\,=\,}\dvVVV} become 
\actSN{\actOut{\pidV}{3}}{\boolT}, \actSN{\actIn{\pidV}{\dbVV}}{\boolT} and \actSN{\actIn{\dvVVV}{\dbVV}}{\boolT} \resp ; we also elide \boolT conditions, and occasionally just write \actSet{\actOut{\pidV}{3}} and \actSet{\actIn{\pidV}{\dbVV}} in lieu of \actSN{\actOut{\pidV}{3}}{\boolT} and \actSN{\actIn{\pidV}{\dbVV}}{\boolT} when the meaning of this shorthand can be inferred from the context.

%
%It therefore denotes the set of output actions produced by the system, \ie $\Setdef{\!\actOutV\!}{\!\actOutV{\,\in\,}\OAct\cdot\mtch{\patOutV}{\actOutV}{\,=\,}\subb{\SubE{\pidV}{\dvV},\SubE{\vV}{\dvVV}} \andt \ceval{\bV\subb{\SubE{\pidV}{\dvV},\SubE{\vV}{\dvVV}}\!}{\!\boolT}\!\!}$\!.
%
%\begin{example} Symbolic action $\actSN{\patInV}{\dvVV{=}1}$ is valid since $(\fv{\dvVV{=}1}=\set{\dvVV})\subseteq(\bv{\patInV}=\set{\dvV,\dvVV})$, but actions $\actSN{\actIn{\dbV}{\dvVV}}{\dvVV=1}$ and $\actSN{\actIn{\dbV}{1}}{\dvVV=1}$ are invalid since $\fv{\dvVV{=}1}\nsubseteq(\bv{\actIn{\dbV}{\dvVV}}=\set{\dvV})$. \qed
%\end{example}
%
%\noindent Two symbolic actions, $\actSN{\pate_1}{\bV_1}$ and $\actSN{\pate_2}{\bV_2}$, are said to be \emph{equivalent} when $\hSemS{\actSN{\pate_1}{\bV_1}}=\hSemS{\actSN{\pate_2}{\bV_2}}$, and \emph{pattern equivalent} when $\hSemS{\actSN{\pate_1}{\boolT}}=\hSemS{\actSN{\pate_2}{\boolT}}$.
%

%
\Cref{fig:recHML} presents the \SHML syntax for some countable set of logical variables $\hVarX,\hVarY\!\in\!\LVars$.
The construct $\hBigAnd{i\in\IndSet}{\hV_i}$ describes a \emph{compound} conjunction,  $\hAnd{{\hV_1}}{\hAnd{\ldots}{{\hV_n}}}$, where $\IndSet=\sset{1,..,n}$  is a finite set of indices. 
The syntax also permits recursive properties using greatest fixpoints, \hMaxXF, which bind free occurrences of \hVarX in \hV.
The central construct is the (symbolic) universal modal operator, \hNec{\pate,\bV}{\hV}, where the binders \bv{\pate} bind the free data variables in \bV and \hV.
We occasionally use the notation \dbVdc to denote ``don't care'' binders in the pattern \pate, whose bound values are not referenced in \bV and \hV.
We also assume that all fixpoint variables, $\hVarX$, are guarded by modal operators.
%  by a modal necessity (\eg $\hMax{\hVarX}{(\hNec{\acta}\hFls \hAndS \hVarX)}$ is invalid, unlike $\hMax{\hVarX}{(\hNec{\actb}\hFls\hAndS\hNec{\acta}\hVarX)}$ in which $\hVarX$ is guarded by $\hNec{\acta}$).

%
Formulas in \SHML are interpreted over the system powerset domain where $S{\in}\pset{\Sys}$.
The semantic definition of \Cref{fig:recHML}, \hSemS{\hV,\rho},  is given for \emph{both} open and closed formulas. It employs a valuation  from logical variables to sets of states, $\rho\in(\LVars \rightarrow \pset{\Sys})$, which permits an inductive definition on the structure of the formulas;
$\rho'=\rho[\hVarX\mapsto S]$ denotes a valuation where $\rho'(\hVarX) = S$ and  $\rho'(\hVarY) = \rho(\hVarY)$ for all other $\hVarY\neq \hVarX$.
The only non-standard case is that for the universal modality formula, \hNec{\pate,\bV}{\hV}, which is satisfied by any system that either \emph{cannot} perform an action \acta that matches \pate while satisfying condition \bV, or for any such matching action \acta with substitution \sV, its derivative state satisfies the continuation $\hV\sV$.
%The only non-standard cases are those for the modal formulas, as we use \SAs.
We consider formulas modulo associativity and commutativity of $\hAndS$, and unless stated explicitly, we assume \emph{closed} formulas, \ie without free logical and data variables. 
Since the interpretation of a closed \hV is independent of the valuation $\rho$ we write $\hSemS{\hV}$ in lieu of $\hSemS{\hV,\rho}$.
A system \pV \emph{satisfies} formula \hV whenever $\pV{\,\in\,}\hSemS{\hV}$, and a formula \hV is \emph{satisfiable}, when %there exists a system \pVV such that $\pVV\in\hSemS{\hV}$, \ie
$\isSatF$.

We find it convenient to define the function \afterFS, describing how an \SHML formula \emph{evolves} in reaction to an action \actu.
Note that, for the case $\hV=\hNec{\pate,\bV}{\hVV}$, the formula returns $\hVV\sV$  when \actu matches successfully the symbolic action $(\pate,\bV)$ with \sV, and \hTru otherwise, to signify a trivial satisfaction.
\begin{defi} 
  We define the function $\afterFS:(\SHML{\times}\Act{\,\cup}\set{\actt})\,{\rightarrow}\,\SHML$ as:
  {	
    \normalfont
    \begin{align*}
    \afterF{\hV}{\acta} &\defeq 
    \begin{xbrace}{c@{\quad}l}
    \hV & \text{if }\hV{\,\in\,}\Set{\hTru,\hFls}\\[2mm]
    \afterF{\hV'\sub{\hV}{\hVarX}}{\acta} & \text{if }\hV{\,=\,}\hMaxX{\hV'}\\[2mm]
    \hBigAndU{i\in\IndSet}\afterF{\hV_i}{\acta} & \text{if }\hV{\,=\,}\hBigAnd{i\in\IndSet}\hV_i \\[2mm]
    \hVV\sV & \hspace{-8mm}\text{if }\hV{\,=\,}\hNec{\pate,\bV}{\hVV} \andt \exists\sV{\cdot}(\mtch{\pate}{\acta}{=}\sV \land \ceval{\bV\sV\!}{\!\boolT}) \\[2mm]
    \hTru & \hspace{-8mm}\text{if }\hV{\,=\,}\hNec{\pate,\bV}{\hVV} \andt \nexists\sV{\cdot}(\mtch{\pate}{\acta}{=}\sV \land \ceval{\bV\sV\!}{\!\boolT})
    \end{xbrace}\\
    \afterF{\hV}{\actt} &\defeq \hV  \tag*{\exqed}
    \end{align*}
  }
  \end{defi}
  We  abuse notation and 
 lift the \afterFS function to (explicit) traces in the obvious way, \ie $\afterF{\hV}{\txr}$ is equal to $\afterF{\afterF{\hV}{\actu}}{\txrr}$ when $\txr=\actu\txrr$ and to $\hV$ when $\txr=\txrE$.
 Our definition of \afterFS is justified vis-a-vis the semantics of \Cref{fig:recHML} via \Cref{thm:sem-just-after}; it will play a role later on when defining our notion of enforcement in \Cref{sec:enforcement}.

\begin{rem} \label{rem:after-welldef}
	The function $\afterFS$ is well-defined due to our assumption that formulas are guarded, guaranteeing that $\hV'\sub{\hV}{\hVarX}$ has fewer top level occurrences of greatest fixpoint operators than $\hMaxX{\hV'}$. \exqed
\end{rem}
\begin{prop} \label[proposition]{thm:sem-just-after}
	For every system state \pV, formula \hV and action \acta, if $\pV{\,\in\,}\hSemS{\hV}$ and $\pV{\wtraS{\acta}}\pV'$ then $\pV'{\,\in\,}\hSemS{\afterF{\hV}{\acta}}$. \qed
\end{prop}
%
% A proof for this theorem is given in \Cref{sec:appendix}.
% %
% \noindent We abuse notation and lift the \afterFS function to (explicit) traces in the obvious way, \ie $\afterF{\hV}{\txr}$ is equal to $\afterF{\afterF{\hV}{\actu}}{\txrr}$ when $\txr=\actu\txrr$ and to $\hV$ when $\txr=\txrE$.

\begin{exa} \label[example]{ex:shml-formula-bi} 
	The \emph{safety} property $\hV_1$ \emph{repeatedly} requires that \emph{every} input request that is made on a port that is \emph{not} \pidVV, cannot be followed by another input on the same port in succession.
	However, following this input it allows a \emph{single} output answer on the same port in response, followed by the logging of the serviced request by outputting a notification on a dedicated port \pidVV. 
	We note how the channel name bound to $x$ is used to constrain sub-modalities. 
  Similarly, values bound to $\dvVVA$ and $\dvVVB$ are later referenced in condition $\dvVVC{\,=\,}\loggTuple{\dvVVA}{\dvVVB}$.
	\begin{align*}
		\hV_1&\defeq\hVdef \\
		\hV'_1&\defeq\hVdefAP
	\end{align*}	
	%
	 %to specify the type of data output by the system.
	%
	% The formula is violated by two consecutive inputs on the same port \dvV, %\ie $\hNec{\pidVInAPat}\hNec{\pidVInBPat}\hFls$, 
	% and when a request is serviced with multiple answers. %, \ie $\hNec{\pidVOutAPat}\hNec{\pidVOutBPat}\hFls$; 
	% %
	% An answer output followed by a log action on port \pidVV is normal, and thus the formula recurses. %, \ie $\hNec{\pidVOutAPat}\hNec{\pidVVOutClogPat}\hVarX$.
	% %
	%
	% \noindent 
  Consider the systems \pVg, \pVbo and \pVbi:
	\begin{align*}
		\pVg&\defeq\pVgdef\\ 
     & \text{(where $\pVcls{\,\defeq\,}\pVclsdef$)} \\
		\pVbo&\defeq\pVbodef \\
    \pVbi&\defeq\pVbidef 
	\end{align*}
	The system \pVg implements a \emph{request-response} server that repeatedly inputs values (for some domain \Val) on port \pidV, $\actReq$, for which it internally computes an answer and assigns it to the data variable \dvVV, $\assign{\dvVV}{\ans(\dvV)}$. 
	Subsequently, it outputs the answer on port \pidV in response to each request, $\actAns$, and then logs the serviced request pair of values by outputting the triple $\logg$ on port $\pidVV$, $\actLog$.
	It terminates whenever it receives a close request $\cls$ from port $\pidVV$, \ie $\actCls$ when $\dvVVV{\,=\,}\cls$.
	Systems \pVbo and \pVbi are similar to \pVg but define additional behaviour: \pVbi requires a startup input, \actIn{\pidV}{\dvVV}, before behaving as \pVg, whereas \pVbo occasionally provides a redundant (underlined) answer prior to logging a serviced request.
	
	Using the semantics of \Cref{fig:recHML}, one can verify that the first system satisfies our correctness property $\hV_1$, \ie $\pVg \in \hSemS{\hV_1}$.  
  However the second system \pVbo does not satisfy this property because it can inadvertently answer twice a request, \ie  $\pVbo \notin \hSemS{\hV_1}$ since we have $\pVbo\wtraS{\actInPidVvVA.\actAnsPidVvVA.\actAnsPidVvVA}$ (for some value  $\vVA$).
  Analogously, the third system \pVbi violates property $\hV_1$ because it can accept two consecutive inputs on port \pidV (without answering the preliminary request first), \ie 
  $\pVbi{\,\notin\,}\hSemS{\hV_1}$ since $\pVbi\wtraS{\actInPidVvVA.\actInPidVvVB}$  (for any pair of values  $\vVA$ and $\vVB$). \exqed 
\end{exa}

%
% \begin{example}
% 	When applied to $\hV_1$ of \Cref{ex:shml-formula-bi} in relation to $\txr^1=\prf{\actInPidVvVA}\prf{\actt}\actAnsPidVvVA$, formula $\afterF{\hV_1}{\txr^1}$ denotes $(\hAnd{\hNec{\actSet{\actOut{\pidV}{\dbVdc}}}\hFls\,}{\,\hNec{\actSN{\actOut{\pidVV}{\dbVVC}}{\dvVVC{=}\loggTuple{\vVA}{\ans(\vVA)}}}\hV_1})$ and when applied to trace $\prf{\actInPidVvVA}\actInPidVvVB$ it evolves into \hFls. \qed
% \end{example}

%
 
%

\section{A Bidirectional Enforcement Model}
\label{sec:model}
% !TEX root = journal.tex

Bidirectional enforcement seeks to transform the entire (visible) behaviour of the \sus in terms of output actions  (instigated by the \sus itself, which in turn controls the payload values being communicated) and input actions (originating from the interacting environment which chooses the payload values); this contrasts with unidirectional approaches that only modify output traces.
%
% When changing the behaviour of a system (and not just a single trace) 
In this richer setting, it helps to differentiate between the transformations performed by the monitor (\ie insertions, suppressions and replacements), and the way they can be used to affect the resulting behaviour of the composite system.
In particular, we say that: 
\begin{itemize}
  \item an action that can be performed by the \sus has been \emph{disabled} when it is no longer visible in the resulting composite system (consisting of the \sus and the monitor);
  \item an action is \emph{enabled} when the composite system can execute it while the \sus cannot;
  \item an action is \emph{adapted} when either its payload  differs from that of the composite system, or when the action is rerouted through a different port.
\end{itemize}

%

% \footnote{We shall not consider the case where an input (\resp output) action is changed to an output (\resp input)action.}
%
% However, since inputs and outputs are fundamentally different, the type of the action itself cannot be adapted, that is, an input cannot become an output, and vice versa.
%

\begin{figure}
  % [t]
	\centering
	\begin{tikzpicture}[scale=1, transform shape]
	%Rectangles
	\draw[fill=gray] (-1.8,-2.6) rectangle (1.8,2); %% Monitor
	\draw (-5.0,-2.6) rectangle (-4,2);  %% System
	\draw (4.5,-2.6) rectangle (5.5,2);  %% Environment
	\draw[densely dashed] (-5.1,-2.7) rectangle (1.9,2.1); %% Composite system
	%Names
	\node at (0,2.4) {Monitor};
	\node at (-4.5,2.4) {System};
	\node at (4.8,2.4) {Environment};
	%System Nodes
	
	\node (sa) at (-4,1.5) {\hspace{-1.5em}(a)};	
	\node (sb) at (-4,0.9) {\hspace{-1.5em}(b)};
	\node (sc) at (-4,0.3) {\hspace{-1.5em}(c)};
	\node (sd) at (-4,-0.3) {\hspace{-1.5em}(d)};
	\node (se) at (-4,-0.9) {\hspace{-1.5em}(e)};	
	\node (sf) at (-4,-1.5) {\hspace{-1.5em}(f)};
	\node (sg) at (-4,-2.1) {\hspace{-1.5em}(g)};
	
	%Mon Left Nodes
	\node (mla) at (-1.75,1.5) {};	
	\node (mlb) at (-1.75,0.9) {};
	\node (mlc) at (-1.75,0.3) {};
	\node (mld) at (-1.75,-0.3) {};
	\node (mle) at (-1.75,-0.9) {};
	\node (mlf) at (-1.75,-1.5) {};		
	\node (mlg) at (-1.75,-2.1) {};
	
	%Mon Right Nodes
	\node (mra) at (1.7,1.5) {};	
	\node (mrb) at (1.7,0.9) {};
	\node (mrc) at (1.7,0.3) {};
	\node (mrd) at (1.8,-0.3) {};
	\node (mre) at (1.8,-0.9) {};
	\node (mrf) at (1.8,-1.5) {};	
	\node (mrg) at (1.8,-2.1) {};
	
	% Monitor Internal Nodes
	\node[text width=8.5em, align=center,draw, fill=white] (mia) at (0,1.5) {\scriptsize suppress output};
	\node[text width=8.5em, align=center,draw, fill=white] (mib) at (0,0.9) {\scriptsize modify output};
	\node[text width=8.5em, align=center,draw, fill=white] (mic) at (0,0.3) {\scriptsize insert output};
	\node[text width=8.5em, align=center,draw, fill=white] (mid) at (0,-0.3) {\scriptsize block input};
	\node[text width=8.5em, align=center,draw, fill=white] (mie) at (0,-0.9) {\scriptsize 
  % block and 
  \nw{insert input}};
	\node[text width=8.5em, align=center,draw, fill=white] (mif) at (0,-1.5) {\scriptsize modify input};
	\node[text width=8.5em, align=center,draw, fill=white] (mig) at (0,-2.1) {\scriptsize suppress input};
	
	%Env Nodes
	\node (ea) at (4.5,1.5) {};
	\node (eb) at (4.5,0.9) {};
	\node (ec) at (4.5,0.3) {};
	\node (ed) at (4.5,-0.3) {};
	\node (ee) at (4.5,-0.9) {};
	\node (ef) at (4.5,-1.5) {};
	\node (eg) at (4.5,-2.1) {};
	
	% Solid Arrows (Actual Action)
	\foreach \from/\to/\val in {sa/mla/{output}, sb/mlb/{output}, ef/mrf/{input},  eg/mrg/{enabled input}}
	\draw [->] (\from) -- (\to) node[pos=.5,above] {\scriptsize\val};
	
	% Crossed Arrows (Disabled/absent Action)
	\foreach \from/\to/\val in {mra/ea/{disabled output}, sc/mlc/{no output}, mld/sd/{no input}, ed/mrd/{disabled input}, ee/mre/{disabled input}, mlg/sg/{no input}}
	\draw [->,cross it, gray] (\from) -- (\to) node[pos=.5,above] {\scriptsize\val};
	
	% Squigglie Arrows (Modified Action)
	\foreach \from/\to/\val in { mrb/eb/{\;modified output}, mrc/ec/{enabled output},  mle/se/{default input}, mlf/sf/{modified input} }
	\draw [->,snake it] (\from) -- (\to) node[pos=.5,above] {\scriptsize\val};
	
\end{tikzpicture}
	\caption{Disabling, enabling and adapting bidirectional runtime enforcement via suppressions, insertions and replacements.}
	\label{fig:en-dis-act}
  % \vspace{-4mm}
\end{figure}

We argue that implementing action enabling, disabling and adaptation differs according to whether the action is an input or an output; see 
 \Cref{fig:en-dis-act}. 
%  we illustrate our proposed instrumentation setup that implements them by using the monitor's existing transformations. 
%
Enforcing actions instigated by the \sus---such as outputs---is more straightforward. 
\Cref{fig:en-dis-act}(a), (b) and (c) \resp state that disabling an output can be achieved by suppressing it, adapting an output amounts to replacing the payload or redirecting it to a different port, whereas output enabling  can be attained via an insertion. 
However, enforcing actions instigated by the environment such as inputs is harder. 
%For the instrumentation 
%
In \Cref{fig:en-dis-act}(d), we propose to \emph{disable} an input by concealing the input port.
Since this may block the \sus from progressing, the instrumented monitor may additionally \emph{insert} a default input to unblock the system \nw{waiting to input on the channel used for the insertion},
% \footnote{The added benefits of this mechanism are further discussed in the forthcoming sections.},  
\Cref{fig:en-dis-act}(e),
\nw{in cases where the environment fails to provide the corresponding output}.
Input \emph{adaptation}, \Cref{fig:en-dis-act}(f), is also attained via a \emph{replacement}, albeit applied in the opposite direction to the output case.
%
% In fact, it modifies the data received by the monitor over some port, and forwards it to the \sus over the same (or a different) port.
%
Inputs can also be \emph{enabled} whenever the \sus is unable to carry them out, \Cref{fig:en-dis-act}(g), by having the monitor accept the input in question and then \emph{suppress} it. Note that, from the perspective of the environment, the input would still be effected.

\begin{figure}
  % [tb]
	\noindent\textbf{\textsf{Syntax}}
	\begin{align*}
	\eV,\eVV\in\Trn
	&\bnfdef \; \eTrns{\pate}{c}{\pate'}{\eV} \;
	\bnfsepp   \chBigText{i\in\IndSet}\, \eV_i \; (\IndSet\text{ is a finite index set}) \;
	\bnfsepp  \rec{\rV}{\eV} \;
	\bnfsepp  \rV\\
	\end{align*}
  \\
	\noindent\textbf{\textsf{Dynamics}}
  % \vspace{-2mm}
	\begin{mathpar}
		%		\inference[\rtit{eId}]{ }{\eIden \traS{\ioact{\actu}{\actu}} \eIden }
		%		\and
		\inference[\rtit{eSel}]{\eV_j \traSS{\actggp} \eVV_j}{
			\chBigText{i\in\IndSet}\,\eV_i \traSS{\actggp} \eVV_j}[$j{\in}\IndSet$]
		\and
		\inference[\rtit{eRec}]{\eV\sub{\rec{\rV}{\eV}}{\rV} \traS{\actggp} \eVV }{\rec{\rV}{\eV} \traS{\actggp} \eVV}
		\and
		\inference[\rtit{eTrn}]{
			% \mtchS{\actSN{\pate}{\bV}}{\acta}{\,=\,}\sV
			\mtch{\pate}{\actg} = \sV
			&&
			\ceval{\bV\sV}{\boolT}
			&&
			\actgp{\,=\,}\actp\sV
		}{\eTrns{\pate}{c}{\actp}{\eV} \traS{\actggp} \eV\sV}
		\\
		% \vspace{-4mm}
	\end{mathpar} 
  \\
	\noindent\textbf{\textsf{Instrumentation}} 	
  % \vspace{-1mm}
	\begin{mathpar}
		\inference[\rtit{biTrnO}]
		{\pV \traS{\actOutVV}\pV' & \eV\tra{\ioact{\actOutVVB}{\actOutVB}}\eVV}
		{ \eBI{\eV}{\pV} \traS{\actOutV} \eBI{\eVV}{\pV'}}
		\and
		\inference[\rtit{biTrnI}]
		{\eV\tra{\ioact{\actInVB}{\actInVVB}}\eVV & \pV \traS{\actInVV}\pV' }
		{ \eBI{\eV}{\pV} \traS{\actInV} \eBI{\eVV}{\pV'}}
		\\
		\inference[\rtit{biDisO}]
		{\pV \traS{\actOutV}\pV' & \eV\traS{\ioact{(\actOutV)}{\actdot}}\eVV}
		{ \eBI{\eV}{\pV} \traS{\actt} \eBI{\eVV}{\pV'}}
		\and
		\inference[\rtit{biDisI}]
		{\eV\traS{\ioact{\actdot}{(\actInV)}}\eVV & \pV \traS{\actInV}\pV'  }
		{ \eBI{\eV}{\pV} \traS{\actt} \eBI{\eVV}{\pV'}}
		\\
		%\and
		%
		\inference[\rtit{biEnO}]
		{%\pV \ntra{\actOutV} & 
			\eV\traS{\ioact{\actdot}{(\actOutV)}}\eVV}
		{ \eBI{\eV}{\pV} \traS{\actOutV} \eBI{\eVV}{\pV}}
		\and
		\inference[\rtit{biEnI}]
		{\eV\traS{\ioact{(\actInV)}{\actdot}}\eVV % & \pV \ntra{\actInV} 
		}
		{ \eBI{\eV}{\pV} \traS{\actInV} \eBI{\eVV}{\pV}}
		\and
		\inference[\rtit{biAsy}]
		{ \pV \traS{\actt}\pV' }
		{ \eBI{\eV}{\pV} \traS{\actt} \eBI{\eV}{\pV'} }
		\and
		\inference[\rtit{biDef}]
		{\pV \traS{\actOutV}\pV'
			& \eV\ntra{\actOutV}
			& \forall\,\pidVV{\in}\Port,\vVV{\in}\Val\cdot\eV\ntra{\ioact{\actdot}{\actOutVV}}
		}
		{ \eBI{\eV}{\pV} \traS{\actOutV} \eBI{\eIden}{\pV'}}
	\end{mathpar}
	% where $\eIden$ is shorthand for $\eIdenFull$ and $\eV{\!\!\ntraSS{\actg}}$ means $\nexists\actgp,\eVV{\cdot}\eV{\tra{\actggp}}\eVV$.
	\caption[]{A bidirectional instrumentation model for enforcement monitors.}
	\label{fig:mod-bi-re}
  % \vspace{-5mm}
\end{figure}

\Cref{fig:mod-bi-re} presents an operational model for the bidirectional instrumentation proposal of \Cref{fig:en-dis-act} in terms of (symbolic) transducers. 
A variant of these transducers was originally introduced in \cite{Cassar2018Concur} for unidirectional enforcement.  
Transducers, $\eV,\eVV{\,\in\,}\Trn$, are monitors that define \emph{symbolic transformation triples}, \actSTN{\pate}{\bV}{\actp}, consisting of an action \emph{pattern} \pate, \emph{condition} \bV, and a \emph{transformation action} $\actp$. 
Conceptually, the action pattern and condition determine the range of system (input or output) actions upon which the transformation should be applied, while the transformation action specifies the transformation that should be applied. %, \ie a suppression, replacement or insertion.
The symbolic transformation pattern \pate is an extended version of those definable in symbolic actions, that may also include \actdot;
when $\pate{\,=\,}\actdot$, it means that the monitor can act independently from the system to insert the action specified by the transformation action. 
Transformation actions are possibly open actions (\ie actions with possibly free variable such as \actIn{\dvV}{\vV} or \actOut{\pidV}{\dvV}) or the special action \actdot; the latter represents the suppression of the action specified by \pate.
We assume a well-formedness constraint where, for every \prf{\actSTN{\pate}{\bV}{\actp}}{\eV},  $\pate$ and $\actp$ cannot both be \actdot, and when neither is, they are of the \emph{same} type \ie an input (\resp output) pattern and action.
Examples of well-formed symbolic transformations are:
\begin{itemize}
  \item $\actSTN{\actdot}{\boolT}{\actInV}$,  inserting an input on port \pidV with value \vV;
  \item $\actSTN{\patOutV}{y \geq 5}{\actdot}$,  suppressing an output action carrying a payload that is greater or equal to 5; and 
  \item $\actSTN{\patOutV}{x = \pidVV}{\actOut{\pidV}{y}}$, redirecting (\ie adapting) outputs on port \pidVV carrying the payload $y$ (learnt dynamically at runtime) to port \pidV.
\end{itemize}

%
% For instance, symbolic transformations \actSTN{\actdot}{\boolT}{\actInV} and \actSTN{\patOutV}{\boolT}{\actdot} are valid since only one of their patterns is \actdot, and so is \actSTN{\patOutV}{\boolT}{\actOutV} since both patterns are output patterns. 
%
% This constraint ensures that input actions cannot be adapted into outputs and vice versa.
%
% It is crucial since inputs and outputs are instigated by different entities, namely, the environment and the \sus respectively.
%

The monitor transition rules in \Cref{fig:mod-bi-re} assume closed terms, \ie every \emph{transformation-prefix transducer} of the form \prf{\actSTN{\pate}{\bV}{\actp}}{\eV} must obey the closure constraint stating that $\bigl(\fv{\bV} {\,\cup\,} \fv{\actp} {\,\cup\,} \fv{\eV}\bigr) {\,\subseteq\,} \bv{\pate}$. A similar closure requirement is assumed for recursion variables \rV and \rec{\rV}{\eV}.
%  {\,\cup\,} \mathcal{V}$, where $\mathcal{V}$ is the set of variables that are bound by priorly defined transformation prefixes.
%
Each transformation-prefix transducer yields an LTS with labels of the form \ioact{\actg}{\actgp}, where $\actg,\actgp\in(\Act\,{\cup}\sset{\actdot})$. %and $\actdot$ is a monitor action. 
Intuitively, transition $\eV \traS{\ioact{\actg}{\actgp}} \eVV$
denotes the way that a transducer in state \eV \emph{transforms} the action $\actg$ into $\actgp$ while transitioning to state \eVV.
The transducer action \ioact{\acta}{\actb} represents the \emph{replacement} of \acta by \actb, \ioact{\acta}{\acta} denotes the
% obvious
\emph{identity} transformation,
whereas \ioact{\acta}{\actdot} and \ioact{\actdot}{\acta} respectively denote the \emph{suppression} and \emph{insertion} transformations of action \acta. 
%
%In the former, \actdot signifies the removal of action \acta from the execution of the monitored system, while in the latter it represents a monitor transition that introduces an action \acta that was not induced by the system.
%
%
The key transition rule in \Cref{fig:mod-bi-re} is \rtit{eTrn}.
It states that the transformation-prefix transducer $\prf{\actSTN{\pate}{\bV}{\actp}}{\eV}$ transforms action $\actg$ into a (potentially) different action \actgp and reduces to state $\eV\sV$, whenever \actg matches pattern \pate, \ie $\mtch{\pate}{\actg}{=}\sV$, and satisfies condition \bV, \ie \ceval{\bV\sV}{\boolT}.
Action \actgp results from instantiating the free variables in $\actp$ as specified by \sV, \ie $\actgp{=}\actp\sV$.
The remaining rules for selection (\rtit{eSel})  and recursion (\rtit{eRec}) are standard.
We employ the shorthand notation $\eV{\!\!\ntraSS{\actg}}$ to mean $\nexists\actgp,\eV' $ such that $\eV{\tra{\actggp}}\eV$.
Moreover, for the semantics of \Cref{fig:mod-bi-re}, we can encode the identity transducer/monitor, \eIden, as follows
\begin{equation}
  \label[equation]{eq:id-def}
  \eIden \defeq \eIdenFull.
\end{equation}
When instrumented with any arbitrary system, the identity monitor \eIden leaves its behaviour unchanged.
As a shorthand notation, we write \prf{\actSIDs{\pate}{\bV}}{\eV} instead of \prf{\actSTN{\pate}{\bV}{\actp}}{\eV} when all the binding occurrences $(\dvV)$ in $\pate$ correspond to free occurrences $\dvV$ in $\actp$, thus denoting an identity transformation.
Similarly, we elide \bV whenever $\bV= \btrue$.
The first contribution of this work lies in the new \emph{instrumentation relation} of \Cref{fig:mod-bi-re},  linking the behaviour of the \sus \pV with that of a monitor \eV: %that \emph{agrees} with the (observable) actions \Act of \pV.
the term \eBI{\eV}{\pV} denotes their composition as a \emph{monitored system}. 
Crucially, the instrumentation rules in \Cref{fig:mod-bi-re} give us a semantics in terms of an LTS over the actions $\Act{\,\cup}\sset{\actt}$, in line with the LTS semantics of the \sus. 
%
% Concretely, r
Following \Cref{fig:en-dis-act}(b), rule \rtit{biTrnO} states that if the \sus transitions with an output \actOutVV to $\pV'$ and the transducer  can \emph{replace} it with \actOutV and transition to $\eVV$, the \emph{adapted} output can be externalised so that the composite system $\eBI{\eV}{\pV}$ transitions over \actOutV to \eBI{\eVV}{\pV'}.
Rule \rtit{biDisO} states that if \pV performs an output \actOutV that the monitor \emph{can suppress}, the instrumentation withholds this output and the composite system silently transitions; this amounts to action \emph{disabling} as outlined in \Cref{fig:en-dis-act}(a).
Rule \rtit{biEnO} is dual, and it \emph{enables} the output \actOutV  on the \sus as outlined in \Cref{fig:en-dis-act}(c): it augments the composite system \eBI{\eV}{\pV} with an output \actOutV whenever \eV can \emph{insert} \actOutV, independently of the behaviour of \pV.
\nw{Rules \rtit{biDisO}, \rtit{biTrnO} and \rtit{biEnO} therefore correspond to items $(a)$, $(b)$ and $(c)$ in \Cref{fig:en-dis-act} respectively.}
Rule \rtit{biDef} is analogous to standard rules for premature monitor termination \cite{francalanza2016theory,Francalanza2017FMSD,Fra17:Concur,Achilleos2018Fossacs}, and accounts for underspecification of transformations.
We, however, restrict defaulting (termination)  to output actions performed by the \sus exclusively, \ie
a monitor only defaults to \eIden when it cannot react to or enable a system output.
By forbidding the monitor from defaulting upon unspecified inputs, the monitor is able to \emph{block} them from becoming part of the composite system's behaviour.
Hence, any input that the monitor is unable to react to, \ie $\eV\ntraS{\ioact{\actInV}{\actg}}$, is considered as being \emph{invalid and blocked} by default.
This technique is thus used to implement \Cref{fig:en-dis-act}(d).
To avoid disabling valid inputs unnecessarily, the monitor must therefore explicitly define symbolic transformations that cover \emph{all} the valid inputs of the \sus.
%
% For instance, the symbolic transformation $\actSTN{\actIn{\pidV}{\dbV}}{\boolT}{\actIn{\pidV}{\dvV}}$ allows values to be input on port \pidV only, while $\actSTN{\actIn{\dbVV}{\dbVdc}}{\dvVV{\neq}\pidVV}{\actdot}$ allows inputs on any port except \pidVV; any other input is invalid and thus blocked.
%
% By including such symbolic transformations, rules \rtit{biTrnI} and \rtit{biEnI} can be applied. %and thus prevent blocking the specified inputs.
%
Note, that rule \rtit{biAsy} still allows the \sus to silently transition  independently of \eV.
Following \Cref{fig:en-dis-act}(f), rule \rtit{biTrnI}  adapts inputs, provided the \sus can accept the adapted input. 
%
% As far as the environment is concerned, the \sus accepted the input provided by the environment on port \pidV.
%
Similarly, rule \rtit{biEnI} \emph{enables} an input on a port \pidV as described in \Cref{fig:en-dis-act}(g): the composite system accepts the input while suppressing it from the \sus. 
%
% Although unspecified inputs on a port \pidV are implicitly \emph{disabled} (since the monitor cannot react to them, \ie $\eV\ntraS{\ioact{\actInV}{\actg}}$), rule \rtit{biDisI} prevents the monitor from blocking systems that require the blocked input in order to progress. 
%
% Specifically, this rule 
Rule \rtit{biDisI} allows the monitor to generate a default input value \vV and forward it to the \sus on a port \pidV, thereby unblocking it \nw{whenever the environment is unable to provide the corresponding output on channel \pidV (carrying \vV)};  
\nw{from the environment's perspective,} 
% externally, 
the composite system silently transitions to some state, following \Cref{fig:en-dis-act}(e).
\nw{
It is worth comparing  rule \rtit{biDisI} with the other instrumentation rule \rtit{biEnO} discussed earlier, since they both handle outputs inserted by the monitor.
In the case of rule \rtit{biEnO}, whenever the monitor inserts an output to be consumed \emph{by the environment}, this is expressed at the level of the composite system as an external output (see conclusion of rule \rtit{biEnO}) since, in our LTS, the actions represent the interaction between the (composite) system and the environment.
Contrastingly, whenever the monitor inserts an output to be input \emph{by the \sus}, then this is expressed at the level of the composite system as a silent action (see conclusion of rule \rtit{biDisI}) since no interaction occurs between the (composite) system and the environment.
}
\nw{We conclude our discussion of the instrumenation rules in \Cref{fig:mod-bi-re} by remarking that rules \rtit{biDisI}, \rtit{biTrnI} and \rtit{biEnI} respectively implement items $(e)$, $(f)$ and $(g)$ of \Cref{fig:en-dis-act}.}
%
% Finally, rule \rtit{biAsy} allows the \sus \pV to internally transition with a silent action \actt to some state $\pV'$ independent of \eV.
%

\nw{
\begin{defi} \label[definition]{def:disa-mon}
  We call disabling monitors/transducers those monitors that only perform disabling actions.  The same applies to enabling and adapting monitors/transducers. \exqed
\end{defi}  
}

\begin{exa}	\label[example]{ex:transducers}
	Consider the following action disabling transducer \eVd, that repeatedly disables every output performed by the system via the branch 
  % $\prf{\actSN{\patOutV}{\actdot}}{\rVV}$.
  $\prf{\actSN{\actOut{\dbVdc}{\dbVdc}}{\actdot}}{\rVV}$. 
  In addition, it limits inputs to those on port \pidVV via the input branch $\prf{\actSet{\actIn{\pidVV}{\dbVdc}}}{\rVV}$;   
  inputs on other ports are disabled since none of the relevant instrumentation rules in \Cref{fig:mod-bi-re} can be applied.
  %  to allow the composite system to transition over these input actions.
	\begin{align*}
	\eVd & \!\!\defeq\!\! \rec{\rVV}{\ch{\prf{\actSet{\actIn{\pidVV}{\dbVdc}}}{\rVV}}{\prf{\actSN{\actOut{\dbVdc}{\dbVdc}}{\actdot}}{\rVV}}}
	\end{align*}
	%
	% \noindent It is a recursive transducer, $\rec{\rVV}{\_}$, 
	%
	%By not defining a branch for transforming inputs on ports other than \pidVV (such as $\prf{\actSN{\patInV}{\dvV\neq\pidVV}}{\rVV}$), it also disables every input interaction performed on these ports.
	%
	%It, however, allows inputs to occur on \pidVV as it defines the branch $\prf{\actSet{\actIn{\pidVV}{\dbVV}}}{\rV}$.
	%
	% Moreover, by only defining the input branch $\prf{\actSet{\actIn{\pidVV}{\dbVdc}}}{\rVV}$ it also restricts the composite system by allowing it to only input values from port \pidVV.
	%
	% Concretely, inputs from other ports are disabled since none of the instrumentation rules in \Cref{fig:mod-bi-re} can be applied to allow the composite system to transition over these input actions.
	%
	%As a side effect, it also blocks every subsequent action.
	%
  Recall the two systems below from \Cref{ex:shml-formula-bi}:
  \begin{align*}
		\pVbo&\defeq\pVbodef \\
    \pVbi&\defeq\pVbidef 
    \intertext{where}
    \pVg & \defeq\pVgdef \qquad \text{and}\\
    \pVcls & \defeq\pVclsdef 
	\end{align*}
	When instrumented with the system \pVbi, monitor \eVd blocks its initial input, \ie we have $\eBI{\eVd}{\pVbi}\ntraS{\acta}$ for any \acta.
	In the case of \pVbo, the composite system \eBI{\eVd}{\pVbo} can only input  requests on port \pidVV, such as the termination request $\eBI{\eVd}{\pVbo}\traS{\actIn{\pidVV}{\cls}}\eBI{\eVd}{\nil}$.
	%
	%
	% Now consider the more elaborate transducer \eVdt.
	%
	\begin{align*}
	\eVdt & \defeq 
	\rec{\rV}{(
		\ch{
			\prf{\trnsInIdS\,}
			(\ch{ \prf{\actSN{ \actIn{\dbVA}{\dbVdc}}{\dvVA\neq\dvV}}{\eIden}\!}
			{\!\prf{\trnsOutIdS}{
					\eVdtP
				}
			}
			)\! 
		}
		{
			\eVclsdef
		}
	)} 
	\\
	\eVdtP & \defeq  \ch{\prf{\trnsOutSupS}{\eVd}}{\ch{\defmonE}{\prf{\trnsLogIdSFull}{\rV}\!}} 
	\end{align*}
	By defining branch $\eVclsdef$, the more elaborate monitor \eVdt (above) allows the \sus to immediately input on port \pidVV (possibly carrying a termination request). 
  % (before defaulting to \eIden).
	%
	% On the other hand, 
  At the same time, the branch prefixed by $\trnsInIdS$ permits the \sus to input the first request via any port $\dvV{\,\neq\,}\pidVV$, subsequently blocking inputs on the same port \dvV (without deterring inputs on other ports) via the input branch $\prf{\actSN{ \actIn{\dbVA}{\dbVdc}}{\dvVA\neq\dvV}}{\eIden}$.
	In conjunction to this branch, \eVdt defines  another branch $\prf{\trnsOutIdS}{\eVdtP}$ to allow outputs on the port bound to variable \dvV.
	The continuation monitor \eVdtP then defines the suppression branch $\prf{\trnsOutSupS}{\eVd}$ by which it disables any \emph{redundant} response that is output following the first one.
	Since it also defines branches $\prf{\trnsLogIdSFull}{\rV}$ and $\defmonE$, it does not affect log events or further inputs that occur immediately after the first response.
	%
	
	%
	%By defining $\eVclsdef$ at every branch level, it also refrains from blocking termination inputs \actIn{\pidVV}{\cls}, from being passed on to the \sus. 
	%	
	%To block invalid requests, it adopts the same approach as that of \eVd, but upon detecting a redundant response, \eVdt suppresses it and then reduces to \eVd (see underlined branch).
	%
	%For instance, when 
	%
	When instrumented with system \pVbi from \Cref{ex:shml-formula-bi}, \eVdt allows the composite system to perform the first input but then blocks the second one, permitting  only input requests on channel $b$, \eg $$\eBI{\eVdt}{\pVbi}\tra{\actIn{\pidV}{\vV}}\cdot\tra{\actIn{\pidVV}{\cls}}\eBI{\eIden}{\nil}.$$
	It also disables the first redundant response of system \pVbo while transitioning to \eVd, which proceeds to suppress every subsequent output (including log actions) while blocking every other port except \pidVV, \ie 
  $$\eBI{\eVdt}{\pVbo}\traS{\actIn{\pidV}{\vV}}\cdot\wtraS{\actOut{\pidV}{\vVV}}\cdot\traS{\actt} 
  \eBI{\eVd}{\prf{{\actLoggTuple{\vV}{\vVV}}}\pVg}
  \traS{\actt} 
  \eBI{\eVd}{\pVg}
  \ntraS{\actIn{\pidV}{\vV}}$$ (for every port \pidV where $\pidV{\neq}\pidVV$ and any value \vV).
	Rule \rtit{iDef} allows it to default when handling unspecified outputs, \eg for system $\prf{{\actLoggTuple{\vV}{\vVV}}}{\pVg}$ 
  % although $\eVdt{\ntraS{\actLoggTuple{\vV}{\vVV}}}$, 
  the composite system can still perform 
   the logging output, \ie 
  $$\eBI{\eVdt}{\prf{{\actLoggTuple{\vV}{\vVV}}}{\pVg}}\tra{{\actLoggTuple{\vV}{\vVV}}}\eBI{\eIden}{\pVg}.$$
	
  \noindent Consider one further monitor, defined below:
	\begin{align*}
    \eVdet &\! \defeq\! \rec{\rV}
    {
      (\ch{			
          \prf{\trnsInIdS\,}\eVdetP			
      }
      {
        \eVclsdef
      }
      )
    }
    \\
    \eVdetP &\!\defeq\!
    \rec{\rVV_1}{
      \ch{ 
        \ch{
          \underline{\prf{\actSN{\actdot}{\actIn{\dvV}{\vVdef}}}\rVV_1}
        }{
          \prf{ \actSet{\actOut{\dvV}{\dbVVB}}}{\eVdetPP} 
        }
      }
      {
         \prf{\actSN{ \actIn{\dbVA}{\dbVdc}}{\dvVA\neq\dvV}}{\eIden}
      } 
    }
    \\
    \eVdetPP &\!\defeq\! \rec{\rVV_2}{\bigl(\ch{\underline{
      \prf{ \actSTN{\actOut{\dvV}{\dbVdc}}{\dvV \neq \pidVV}{\!\actdot} }{\rVV_2}}\!}{\!\ch{\prf{\trnsLogIdSFull}{\rV}\!}{\!\defmonE}} \bigr)}
    \end{align*}
	Monitor \eVdet (above) behaves similarly to \eVdt but instead 
  % of reducing to \eVd after suppressing the first redundant response, it
   employs a loop of suppressions (underlined in $\eVdetPP$) to disable further responses until a log or termination input is made.
	%
	%  As opposed to \eVdt, 
  % monitor \eVdet only attempts to disable further responses via the suppression loop of \eVdetPP after disabling the redundant response of \pVbo, thereby allowing 
  When composed with \pVbo, it permits the 
  % subsequent 
  log action to go through:
  % , as follows:
	%
	\begin{align*}
	\eBI{\eVdet}{\pVbo}\traS{\actInV}\cdot\wtraS{\actOut{\pidV}{\vVV}}\cdot\traS{\actt}\eBI{\eVdetPP}{\prf{{\actLoggTuple{\vV}{\vVV}}}\pVbo}\traS{{\actLoggTuple{\vV}{\vVV}}} \eBI{\eVdet}{\pVbo}. 
	\end{align*}
	\eVdet also defines a branch prefixed by the insertion transformation $\actSN{\actdot}{\actIn{\dvV}{\vVdef}}$ (underlined in $\eVdetP$) where $\vVdef$ is a default input domain value.
	This permits the instrumentation to silently unblock the \sus when this is waiting for a request following an unanswered one.
	In fact, when instrumented with \pVbi, \eVdet not only forbids invalid input requests, but it also (internally) unblocks \pVbi by supplying the required input via the added insertion branch.
	This allows the composite system to proceed, as shown below (where $\pVg'{\,\defeq\,}\prf{\assign{\dvVV}{\ans(\vVdef)}}\prf{\actOut{\pidV}{\dvVV}}\prf{{\actLoggTuple{\vVdef}{\dvVV}}}{\pVg}$):
	\begin{align*}
	\eBI{\eVdet}{\pVbi}&\traS{\actInV}\eBI{\rec{\rVV}{(\ch{
				\ch{\prf{\actSN{\actdot}{\actIn{\pidV}{\vVdef}}}{\rVV}}{\prf{\actSet{\actOut{\pidV}{\dbVVB}}}{\eVdetPP}}
				}{\eVclsdef}
		)}}{\pVg}\\
	&\traS{\;\;\actt\;\;}\eBI{\rec{\rVV}{(\ch{ 
				\ch{\prf{\actSN{\actdot}{\actIn{\pidV}{\vVdef}}}{\rVV}}{\prf{\actSet{\actOut{\pidV}{\dbVVB}}}{\eVdetPP}}
				}{\eVclsdef}
		)}}{\pVg'} \\
	&\wtraS{\actOut{\pidV}{\ans(\vVdef)}.{\actLoggTuple{\vVdef}{\dvVV}}}\eBI{\eVdet}{\pVg} 
  & \hfill\qquad\qquad \tag*{\exqed}
	\end{align*}
	%	
	% where $\pVg'{\,\defeq\,}\prf{\assign{\dvVV}{\ans(\vVdef)}}\prf{\actOut{\pidV}{\dvVV}}\prf{{\actLoggTuple{\vVdef}{\dvVV}}}{\pVg}$. \qed
\end{exa}

Although in this paper we mainly focus on action disabling monitors, using our model one can also define action enabling and adaptation monitors.

\begin{exa} \label[example]{ex:other-transducers}
	Consider now the transducers \eVe and \eVa below:
	\begin{align*}
	\eVe & \defeq \prf{
		\actSTD{\patInV}{\dvV{\neq}\pidVV}
	}
	{
		\prf{
			\actSN{\actdot}{\actOut{\dvV}{\ans(\dvVV)}}
		}{\prf{
				\actSN{\actdot}{{\actLoggTuple{\dvVV}{\ans(\dvVV)}}}
			}{\eIden}}
	}
	\\
	\eVa &\defeq \rec{\rV}{
		\ch{
			\prf{\actSN{\actIn{\pidVV}{\dbVV}}{\actIn{\pidV}{\dvVV}}}{\rV}  
		}{
%			\ch{
%				\prf{\actSTN{\actIn{\pidVV}{\dbVV}}{\dvVV{\,=\,}\cls}{\actIn{\pidVV}{\dvVV}}}{\rV}  
%			}{
				\prf{\actSN{\actOut{\pidV}{\dbVV}}{\actOut{\pidVV}{\dvVV}}}{\rV}
%			}
		}
	}. %\cmt{UPDATE EXPLANATION!!!}
	\end{align*}
	Once instrumented, \eVe first uses a suppression to enable 
	an input on any port $\dvV{\,\neq\,}\pidVV$ (but then gets discarded). 
	It then automates a response by inserting an answer followed by a log action. % that are output on ports \pidVVV and \pidVV respectively. 
	Concretely, when composed with the systems $\pVV{\,\in\,}\set{\pVbo,\pVbi}$ from \Cref{ex:shml-formula-bi} (restated in \Cref{ex:transducers}), the execution of the composite system can only start as follows, for some channel name $\pidVVV \neq \pidVV$, values $\vV$ and $\vVV{\,=\,}\ans(\vV)$:
	\begin{align*}
	\eBI{\eVe}{\pVV}
	\tra{\actIn{\pidVVV}{\vV}}
	\eBI{\prf{\actSN{\actdot}{\actOut{\pidVVV}{\vVV}}}{\prf{\actSN{\actdot}{{\actLoggTuple{\vV}{\vVV}}}}{\eIden}}}{\pVV}
	\wtra{\actOut{\pidVVV}{\vVV}}
	\eBI{\prf{\actSN{\actdot}{{\actLoggTuple{\vV}{\vVV}}}}{\eIden}}{\pVV}
	\tra{{\actLoggTuple{\vV}{\vVV}}}
	\eBI{\eIden}{\pVV}.
	\end{align*}
	By contrast, \eVa uses action adaptation to redirect the inputs and outputs from the \sus through port $\pidVV$:
	%
	%Concretely, when the composite system inputs a value on port \pidVV that is not termination request, \eg \actIn{\pidVV}{\req}, this gets redirected to the \sus on port \pidV through the monitor's replacement transformation, and every output gets intercepted and rerouted to the environment through port \pidVV.
	%
	it allows the composite system to exclusively input values on port \pidVV forwarding them to the \sus on port \pidV, and dually allowing outputs from the \sus on port \pidV to reroute them to port \pidVV.
	\nw{As a result, from an external viewpoint, the resulting composite system can only be seen to communicate on port \pidVV with its environment.}
	For instance, for the systems \pVbi and \pVbo restated earlier, we can observe the following behaviour:%we have that:
	\begin{align*}
  % \Eg  
	% $
  & \eBI{\eVa}{\pVbi}
	\tra{\actIn{\pidVV}{\vVA}}
	\eBI{\eVa}{\pVg}
	\wtra{\actIn{\pidVV}{\vVB}.\actOut{\pidVV}{\vVVB}.{\actLoggTuple{\vVB}{\vVVB}}}
	\eBI{\eVa}{\pVg} %\ldots 
	% $ 
  \quad \text{and} 
  \\
  % \quad  
	% $
  &\eBI{\eVa}{\pVbo} 
	\wtra{\actIn{\pidVV}{\vVA}.\actOut{\pidVV}{\vVVA}.{\actLoggTuple{\vVA}{\vVVA}}}
	\eBI{\eVa}{\pVbo}. 
  % $
  \tag*{\exqed}
\end{align*}
\end{exa}

\section{Enforcement}
\label{sec:enforcement}
%\begin{itemize}
%	\item Define the classical definitions of enforcement ie based on soundness and transparency.
%	\item Argue + give examples showing that transparency (and the more classical and weaker, trace transparency) is not strict enough.
%	\item Introduce eventual transparency.
%	\item Give examples.
%	\item Introduce the notion of optimality as means to assess how a monitor modifies invalid system behaviour.
%\end{itemize}

We are concerned with extending the enforceability result obtained in prior work~\cite{Cassar2018Concur} to the extended setting of bidirectional enforcement.
The \emph{enforceability} of a logic rests on the relationship between the semantic behaviour specified by the logic on the one hand, and the ability of the operational mechanism (that of \Cref{sec:model} in this case) to enforce the specified behaviour on the other.
This is captured by the predicate ``(monitor) \eV \emph{adequately enforces} (property) \hV'' in \Cref{def:enforceability} below.
In fact, the definitions of formula and logic enforceability in \Cref{def:enforceability} are parametric with respect to the precise meaning of such a predicate. 
In what follows, we will explore the design space for formalising this predicate.

\begin{defi}[Enforceability~\cite{Cassar2018Concur}] \label[definition]{def:enforceability}
	A formula \hV is \emph{enforceable} iff there \emph{exists} a transducer \eV such that \eV \emph{adequately enforces} \hV.
	A logic \LSet is enforceable iff \emph{every} formula $\hV{\,\in\,}\LSet$ is \emph{enforceable}. \exqed
\end{defi}

% \Cref{def:enforceability} relies on the meaning of ``\eV \emph{adequately enforces} \hV''.
% 
Since we have limited control over the \sus that a monitor is composed with,
%  adequate enforcing 
``\eV \emph{adequately enforces} \hV'' should hold for \emph{any} (instrumentable) system. %
%In particular, 
In \cite{Cassar2018Concur} we stipulate that any notion of adequate enforcement should at least entail soundness. 
% \ie if the property of interest \hV is \emph{satisfiable}, \ie $\isSatF$, then the composite system, $\eI{\eV}{\pV}$, should satisfy \hV for \emph{every} system state \pV.

\begin{defi}[Sound Enforcement~\cite{Cassar2018Concur}] \label[definition]{def:senf} Monitor \eV \emph{soundly enforces} a \nw{formula \hV, denoted as \senfdef{\eV}{\hV}, iff, whenever \hV is satisfiable, \ie $\hSemS{\hV} \neq \emptyset$, then  for every state $\pV{\,\in\,}\Sys$, it is the case that $\eI{\eV}{\pV}{\,\in\,}\hSemS{\hV}$.} \exqed
\end{defi}

\begin{exa}\label[example]{ex:sound-enf} 
	Although showing that a monitor soundly enforces a formula should consider \emph{all} systems, we give an intuition based on \pVg, \pVbo, \pVbi for formula $\hV_1$ from \Cref{ex:shml-formula-bi} (restated below)  where $\pVg{\,\in\,}\hSemS{\hV_1}$ (hence $\isSat{\hV_1}$) and $\pVbo,\pVbi{\,\notin\,}\hSemS{\hV_1}$.
	\begin{align*}
	\hV_1&\defeq\hVdef \\
	\hV'_1&\defeq\hVdefAP
	\end{align*}	
	Recall the transducers \eVe, \eVa, \eVd, \eVdt and \eVdet from \Cref{ex:transducers,ex:other-transducers}, restated below:
  \begin{align*}
    \eVe & \defeq \prf{
      \actSTD{\patInV}{\dvV{\neq}\pidVV}
    }
    {
      \prf{
        \actSN{\actdot}{\actOut{\dvV}{\ans(\dvVV)}}
      }{\prf{
          \actSN{\actdot}{{\actLoggTuple{\dvVV}{\ans(\dvVV)}}}
        }{\eIden}}
    }
    \\
    \eVa &\defeq \rec{\rV}{
      \ch{
        \prf{\actSN{\actIn{\pidVV}{\dbVV}}{\actIn{\pidV}{\dvVV}}}{\rV}  
      }{
  %			\ch{
  %				\prf{\actSTN{\actIn{\pidVV}{\dbVV}}{\dvVV{\,=\,}\cls}{\actIn{\pidVV}{\dvVV}}}{\rV}  
  %			}{
          \prf{\actSN{\actOut{\pidV}{\dbVV}}{\actOut{\pidVV}{\dvVV}}}{\rV}
  %			}
      }
    }\\
    \eVd & \defeq \rec{\rVV}{\ch{\prf{\actSet{\actIn{\pidVV}{\dbVdc}}}{\rVV}}{\prf{\actSN{\actOut{\dbVdc}{\dbVdc}}{\actdot}}{\rVV}}}
    \\
    \eVdt & \defeq 
	\rec{\rV}{(
		\ch{
			\prf{\trnsInIdS\,}
			(\ch{ \prf{\actSN{ \actIn{\dbVA}{\dbVdc}}{\dvVA\neq\dvV}}{\eIden}\!}
			{\!\prf{\trnsOutIdS}{
					\eVdtP
				}
			}
			)\! 
		}
		{
			\eVclsdef
		}
	)} 
	\\
  \eVdet &\! \defeq\! \rec{\rV}
    {
      (\ch{			
          \prf{\trnsInIdS\,}\eVdetP			
      }
      {
        \eVclsdef
      }
      )
    }
  \intertext{where}  
	\eVdtP & \defeq  \ch{\prf{\trnsOutSupS}{\eVd}}{\ch{\defmonE}{\prf{\trnsLogIdSFull}{\rV}\!}} 
  \\ 	
  \eVdetP &\defeq
  \rec{\rVV_1}{
    \ch{ 
      \ch{
        \underline{\prf{\actSN{\actdot}{\actIn{\dvV}{\vVdef}}}\rVV_1}
      }{
        \prf{ \actSet{\actOut{\dvV}{\dbVVB}}}{\eVdetPP} 
      }
    }
    {
       \prf{\actSN{ \actIn{\dbVA}{\dbVdc}}{\dvVA\neq\dvV}}{\eIden}
    } 
  }
  \\
  \eVdetPP &\defeq \rec{\rVV_2}{\bigl(\ch{\underline{
    \prf{ \actSTN{\actOut{\dvV}{\dbVdc}}{\dvV \neq \pidVV}{\!\actdot} }{\rVV_2}}\!}{\!\ch{\prf{\trnsLogIdSFull}{\rV}\!}{\!\defmonE}} \bigr)}
    \end{align*}
	When assessing their soundness in relation to the property $\hV_1$, we have that:
	\begin{itemize}
		\item \eVe is \emph{unsound} for $\hV_1$. When composed with \pVbo, the resulting monitored system produces two consecutive output replies (namely those underlined in the trace $\tr_{\textsf{e}}^{1}$ below), thus meaning that the composite system violates the property in question, \ie $\eBI{\eVe}{\pVbo}{\notin}\hSemS{\hV_1}$.  More concretely, we have  
    $$\eBI{\eVe}{\pVbo}\wtraS{\tr_{\textsf{e}}^{1}}\eBI{\eIden}{\pVbo} \quad\text{ where }\tr_{\textsf{e}}^{1}{\defeq}\actIn{\pidVVV}{\vVA}.\actOut{\pidVVV}{\ans(\vVA)}.\actLoggTuple{\vVA}{\ans(\vVA)}.\actInPidVvVB.\underline{\actOut{\pidV}{\vVVB}.\actOut{\pidV}{\vVVB}}.$$ 
    % as it allows invalid behaviour such as $\eBI{\eVe}{\pVbo}\wtraS{\tr_{\textsf{e}}^{1}}\eBI{\eIden}{\pVbo}$ where $\tr_{\textsf{e}}^{1}{\defeq}\actIn{\pidVVV}{\vVA}.\actOut{\pidVVV}{\ans(\vVA)}.\actLoggTuple{\vVA}{\ans(\vVA)}.\actInPidVvVB.\underline{\actOut{\pidV}{\vVVB}.\actOut{\pidV}{\vVVB}}$.
		%
		% This shows that the composite system $\eBI{\eVe}{\pVbo}$ can still make two consecutive output replies (underlined), and so $\eBI{\eVe}{\pVbo}{\notin}\hSemS{\hV_1}$.
		%
		Similarly, the system \pVbi instrumented with the transducer \eVe also violates property $\hV_1$, \ie $\eBI{\eVe}{\pVbi}\notin\hSemS{\hV_1}$, since the $\eBI{\eVe}{\pVbi}$ executes the erroneous trace with two consecutive inputs on port \pidV (underlined), $\actIn{\pidVVV}{\vVA}.\actOut{\pidVVV}{\ans(\vVA)}.\actLoggTuple{\vVA}{\ans(\vVA)}.\underline{\actIn{\pidV}{\vVVB}.\actIn{\pidV}{\vVVC}}$.
		This demonstrates that $\eBI{\eVe}{\pVbi}$ can still input two consecutive requests on port \pidV (underlined).
		Either one of these \nw{counterexamples} \emph{disproves} $\senfdef{\eVe}{\hV_1}$.
		\item Monitor \eVa turns out to be \emph{sound} for $\hV_1$ because once instrumented, the resulting composite system is adapted to only interact on port \pidVV.
    % , and so its actions are not of concern to $\hV_1$.
		%
		% As \eVa applies this enforcement strategy to any \sus, we can safely conclude that  $\senfdef{\eVa}{\hV_1}$, 
    In fact, we have $\{ \eI{\eVa}{\pVg},\eI{\eVa}{\pVbo},\eI{\eVa}{\pVbi}\} \subseteq \hSemS{\hV_1}$.
		The other monitors \eVd, \eVdt and \eVdet are also \emph{sound} for $\hV_1$. 
		Whereas, monitor \eVd prevents the violation of $\hV_1$ by also blocking all input ports except \pidVV,  the transducers \eVdt and \eVdet achieve the same goal by disabling the invalid consecutive requests and answers that occur on any port except \pidVV.
		%
		%From this intuitive reasoning, and since $\eBI{\eV_x}{\pVg},\eBI{\eV_x}{\pVbo},\eBI{\eV_x}{\pVbi}{\,\in\,}\hSemS{\hV_1}$ where $\eV_x{\,\in\,}\set{\eVd,\eVdt,\eVdet}$ we conclude that $\senfdef{\eV_x}{\hV_1}$.
		\exqed
	\end{itemize}
\end{exa}

By itself, sound enforcement is a weak criterion because it does not regulate the \emph{extent} to which enforcement is applied.
More specifically, although \eVd from \Cref{ex:transducers} is sound,  it needlessly modifies the behaviour of $\pVg$ even though $\pVg$ satisfies $\hV_1$: by blocking the initial input of \pVg, \eVd causes it to block indefinitely. 
The requirement that a monitor should not modify the behaviour of a system that satisfies the property being enforced can be formalised using a transparency criterion.
%
% For 
% % the enforcement of 
% properties describing computation trees such as ours 
% % (as opposed to those describing single executions~\cite{Ligatti2005,Ligatti2009,Sakarovitch:2009:Transducers,Bielova2011,Khoury2012,Konighofer2017,Falcone2012}), 
% a transparency criterion is also required, defined  in terms of some adequate notion of system equivalence. 
% (\eg bisimulation equivalence, $\sim$~\cite{Sangiorgi2011Bisim}). 
%
% dictating  that, whenever a system $\pV$ already satisfies the property \hV, the assigned monitor \eV should not alter the behaviour of $\pV$.
%
% in addition to soundness, in  we posited that adequate enforcement must also require a degree of \emph{transparency} that safeguards the integrity of well-behaved systems. 
% %
% Transparency thus dictates that, whenever a system $\pV$ already satisfies the property \hV, the assigned monitor \eV should not alter the behaviour of $\pV$.
%
%To put it simply, the behaviour of the composite system should be \emph{bisimilar} to that of the valid \sus.

\begin{defi}[Transparent Enforcement~\cite{Cassar2018Concur}] \label[definition]{def:tenf} A monitor \eV \emph{transparently} enforces a formula \hV, 
  % written as 
  \tenfdef{\eV}{\hV}, iff  for \emph{all} %LTSs $\langle\Sys,\Act\cup\sset{\actt},\rightarrow\rangle$ and 
	% systems  
  $\pV\in\Sys$,  $\pV{\in}\hSemS{\hV}$ implies $\;\eI{\eV}{\pV}\sim\pV$. \qed
\end{defi}

\begin{exa} \label[example]{ex:transparency} As argued earlier, 
  % the counter-example given in relation to 
  \pVg suffices to disprove $\tenfdef{\eVd}{\hV_1}$.
	\nw{Monitor \eVa from \Cref{ex:other-transducers} also breaches \Cref{def:tenf}: although $\pVg \in \hSemS{\hV_1}$, we have  $\eBI{\eVa}{\pVg}\nbisim\pVg$ since for any value $\vV$ and $\vVV$, 
  $\pVg\traS{\actIn{\pidV}{\vV}}
  % \cdot\traS{\actOut{\pidVV}{\vVV}}
  $ but for any value $\vV$ we can \emph{never} have 
  $\eBI{\eVa}{\pVg}\traS{\actIn{\pidV}{\vV}}
  % \cdot\ntraS{\actOut{\pidVV}{\vVV}}
  $.}
	By contrast, monitors \eVdt and \eVdet turn out to satisfy \Cref{def:tenf}, since they only intervene when it becomes apparent that a violation will occur.
	For instance, they only disable inputs on a specific port, as a precaution, following an unanswered request on the same port, and they only disable the redundant responses that are produced after the first response to a request.
	%
	% The universal quantification over all systems makes it difficult to show that $\tenfdef{\eVdt}{\hV_1}$ and $\tenfdef{\eVdet}{\hV_1}$.
	% %
	% However, since both monitors do not modify valid systems such as \pVg, \ie $\pVg{\,\in\,}\hSemS{\hV_1}$ and $\eBI{\eVdt}{\pVg}\bisim\pVg\bisim\eBI{\eVdet}{\pVg}$, and only modify invalid ones, such as \pVbo and \pVbi (see \Cref{ex:sound-enf}), we get a good intuition for why this is the case. 
  \exqed
\end{exa}

It turns out that, 
by some measures, \Cref{def:tenf} is still a relatively weak requirement since it only limits transparency requirements to \emph{well-behaved} systems, \ie those that satisfy the property in question, and disregards enforcement behaviour for systems that violate this property.
For instance, consider monitor \eVdt from \Cref{ex:transducers} (restated in \Cref{ex:sound-enf}) and system \pVbo from \Cref{ex:shml-formula-bi} (restated in \Cref{ex:transducers}). 
%
%Despite not modifying the behaviour of well-behaved systems, such as \pVg (thus adhering to transparency), when instrumented with an invalid system, such as \pVbo, it also disables the log action, along with its redundant response, as a result of reducing into \eVd.
%%
%In fact, after reducing to \eVd, it keeps on disabling actions even after \pVbo reaches a valid point by reducing to state \pVg, \ie $\eI{\eVdt}{\pVbo}\wtraS{\actInV.\actOut{\pidV}{\vVV}}\eI{\eVd}{\pVg}$ and $\pVg{\,\in\,}\hSemS{\afterF{\hV_1}{\actInV.\actOut{\pidV}{\vVV} }}=\hSemS{\hAnd{\hNec{\actOut{\pidV}{\dbVVB}}\hFls\,}{\,\hNec{\actLoggTuple{\vVA}{\vVVA}}\hV_1}}$ but $\eI{\eVd}{\pVg}\nbisim\pVg$ since for every input value \vVB $\pVg\traS{\actInPidVvVB}$ and $\eI{\eVd}{\pVg}\ntraS{\actInPidVvVB}$.
%
At runtime, \pVbo can exhibit the following invalid behaviour: $$\pVbo\wtraS{\tr_1}\prf{\actOut{\pidVV}{\loggTuple{\vV}{\vVV}}}{\pVg} \quad \text{where } {\tr_1}{\,\defeq\,}{\prf{\actIn{\pidV}{\vV}}\prf{\actOut{\pidV}{\vVV}}\actOut{\pidV}{\vVV}} \text{ for some appropriate pair of values } \vV,\vVV.$$
In order to 
% bring the invalid 
rectify this violating behaviour 
% of $\pVbo$ (show in $\tr_1$) in line with our specification 
\wrt formula $\hV_1$, it suffices to use a monitor that disables \emph{one} of the responses in $\tr_1$, \ie $\actOut{\pidV}{\vVV}$.
%
% After correcting $\tr_1$ into ${\tr'_1}{\,\defeq\,}{\prf{\actIn{\pidV}{\vV}}\actOut{\pidV}{\vVV}}$, 
Following this disabling, no further modifications are required since the \sus reaches a state that does not violate the remainder of the formula $\hV_1$, \ie 
 $\prf{\actOut{\pidVV}{\loggTuple{\vV}{\vVV}}}{\pVg}{\in}\hSemS{\afterF{\hV_1}{\tr'_1}}$ where ${\tr'_1}{\,\defeq\,}{\prf{\actIn{\pidV}{\vV}}\actOut{\pidV}{\vVV}}$. 
However, when instrumented with \eVdt, this monitor does not only disable the invalid response, namely $\eI{\eVdt}{\pVbo}\wtraS{\prf{\actIn{\pidV}{\vV}}\prf{\actOut{\pidV}{\vVV}}}\eI{\eVd}{\prf{\actOut{\pidVV}{\loggTuple{\vV}{\vVV}}}{\pVg}}$, but subsequently disables every other action by reaching \eVd, $\eI{\eVd}{\prf{\actOut{\pidVV}{\loggTuple{\vV}{\vVV}}}{\pVg}}\traS{\actt}\eI{\eVd}{\pVg}$.
%
%
%
% It thus makes sense that transparency should also start applying whenever an invalid \sus reaches a valid point while instrumented with the monitor.
% %
% Put differently, if a composite system, $\eI{\eV}{\pV}$ (where $\pV{\notin}\hSemS{\hV}$), reduces to some state $\eI{\eV'}{\pV'}$ over a trace \tr, where $\pV'$ is in agreement with $\hV$ after following \tr (\ie $\pV'{\,\in\,}\hSemS{\afterF{\hV}{\tr}}$), then the behaviour of $\eI{\eV'}{\pV'}$ should be \emph{equivalent} to that of $\pV'$.
%
To this end, we introduce the novel requirement of \emph{eventual transparency}.
\begin{defi}[Eventually Transparent Enforcement] \label[definition]{def:evtenf} Monitor \eV enforces property \hV in an eventually transparent way, 
  % denoted as 
  $\evtenfdef{\eV}{\hV}$, iff for all %LTS $\langle\Sys,\Act\cup\sset{\actt},\rightarrow\rangle$, 
	systems 
  % states 
  $\pV, \pV'$, traces \tr and monitors $\eV'$, $\eI{\eV}{\pV}\wtraS{\tr}\eI{\eV'}{\pV'}$ and $\pV'\in\hSemS{\afterF{\hV}{\tr}}$ imply  $\eI{\eV'}{\pV'}\bisim\pV'$. \exqed
\end{defi}

\begin{exa} \label[example]{ex:evtenf} We have already argued why \eVdt (restated in \Cref{ex:sound-enf}) does not adhere to eventual transparency via the counterexample \pVbo. 
  This is not the case for \eVdet (also restated in \Cref{ex:sound-enf}). 
  % because $\evtenfdef{\eVdet}{\hV_1}$. 
	%
	Although the universal quantification over all systems and traces make it hard to prove this property, we get an intuition of why this is the case from the system \pVbo. 
  More concretely,  when 
  $$\eBI{\eVdet}{\pVbo}\wtraS{\actIn{\pidV}{\vVA}.\actOut{\pidV}{\vVVA}}\cdot\traS{\actt}\eBI{\eVdetPP}{\prf{\actLoggTuple{\vVA}{\vVVA}}\pVg}$$
%	\begin{gather*}		\eBI{\eVdet}{\pVbo}\wtraS{\actIn{\pidV}{\vVA}.\actOut{\pidV}{\vVVA}}\eBI{\eVans}{\ch{\prf{\actOut{\pidV}{\vVVA}}\prf{\actLoggTuple{\vVA}{\vVVA}}\pVg}{\prf{\actLoggTuple{\vVA}{\vVVA}}{\pVbo}}} \traS{\actt}\eBI{\eVans}{\prf{\actLoggTuple{\vVA}{\vVVA}}\pVg} %\traS{\actLoggTuple{\vVA}{\vVVA}}\eBI{\eVdet}{\pVg}
%	\end{gather*}
	we have 
  \begin{align*}
    \prf{\actLoggTuple{\vVA}{\vVVA}}\pVg &\in \hSemS{\afterF{\hV_1}{\actIn{\pidV}{\vVA}.\actOut{\pidV}{\vVVA}}} \\
    \intertext{ since } 
    \afterF{\hV_1}{\actIn{\pidV}{\vVA}.\actOut{\pidV}{\vVVA}} &=  (\hAnd{\hNec{\actSN{\actOut{\dbVC}{\dbVdc}}{\dvVC{=}\pidV}}\hFls\,}
    {\,\hNec{\actSN{\actOut{\dbVD}{\dbVVC}}{\dvVD{=}\pidVV\land\dvVVC{=}\loggTuple{\vV_1}{\vVV_1}}}\hV_1})
  \end{align*}
  and, moreover, $\eI{\eVdetPP}{\prf{\actLoggTuple{\vVA}{\vVVA}}\pVg}\bisim\prf{\actLoggTuple{\vVA}{\vVVA}}\pVg$. \exqed
\end{exa}

\begin{cor} \label[corollary]{cor:evtran-implies-tran}
  For all monitors $\eV\in\Trn$ and properties $\hV\in\SHML$,  $\evtenfdef{\eV}{\hV}$ implies $\tenfdef{\eV}{\hV}$. \qed
\end{cor}

% Since \Cref{def:tenf} (transparency) is just an instance of \Cref{def:evtenf} (eventual transparency) (\ie when \tr is the empty trace \trE), the latter requirement 
Along with \Cref{def:senf} (soundness), \Cref{def:evtenf} (eventual transparency)  makes up our definition for \emph{``\eV (adequately) enforces \hV''}. From \Cref{cor:evtran-implies-tran}, it follows that is definition is stricter than the one given in \cite{Cassar2018Concur}.

\begin{defi}[Adequate Enforcement] \label[definition]{def:enforcement}
	A monitor \eV (adequately) \emph{enforces} property \hV, \nw{denoted as $\enfdef{\eVV}{\hV}$,}  iff it adheres to $(i)$ \emph{soundness}, \Cref{def:senf}, and $(ii)$ \emph{eventual transparency}, \Cref{def:evtenf}. \exqed
\end{defi}

% \begin{corollary}
% 	Since \Cref{def:enforcement} is defined in terms of eventual transparency (\Cref{def:evtenf}) that is stronger than transparency (\Cref{def:tenf}), our new definition for ``\eV enforces \hV'' is \emph{stricter} than the one given in \cite{Cassar2018Concur}. \qed
% \end{corollary}

\section{Synthesising Action Disabling Monitors}
\label{sec:synthesis}
% !TEX root = journal.tex

Although \Cref{def:enforceability} (instantiated with \Cref{def:enforcement}) enables us to rule out erroneous monitors that purport to enforce a property,  the \emph{universal quantifications} over all systems in \Cref{def:senf,def:evtenf} make it difficult to prove that a monitor does indeed enforce a property correctly in a bidirectional setting (disproving, however, is easier).
Establishing that a formula is enforceable,  \Cref{def:enforcement}, involves a further existential quantification over a monitor that enforces it correctly;
put differently, in order to show that a formula is \emph{not} enforceable, amounts to another universal quantification, this time over all possible monitors.
Moreover, establishing the enforceability of a logic entails yet another universal quantification, on all the formulas in the logic. 
In many cases (including ours), the sets of systems, monitors and formulas are infinite.

We address these problems through an \emph{automated synthesis procedure} that produces an enforcement monitor from a safety \recHML formula, expressed in a syntactic fragment of \SHML. 
This fragment, called \SHMLnf, has already been used to establish enforceability results in a uni-directional setting~\cite{Cassar2018Concur} and is the source logic employed by the tool detectEr\footnote{\texttt{https://duncanatt.github.io/detecter/}}\cite{AttardAAFIL21,AcetoAAEFI22,AchilleosEFLX22}  used to verify the correctness of concurrent systems written in Erlang~\cite{AcetoAFI21} and Elixir~\cite{BrunAF21}; it also coincides with \SHML in the regular setting~\cite{AcetoAFIK:jlap20:determ}.
We show that the synthesised monitors are correct, according to \Cref{def:enforcement}.
For a unidirectional setting, it has been shown that monitors that only administer  
% action 
\emph{omissions} are expressive enough to enforce \emph{safety properties}~\cite{Ligatti2005,Falcone2012,VanHulst2017,Cassar2018Concur}.
Analogously, for our bidirectional case, we restrict ourselves to action disabling monitors and show that they can enforce \emph{any} property expressed in terms of this \SHML fragment.

Our synthesis procedure is compositional, meaning that the monitor synthesis of a composite formula is defined in terms of the enforcement monitors generated from its constituent sub-formulas. 
Compositionality simplifies substantially our correctness analysis of the generated monitors (\eg we can use standard inductive proof techniques).
The choice of the logical fragment, \ie \SHMLnf, facilitates this compositional definition. 
An automated procedure to translate an \SHML formula with symbolic actions where the scope of the data binders is limited to the immediate symbolic action condition, into a corresponding \SHMLnf one (with the same semantic meaning) is given in \cite{Cassar2018Concur,AcetoAFIK:jlap20:determ}.
%
% \SHMLnf with data as also been used as the source language for the implementation of a runtime verification tool called \textsf{detectEr}\footnote{\texttt{https://duncanatt.github.io/detecter/}}\cite{AttardAAFIL21,AcetoAAEFI22,AchilleosEFLX22} which been used to verify the correctness of concurrent systems written in Erlang~\cite{AcetoAFI21} and Elixir~\cite{BrunAF21}. 

\begin{defi}[\SHML normal form] \label[definition]{def:shmlnf}
	The set of normalised \SHML formulas is generated by the following grammar (where\footnote{\nw{Recall that from \Cref{fig:recHML}, $I$ always denotes a \emph{finite} set of indices which is crucial for a synthesis process to terminate.}} $|\IndSet| \geq 1$):
	\begin{align*}
	\hV,\hVV
  \in\SHMLnf\;\bnfdef\; \hTru
	\bnfseppp\hFls
	\bnfseppp \hBigAndU{i\in\IndSet}\hNec{\pateI,\bVI}{\hVI}
	\bnfseppp\hVarX
	\bnfseppp\hMaxXF\;.
	\end{align*}
	In addition, \SHMLnf formulas are required to satisfy the following conditions:
	\begin{enumerate}
		\item Every branch in $\hBigAndU{i\in\IndSet}\hNec{\pateI,\bVI}{\hVI}$,
		must be \emph{disjoint}, 
    % $\bigdistinct{i{\in}\IndSet}\actSN{\pateI}{\bVI}$, which entails that
    \ie for every $i,j\in\IndSet$, $i\neq j$ implies $\hSemS{\actSN{\pateI}{\bVI}}{\,\cap\,}\hSemS{\actSN{\pateJ}{\bVJ}}=\emptyset$. %$\bigintersectU{i\in\IndSet}{\actSN{\pateI}{\bVI}}=\emptyset$.
		\item For every \hMaxXF we have $\hVarX \in \fv{\hV}$.  \exqed
		%\item Every logical variable is \emph{guarded} by a modal necessity. \qed
	\end{enumerate}
\end{defi}

In a (closed) \SHMLnf formula, the basic terms \hTru and \hFls can never appear unguarded unless they are at the top level (\eg we can never have $\hAnd{\hV}{\hFls}$ or $\hMax{\hVarX_{0}}{\ldots\hMax{\hVarX_{n}}{\hFls}}$). 
%
%The symbolic actions defined in modal operators are assumed to be unshortened, and so we forgo the shorthand notation for symbolic actions in \SHMLnf.
%
Modal operators are combined with conjunctions into one construct $\hBigAndU{i\in\IndSet}\hNec{\pateI,\bVI}{\hVI}$ that is written as $\hNec{\pate_0,\bV_0}{\hV_{0}}\hAndS\ldots\hAndS\hNec{\pate_n,\bV_n}{\hV_{n}}$ when $\IndSet=\Set{0,\ldots,n}$ and simply as $\hNec{\pate_0,\bV_0}{\hV}$ when $\len{\IndSet}=1$. 
The conjunct modal guards must also be \emph{disjoint} so that \emph{at most one} necessity guard can satisfy any particular visible action.
Along with these restrictions, we still assume that \SHMLnf fixpoint variables are guarded, and that for every $\actSN{\patInV}{\bV}$, $\dvVV{\,\notin\,}\fv{\bV}$.
%
%To further simplify our synthesis we also forgo the shorthand notation for symbolic actions and work \wrt unshortened formulas.
%

\begin{exa} The  formula $\hV_3$ 
  % defines a recursive property that is satisfied by all systems in which every input on port \pidV followed by an output with payload 4 is preceded by an output on port \pidV with payload 3.
  %
  % \nw{defines a recursive property stating that, following an input on port \pidV (carrying any value), prohibits that the system outputs a value of 4 (on any port), unless the output is made on port \pidV with a value that is not equal to 3 (in which cases, it recurses).}
  %
  \nw{defines a recursive property stating that an input on port \pidV (carrying any value) cannot be followed by an output with value of 4 (on any port), and that this continues to hold if the subsequent output is made on port \pidV with a value that is not equal to 3 (in which cases, the formula recurses)}
	\begin{align*}
		\hV_3 \defeq \hMaxX{\hNec{\actSN{\actIn{\dbVA}{\dbVVA}}{\dvVA{=}\pidV}}
			\begin{xbrackets}{c}
					\hNec{\actSN{\actOut{\dbVB}{\dbVVB}}{\dvVB{=}\pidV\land\dvVVB{\neq}3}}\hVarX \\[0.5em]
					\hAndS \; \hNec{\actSN{\actOut{\dbVC}{\dbVVC}}{\dvVVC{=}4}}\hFls
			\end{xbrackets} 
		}
	\end{align*}
	$\hV_3$ is not an $\SHMLnf$ formula since its conjunction is not disjoint (\eg the action \actOut{\pidV}{4} satisfies both branches). 
  % \ie $\hSemS{\actSN{\actOut{\dbVB}{\dbVVB}}{\dvVB{=}\pidV\land\dvVVB{\neq}3}}{\cap}\hSemS{\actSN{\actOut{\dbVC}{\dbVVC}}{\dvVVC{=}4}}{=}\set{\actOut{\pidV}{4}}$.
	%
	% Using the normalisation procedure of \cite{Cassar2018Concur}, 
  Still, we can reformulate $\hV_3$ as $\hV'_3\in\SHMLnf$:
	\begin{align*}
	\hV'_3 \defeq \hMaxX{\hNec{\actSN{\actIn{\dbVA}{\dbVVA}}{\dvVA{=}\pidV}}
		\begin{xbrackets}{c}
		\hNec{\actSN{\actOut{\dbVD}{\dbVVD}}{\dvVD{=}\pidV\land%\dvVVD{\neq}3\land
				\dvVVD{\neq}4 \land \dvVVD{\neq}3}}\hVarX\\[0.5em]
		 \hAndS \; \hNec{\actSN{\actOut{\dbVD}{\dbVVD}}{\dvVD{=}\pidV\land%\dvVVD{\neq}3\land
		 		\dvVVD{=}4}}\hFls
		\end{xbrackets} 
	}
	\end{align*}
	where $\dvVD$ and $\dvVVD$ are fresh variables. \exqed
\end{exa}

Our monitor synthesis function in \Cref{def:synthesis-bi} converts an \SHMLnf formula \hV into a transducer \eV.
This conversion also requires information regarding the input ports employed by the \sus, as this is used to add the necessary insertion branches to silently unblock the \sus at runtime; \nw{this prevents the monitor from unnecessarily blocking the resulting composite system.}
% the runtime progression of 
%  the resulting composite system unnecessarily.
.
The synthesis function must therefore be supplied with this information in the form of a \emph{finite} set of input ports $\prtSet{\,\subset\,}\Port$, which then relays this information to the resulting monitor. 
It also assumes a default value $\vVdef$ for the payload data domain.

\begin{defi} \label[definition]{def:synthesis-bi}
	% [Monitor Synthesis]
	The synthesis function $\eSembiS{-}{\,:\,}\SHMLnf{\,\times\,}\psetfin{\Port}{\,\to\,}\Trn$ is defined inductively as:
	\begin{align*}
		\eSembiP{\hVarX} & \defeq \rV 
    \\
		% \and
		\eSembiP{\hTru} &\defeq \eIden 
    \\ 
    \eSembiP{\hFls} &\defeq \eIden
    \\
		% \and
		\eSembiP{\hMaxX{\hV}} &\defeq \rec{\rV}{\eSembiP{\hV}} 
    \\
    % \and
		\eSembiP{\hV{=}\hBigAndD{i{\,\in\,}\IndSet}\hNec{\actSN{\pateI}{\bVI}}{\hVI}}
		&\defeq
		\rec{\rVV}{
			\ch{
				\begin{xbrackets}{c}
					\chBigI
					\begin{xbrace}{ll}
						\dis{\pateI}{\bVI}{\rVV}{\prtSet} & \quad \text{if } \hV_i{=}\hFls\\
						\prf{\actSIDs{\pateI}{\bVI}}{\eSembiP{\hV_i}} & \quad \text{otherwise}
					\end{xbrace}
				\end{xbrackets}
			}
			{
				\defmon{\hV}
			}
		} 
		%\tag*{
		%		\qed
		%	}
	\end{align*} 
	where $\dis{\pate}{\bV}{\eV}{\prtSet} \defeq 
	\begin{xbrace}{ll}
	\eDrp{\pate}{\bV}{\eV} & \quad \text{if } \pate=\actOut{\dbV}{\dbVV}\\[0.5em]
	\chBig{\pidVV{\,\in\,}\prtSet} \eIns{\bV\Sub{\pidVV}{\dvV}}{\actIn{\pidVV}{\vVdef}}{\eV} & \quad \text{if } \pate=\actIn{\dbV}{\dbVV} %5\text{ where \vVdef is} \\[-0.5em] & \quad \text{a default domain value.} 
	\end{xbrace}$
	\\
  and 
  $$\defmon{\hBigAndD{i{\,\in\,}\IndSet\!\!}\!\hNec{\actSN{\actIn{\bVar{\dvV_i}}{\bVar{\dvVV_i}}}{\bVI}}{\hVI}\hAndS\hVV} \defeq 
	\begin{xbrace}{ll} \\[1mm]
		\defmonE & \!\!\whent \IndSet{=}\emptyset \\[0.5em]
		\defmonFull & \otherwiset
	\end{xbrace}$$
	where $\hVV$ has no conjuncts starting with an input modality, variables $\dvV$ and $\dvVV$ are fresh, and $\vVdef$ is a default value. \exqed
\end{defi}
The definition above assumes a bijective mapping between formula variables and monitor recursion variables. 
% and converts logical variables \hVarX accordingly,  whereas maximal fixpoints, $\hMaxXF$, are converted into the corresponding recursive monitor.
% %
% The synthesis also converts truth, \hTru, and falsehood, \hFls, formulas into the identity monitor \eIden.
%
Normalised conjunctions, $\hBigAndU{i{\,\in\,}\IndSet}\hNec{\pateI,\bVI}{\hVI}$, are synthesised as a \emph{recursive summation} of monitors, \ie $\rec{\rVV}{\chBigU{i\in\IndSet}\eV_i}$, where $\rVV$ is fresh, and every branch $\eV_i$ can be one of the following:
\begin{enumerate}[$(i)$]
	\item when $\eV_i$ is derived from a branch of the form $\hNec{\pateI,\bVI}\hVI$ where $\hVI{\,\neq\,}\hFls$, the synthesis produces a monitor with the \emph{identity transformation} prefix, $\actSN{\pateI}{\bVI}$, followed by the monitor synthesised from the continuation $\hVI$,
	\ie  \eSembiP{\hV_i};
  % $\hNec{\actSN{\pateI}{\bVI}}\hVI$ is synthesised as $\prf{\actSN{\pateI}{\bVI}}{\eSembiP{\hV_i}}$;
	%
	\item when $\eV_i$ is derived from a violating branch of the form $\hNec{\pateI,\bVI}\hFls$, the synthesis produces an \emph{action disabling transformation} via $\dis{\pateI}{\bVI}{\rVV}{\prtSet}$. 
	%
	%\ie a branch of the form $\hNec{\actSN{\pateI}{\bVI}}\hFls$ is translated into $\prf{\dis{\pateI}{\bVI}{\rVV}{\prtSet}}{\rVV}$.
\end{enumerate} 
Specifically, in clause $(ii)$, the \disS function produces either a \emph{suppression transformation}, \actSTD{\pateI}{\bVI}, when \pateI is an \emph{output} pattern, $\actOut{(\dvV_i)}{(\dvVV_i)}$, or a \emph{summation of insertions}, $\chBigU{\pidVV\in\prtSet}\eIns{\bVI\Sub{\pidVV}{\dvV_i}}{\actIn{\pidVV}{\vVdef}}{\eV_i}$, when \pateI is an \emph{input} pattern, $\actIn{(\dvV_i)}{(\dvVV_i)}$. 
The former signifies that the monitor must react to and suppress every matching (invalid) system output thus stopping it from reaching the environment. 
By not synthesising monitor branches that react to the erroneous input, the latter allows the monitor to hide the input synchronisations from the environment. %thereby preventing the composite system from procuring the erroneous input.
At the same time, the synthesised insertion branches insert a default domain value \vVdef on every port $\pidVV{\,\in\,}\prtSet$ whenever the branch condition $\bVI\Sub{\pidVV}{\dvV_i}$ evaluates to true at runtime.
This stops the monitor from blocking  
the runtime progression of 
 the resulting composite system unnecessarily. 
This blocking mechanism can, however, block \emph{unspecified} inputs, \ie those that do not satisfy any modal necessity in the normalised conjunction.
This is undesirable since the unspecified actions do not contribute towards a safety violation and, 
% on the contrary, 
instead, lead to its trivial satisfaction.
To prevent this, the \emph{default monitor} $\defmon{\hV}$ is also added to the resulting summation.
Concretely, the \defmonS function produces a \emph{catch-all} identity monitor that forwards an input to the \sus whenever it satisfies the negation of \emph{all} the conditions associated with modal necessities for input patterns in the normalised conjunction.
%
% Put differently, the default monitor allows inputs to reach the system whenever they satisfy the condition $\hBigAnd{i{\in}\IndSet}\lnot\bVI$.
%
This condition is constructed for a normalised conjunction of the form $\hBigAndU{i\in\IndSet}\hNec{\actSN{\actIn{(\dvV_i)}{(\dvVV_i)}}{\bVI}}{\hVI}\,\hAndS\,\hVV$ (assuming that $\hVV$ does not include further input modalities). 
Otherwise, 
if none of the conjunct modalities define an input pattern, every input is allowed, \ie the default monitor becomes $\defmonE$,
%
% Upon matching a system action, the default monitor 
which transitions to \eIden after forwarding the input to the \sus.

\begin{exa} \label[example]{ex:synthesis-bi} %\cmt{highlight the issue with \prtSet, \ie the more info it has the better}
  % Consider the following unshortened 
  Recall (the full version of) formula $\hV_1$ from \Cref{ex:shml-formula-bi}.
  \begin{align*}
    \hV_1&\defeq\hVdefF \\
    \hV'_1&\defeq\hVdefAPF
  \end{align*}	
  For any arbitrary set of ports \prtSet, the synthesis of \Cref{def:synthesis-bi} produces the following monitor.
  \begin{align*}
  \eV_{\hV_1} & \!\!\defeq \!\!
  \rec{\rV}
  {
    \rec{\rVVV}
    {
      \ch{
        (
          \prf{\trnsInIdS\,}
          \rec{\rVV_1}
          {
            \eV'_{\hV_1}
          }
        \!)
      }
      {
        \prf{\actSIDs{\actIn{\dbVdef}{\dbVdc}}{\dvVdef=\pidVV}}\eIden
      }
    }
  } 
  \\
  \eV'_{\hV_1} & \!\!\defeq \!\! 
    \ch{ 
      \ch{\!\chBig{\!\!\pidV{\,\in\,}\prtSet}{\prf{\actSTI{\pidV{=}\dvV}{\actIn{\pidV}{\vVdef}}}\rVV_1}\!}
      {\!\prf{ \actSIDs{\actOut{\dbVB}{\dbVVB}}{\dvVB{=}\dvV}}{\rec{\rVV_2}{\eV''_{\hV_1}}} }
    \!\!\!}
    {\!\prf{\actSIDs{\actIn{\dbVdef}{\dbVdc}}{\dvVdef{\neq}\dvV}}\eIden\!}
  \\
  \eV''_{\hV_1} & \!\!\defeq \!\! 
    \ch{
      \ch{
        \prf{\actSTN{\actOut{\dbVC}{\dbVdc}}{\dvVC{=}\dvV}{\!\actdot}}\rVV_2
      \!}
      {\!
        \prf{\actSN{\actOut{\dbVD}{\dbVVC}}{\dvVD{=}\pidVV\land\dvVVC{=}\loggTuple{\dvVVA}{\dvVVB}}}{\rV}
      }
    \!}
    {\!
      \prf{\actSet{\actIn{\dbVdc}{\dbVdc}}}\eIden
    }	
  \end{align*}
  Monitor $\eV_{\hV_1}\!$ can be optimised by removing redundant 
  recursive constructs such as 
  % recursions, \eg 
  $\rec{\rVVV}{\_}$ that are introduced mechanically by our synthesis.
  \exqed
  \end{exa}
  \medskip

  Monitor $\eV_{\hV_1}\!$ from \Cref{ex:synthesis-bi} (with $\eSembiP{\hV_1} = \eV_{\hV_1}$) is very similar to \eVdet of \Cref{ex:transducers}, differing only in how it defines its insertion branches for unblocking the \sus. 
  For instance, if we consider $\prtSet=\set{\pidVV,\pidVVV}$, $\eSembiP{\hV_1}$ would synthesise two insertion branches, namely $\actSTI{\pidVV=\dvV}{\actIn{\pidVV}{\vVdef}}$ and $\actSTI{\pidVVV=\dvV}{\actIn{\pidVVV}{\vVdef}}$, 
  % whereas
  but if \prtSet also includes $\pidVVVV$, it would add another branch.
  By contrast,  the manually defined 
  \eVdet attains the same result more succinctly via the single insertion branch $\trnsInInsS$.
  %
  %
  % As argued in \Cref{ex:oenf-fail}, 
  % In general, it is not always possible to define a monitor like \eVdet, especially when the formula defines complex conditions in its violating modal necessities, such as $\hV_2$ of \Cref{ex:oenf-fail}.
  %
  %Hence, the lack of ``elegance'' comes as a cost of generality.
  %
  %This is possible since the violating modal necessity $\hNec{\dvVVVInBPat}\hFls$ of $\hV_1$ is shorthand for $\hNec{\actSN{\actIn{\dbVA}{\dbVVB}}{\dvVA{=}\dvVVV}}\hFls$ (for some data binder \dbVA), and condition %% ELABORATE IF THIS EXPLANATION IS NECESSARY!!
  %
  % Despite this, the following results, namely, \Cref{thm:enf,thm:oenf}, show that 
  Importantly, our synthesis provides the witness monitors needed to show enforceability. 
  % and optimality.
  %
  \begin{thm}[Enforceability] \label[theorem]{thm:enf}
    \SHMLnf is bidirectionally enforceable using the monitors and instrumentation of \Cref{fig:mod-bi-re}. 
    % in a  setting. 
    % \qed
  \end{thm}
  \begin{proof}
    % Since \SHML is logically equivalent to \SHMLnf, 
    By \Cref{def:enforceability} the result follows from showing that for every $\hV{\,\in\,}\SHMLnf$ and $\prtSet{\,\subseteq\,}\Port$, \eSembiP{\hV} enforces \hV (for every \prtSet). 
    By \Cref{def:enforcement}, this follows from  \Cref{lemma:soundness-bi,lemma:evt-transparency-bi}, stated and proved in \Cref{sec:enforce-proof}. 
  \end{proof}  

\Cref{thm:enf} entails that the synthesised monitors generated by the function described in \Cref{def:synthesis-bi} do \emph{enforce} their respective \SHMLnf formula and are correct by construction.
\nw{By this we mean that, if the formula $\hV$ being enforced can be expressed in the syntactic fragment \SHMLnf, and it is satisfiable (\ie $\hSemS{\hV} \neq \emptyset$), then the resulting composite system, \eI{\eV}{\pV}, consisting of the synthesises monitor, \mV, composed with the \sus, \pV, is guaranteed to satisfy \hV and the changes to its original behaviour are only those that led to a violation.
It is worth pointing out that should \hV be unsatisfiable, there is very little that can be done by way of enforcement;  
%
% Accordingly, 
the satisfiability caveats in \Cref{def:senf,def:tenf,def:evtenf}  are intentionally inserted so as not to require anything of the synthesised monitor in such cases.  
We argue that this way of dealing with unsatisfiable formulas is not a deficiency of the enforcement setup, but rather a flaw in the correctness specifications being imposed. 
}

\nw{
Note that the enforcement of formulas that use \recHML constructs outside of the \SHML is problematic. 
For instance, consider the disjunction formula $\hV_1 \vee \hV_2$ (recall that disjuctions are not part of the \SHML syntax).
In a branching-time setting, the subformulas $\hV_1$ and $\hV_2$ can, in principle, describe computation from different parts of the computation tree. 
This means that, although the current execution observed by a monitor might provide enough information to determine that one subformula is about to be violated (say $\hV_1$), there could never be an execution that allows the monitor to determine when to intervene whenever \emph{both} subformulas become violated.
More precisely, by intervening to prevent $\hV_1$ from being violated might break transparency, \Cref{def:evtenf}, in cases where $\hV_2$ is still satisfied (and thus $\hV_1 \vee \hV_2$ still holds).
Conversely, not intervening might affect soundness, \Cref{def:senf}, in cases where $\hV_2$ is also violated (and thus $\hV_1 \vee \hV_2$ is certainly violated). 
It has been well established that a number of \recHML properties are not monitorable for a variety of settings~\cite{Francalanza2017FMSD,Achilleos2018FSTTCS,Achilleos2018Fossacs,Aceto2019POPL,AcetoAFIL21,AcetoAFIL21b} and it is therefore reasonable to expect similar limits in the case of enforceability.
}

\subsection{Enforceability Proofs}
\label[section]{sec:enforce-proof}

  \begin{figure}[t]
    \begin{displaymath}
    \begin{array}{r@{\;\,}c@{\;\,}ll} 
    (\pV,\hTru)&\in&\R & \imp \textsl{ true } \\[.5mm]
    (\pV,\hFls)&\in&\R & \imp \textsl{ false } \\[.5mm]
    (\pV,\hBigAnd{i\in\IndSet}\hV_i)&\in&\R & \imp (\pV,\hV_i)\in\R \textsl{ for all } i{\,\in\,}\IndSet \\[.5mm]
    (\pV,\hNec{\actSN{\pate}{\bV}}\hV)&\in&\R &\imp (\forall\acta,\pVV\cdot\pV\wtraS{\acta}\pVV \textsl{ and } \mtchS{\actSN{\pate}{\bV}}{\acta}=\sV)\,\imp\, (\pVV,\hV\sV)\in\R \\[.5mm]
    (\pV,\hMaxXF)&\in&\R & \imp (\pV,\hV\Sub{\hMaxXF}{\hVarX})\in\R \\[.5mm]
    \end{array}
    \end{displaymath}
    where $\mtchS{\actSN{\pate}{\bV}}{\acta}=\sV$ is short for $\mtch{\pate}{\acta}=\sV$ and $\ceval{\bV\sV}{\btrue}$.
    \caption{A satisfaction relation for \SHML formulas}
    \label{fig:uhml-sat}
  \end{figure}

  In what follows, we state and prove monitor soundness and transparency, \Cref{def:senf,def:evtenf} for the synthesis function presented in \Cref{def:synthesis-bi}.
  Upon first reading, the remainder of the section can be safely skipped without affecting the comprehension of the remaining material.
  
  To facilitate the forthcoming proofs we occasionally use the satisfaction semantics for \SHML from \cite{Aceto1999TestingHML,Hennessy1995} which is defined in terms of the \emph{satisfaction relation}, \hSat. When restricted to \SHML, \hSat is the \emph{largest relation} \R satisfying the implications defined in \Cref{fig:uhml-sat}. 
  It is well known that this semantics agrees with the \SHML semantics of \Cref{fig:recHML}. 
  As a result, we use $\pV\hSat\hV$ in lieu of $\pV\in\hSemS{\hV}$. 
  At certain points in our proofs we also refer to the \actt-closure property of \SHML, \Cref{prop:tau-closure}, that was proven in \cite{Aceto1999TestingHML}.
  \begin{prop}
    \label[proposition]{prop:tau-closure}
    if $\pV\traS{\actt}\pV'$ and $\pV\hSatS\hV$ then $\pV'\hSatS\hV$. \qed
  \end{prop}

We start by stating and proving synthesis soundness, which relies on the following technical lemma relating recursive monitor unfolding and its behaviour.
\begin{lem}
  \label[lemma]{lem:monitor-unfold-bisim}
  \eBI{\rec{\rV}{\eV}}{\pV} \bisim \eBI{\bigl(\eV\sub{\rec{\rV}{\eV}}{\rV}\bigr)}{\pV}
\end{lem}
\begin{proof}
  Follows from the instrumentation relation of \Cref{fig:mod-bi-re} and, more importantly, the monitor rule \rtit{eRec}, also in \Cref{fig:mod-bi-re}.
\end{proof}

  \begin{prop}[Soundness] 
    \label[proposition]{lemma:soundness-bi} For every finite port set \prtSet, system state $\pV{\,\in\,}\Sys$ whose port names are included in \prtSet, and $\hV{\,\in\,}\SHMLnf$, if $\isSatF$ then $\eBI{\eSembiP{\hV}}{\pV}{\,\in\,}\hSemS{\hV} $. 
    % \qed
  \end{prop}

\begin{proof}
  To prove that for every system \pV, formula \hV and finite set of ports \prtSet
    $$ \ift \isSat{\hV} \thent \eBI{\eSembiP{\hV}}{\pV}{\,\hSat\,}\hV $$
    we setup the relation \R (below) and show that it is a \emph{satisfaction relation} (\hSat) by demonstrating that that it abides by the rules in \Cref{fig:uhml-sat}. 
    \begin{align*}
      \R & \;\defeq\; 
      \Setdef{(\pVV,\hV)}
      {
        (i)\ \isSat{\hV} \andt \pVV = \eBI{\eSembiP{\hV}}{\pV} \qquad \ort\\
        (ii)\ \isSat{\hV} \andt \pVV = \eBI{\eSembiP{\hMax{\hVarX_1}{\ldots\hMax{\hVarX_n}{\hVV}}}}{\pV} \\
        \quad \andt \hV = \hVV \Sub{\hMax{\hVarX_1}{\ldots\hMax{\hVarX_n}{\hVV}}}{\hVarX_1}\ldots\Sub{\hMax{\hVarX_n}{\hVV}}{\hVarX_n}
      }
    \end{align*}
    The second case  defining the tuples $(\pVV,\hV) \in \R$ (labeled as $(ii)$ for clarity) maps monitored system \eBI{\mV}{\pV} where \mV is obtain by synthesising a formula consisting of a prefix of maximal fixpoint binders of length $n$, \ie $\mV = \eSembiP{\hV}$ where $\hV=\hMax{\hVarX_1}{\ldots\hMax{\hVarX_n}{\hVV}}$, with the \resp \emph{unfolded} formula $\hVV \Sub{\hMax{\hVarX_1}{\ldots\hMax{\hVarX_n}{\hVV}}}{\hVarX_1}\ldots\Sub{\hMax{\hVarX_n}{\hVV}}{\hVarX_n}$. 

     We prove the claim that 
    $\mathord{\relR} \subseteq  \mathord{\hSat}$ by case analysis on the structure of $\hV$.   We here consider the two main cases:

    \begin{description}[leftmargin=5mm]
      \item[Case $\hV = \hMaxX{\hVV}$]
      We consider two subcases for why $(\pVV,\hV) \in \R$, following either condition $(i)$ or $(ii)$:
      \begin{description}[leftmargin=5mm]
        \item[Case $(i)$] We know that $\isSat{\hMaxX{\hVV}}$ and $\pVV = \eBI{\eSembiP{\hMaxX{\hVV}}}{\pV}$ for some \pV. By the rules defining  (\hSat) in \Cref{fig:uhml-sat} we need to show that 
        $$(\eBI{\eSembiP{\hMaxX{\hVV}}}{\pV}, \hVV\Sub{\hMaxX{\hVV}}{\hVarX}) \in \R$$ 
        as well. This follows immediately from rule $(ii)$ defining \R.
        \item[Case $(ii)$]  We know that  
        $\isSat{\hMaxX{\hVV}}$, that 
        $$\qquad\qquad\hMaxX{\hVV} = \hMaxX{(\hVVV\Sub{\hMax{\hVarY_1}{\ldots\hMax{\hVarY_k}{\hMaxX{\hVVV}}}}{\hVarY_1}\ldots\Sub{\hMax{\hVarY_k}{\hMaxX{\hVVV}}}{\hVarY_k})}$$   for some  \hVVV and $k$, and that 
        $\pVV = \eBI{\eSembiP{\hMax{\hVarY_1}{\ldots\hMax{\hVarY_k}{\hMaxX{\hVVV}}}}}{\pV}$ for some \pV.  
        Again, by the rules defining  (\hSat) in \Cref{fig:uhml-sat} we need to show that 
        $$(\eBI{\eSembiP{\hMax{\hVarY_1}{\ldots\hMax{\hVarY_k}{\hMaxX{\hVVV}}}}}{\pV}, 
        \hVVV'
        ) \in \R$$ 
        for $\hVVV'= \hVVV\Sub{\hMax{\hVarY_1}{\ldots\hMax{\hVarY_k}{\hMaxX{\hVVV}}}}{\hVarY_1}$ $\ldots$ $\Sub{\hMax{\hVarY_k}{\hMaxX{\hVVV}}}{\hVarY_k} \Sub{\hMaxX{\hVVV}}{\hVarX}$.
        This follows again from rule $(ii)$ defining \R with an maximal fixpoint binder length set at $n=k+1$.
      \end{description}
      \item[Case $\hV = \hBigAndDhVI 
      \text{ and } 
      \bigdistinct{h\in\IndSet}\actSN{\pate_h}{\bV_h}$] Again we have two subcases to consider for why we have the inclusion $(\pVV,\hV) \in \R$:
      \begin{description}[leftmargin=5mm]
        \item[Case $(i)$] We know that 
        \begin{equation}
          \isSat{\hBigAndDhVI} 
          \label[equation]{eq:frm-sat}
        \end{equation}
        and that $\pVV = \eBI{\eSembiP{\hBigAndDhVI}}{\pV}$ for some \pV. Recall that 
        \begin{equation}
          \eSembiP{\hBigAndDhVI} = \rec{\rVV}{
          \ch{\Big(\chBigI
            \begin{xbrace}{ll}
            \dis{\pateI}{\bVI}{\rVV}{\prtSet} & \quad \text{if } \hV_i{=}\hFls\\
            \prf{\actSIDs{\pateI}{\bVI}}{\eSembiP{\hV_i}} & \quad \text{otherwise}
            \end{xbrace}
          \Big)}{\defmon{\hBigAndDhVI} }	
            }
            \label[equation]{eq:big-mon-syn}
      \end{equation}
      By the rules defining  (\hSat) in \Cref{fig:uhml-sat} (for the case involving $\hBigAnd{i\in\IndSet}\hV_i$ and $\hNec{\actSN{\pate}{\bV}}\hV$ combined) we need to show that 
      \begin{equation}
        \forall i\in \IndSet, \acta,\pVVV \text{ if } \eBI{\eSembiP{\hBigAndDhVI}}{\pV}\wtraS{\acta}\pVVV \text{ and } \mtchS{\actSN{\pate_i}{\bV_i}}{\acta}=\sV) \text{ then } (\pVVV,\hV_i\sV)\in\R
        \label{eq:sat-oblig}
      \end{equation}
       Pick any $\hNec{\actSNVI}\hVI$ and proceed by case analysis:
       \begin{description}[leftmargin=5mm]
         \item[Case $\hNec{\actSNVI}\hVI = \hNec{\actSN{\patOutV}{\bVI}}{\hFls}$] 
         For any output action $\actOutV$ that the system \pV can produce, \ie $\pV \wtraS{\actOutV} \pV'$, that matches the pattern of the necessity formula considered, \ie $\mtchS{\actSN{\patOutV}{\bVI}}{\actOutV} = \sV$,  the monitor synthesised in \Cref{eq:big-mon-syn} transitions as 
         $$ \eSembiP{\hBigAndDhVI} \traS{\ioact{(\actOutV)}{\actdot}} \eSembiP{\hBigAndDhVI}$$
         Thus, by the instrumentation in \Cref{fig:mod-bi-re} (particularly rule \rtit{biDisO}), we conclude that it could \emph{never} be the case that 
         $$\eBI{\eSembiP{\hBigAndDhVI}}{\pV}\wtraS{\actOutV}\pVVV \text{ for any } \pVVV$$ meaning that condition \eqref{eq:sat-oblig} is satisfied.
         \item[Case $\hNec{\actSNVI}\hVI = \hNec{\actSN{\patInV}{\bVI}}{\hFls}$]  The reasoning is analogous to the previous case. 
         For any input action $\actInV$ that $\pV \wtraS{\actInV} \pV'$, that matches the pattern of the necessity formula, $\mtchS{\actSN{\patInV}{\bVI}}{\actInV} = \sV$,  the monitor synthesised in \Cref{eq:big-mon-syn} transitions as 
         $$ \eSembiP{\hBigAndDhVI} \traS{\ioact{\actdot}{(\actInV)}} \eSembiP{\hBigAndDhVI}$$
         Thus, by the instrumentation in \Cref{fig:mod-bi-re} (particularly rule \rtit{biDisI}), we conclude that it could \emph{never} be the case that 
         $$\eBI{\eSembiP{\hBigAndDhVI}}{\pV}\wtraS{\actInV}\pVVV \text{ for any } \pVVV$$ meaning that condition \eqref{eq:sat-oblig} is satisfied. 
         \item[Case $\hVI \neq \hFls$]  
         From \Cref{eq:frm-sat} we know that for any \acta such that $\mtchS{\actSN{\pate_i}{\bV_i}}{\acta}=\sV$ it holds that 
         \begin{equation}
          \isSat{\hVI\sV}. \label[equation]{eq:sat-small}  
         \end{equation}
         Now if $\pV \wtra{\acta} \pV'$, from the form of \eSembiP{\hBigAndDhVI} in \Cref{eq:big-mon-syn} and $\bigdistinct{h\in\IndSet}\actSN{\pate_h}{\bV_h}$ we conclude that 
         $$ \eSembiP{\hBigAndDhVI} \traS{\ioact{\acta}{\acta}} \eSembiP{\hVI}\sV = \eSembiP{\hVI\sV}$$
         Thus, by the instrumentation in \Cref{fig:mod-bi-re} (particularly rules \rtit{biTrnI} and \rtit{biTrnO}) we conclude
         $$ \eBI{\eSembiP{\hBigAndDhVI}}{\pV} \wtraS{\acta} \eBI{\eSembiP{\hVI\sV}}{\pV'}$$
         and from \Cref{eq:sat-small} and the definition of \R (i) we conclude that \linebreak
         $(eBI{\eSembiP{\hVI\sV}}{\pV'},\hVI\sV) \in \R$, thus satisfying \Cref{eq:sat-oblig} as required.
       \end{description}
        \item[Case $(ii)$]  We know that \isSat{\hBigAndDhVI}, that for all $i \in \IndSet$ 
        \begin{equation*}
          \hVI = \hVV_{i}\Sub{\hMax{\hVarY_1}{\ldots\hMax{\hVarY_k}{\hVV_{i}}}}{\hVarY_1}\ldots\Sub{\hMax{\hVarY_k}{\hVV_{i}}}{\hVarY_k} \text{ for some } \hVV_{i}
        \end{equation*}
        and that $\pVV = \eBI{\eSembiP{\hMax{\hVarY_1}{\ldots\hMax{\hVarY_k}{\hBigAndDhVI}}}}{\pV}$ for some \pV.  Similar to the previous case, we need to show that \pVV satisfies a requirement akin to \Cref{eq:sat-oblig}.   This follows using a similar reasoning employed in the previous case, \Cref{lem:monitor-unfold-bisim} and the transitivity of (strong) bisimulation. \qedhere  
      \end{description} 
    \end{description}
\end{proof}

\bigskip
We next state and prove the fact that the synthesis function of \Cref{def:synthesis-bi} is eventually transparent, according to \Cref{def:evtenf}. 
%
% From \Cref{cor:evtran-implies-tran}, we can also conclude that it is transparent following \Cref{def:tenf}. 
%
This proof for eventual transparency refers to the auxiliary \Cref{lemma:evt-x} and another transparency result \Cref{lemma:transparency-bi} (Transparency) for \Cref{def:tenf}, defined and proved below.
% , and whose proofs are provided in \Cref{sec:appendix-aux}.
%
The proof of \Cref{lemma:evt-x}, in turn, relies on the following technical lemma 
\nw{which states that any sequence $\tau$ transitions from a composite system enforced by a monitor synthesised from a conjuncted modal guard formula according to \Cref{def:synthesis-bi} can be decomposed such that the monitored system is allowed to produce external actions by the monitor remains in the same state.
} 
\nw{which states that any sequence of $\tau$ transitions from a composite  system enforced by a monitor synthesised from a conjunctive formula,  according to \Cref{def:synthesis-bi}, stems from a corresponding sequence of external actions of the monitored system while the monitor remains in the same state.
}
% whose proof is given following the end of the current one.
	\begin{lem}\label[lemma]{lemma:evt-y-bi} 
		For every formula of the form $\hBigAndUhVI$ and system states $\pV$ and $\pVV$, if $\eBI{\eSembiP{\hBigAndUhVI}}{\pV}\traSC{\actt}\ \pVV$ then there exists some state  $\pV'$ and trace $\trr$ such that $\pV\wtraS{\trr}\pV'$ and $\pVV=\eBI{\eSembiP{\hBigAndUhVI}}{\pV'}$.
	\end{lem} 
  \begin{proof}
    % [Proof for \Cref{lemma:evt-y-bi}] 
    % We must now prove that for every formula of the form $\hBigAndUhVI$ and states $\pV$ and $\pVV$, if $\eBI{\eSembiP{\hBigAndUhVI}}{\pV}\traSC{\actt}\pVV$ then there exists some state $\pV'$ and trace $\trr$ such that $\pV\wtraS{\trr}\pV'$ and $\pVV=\eBI{\eSembiP{\hBigAndUhVI}}{\pV'}$. 
    %
    We proceed by mathematical induction on the number of \actt transitions.
	
    \begin{Case}[0\text{ transitions}]
      This case holds trivially given that $\pV\wtraS{\trE}\pV$ and so that $\pVV=\eBI{\eSembiP{\hBigAndUhVI}}{\pV}$.
    \end{Case} \vspace{-2mm}
    
    \begin{Case}[k+1\text{ transitions}] 
      Assume that $\eBI{\eSembiP{\hBigAndUhVI}}{\pV}\traSCN{\actt}{k+1}\pVV$ and so we can infer that
      \begin{gather}
      \eBI{\eSembiP{\hBigAndUhVI}}{\pV}\traS{\actt}\pVV' \quad \text{(for some }\pVV') \label{proof:evt-y-bi-1} \\
      \pVV'\traSCN{\actt}{k}\pVV.  \label{proof:evt-y-bi-2} 
      \end{gather} 
      By the definition of \eSembiS{\!-\!} we know that $\eSembiP{\hBigAndUhVI}$ synthesises the monitor $$\rec{\rVV}{\chBigI\begin{xbrace}{ll} \dis{\pateI}{\bVI}{\rVV}{\prtSet} & \ift \hV{=}\hFls \\ \prf{\actSIDs{\pateI}{\bVI}}{\eSembiP{\hVI}} & \otherwiset \end{xbrace}}$$ 
      which can be unfolded into 
      \begin{gather}
      \eSembiP{\hBigAndDhVI}{\,=\,}\chBigI\begin{xbrace}{ll} \dis{\pateI}{\bVI}{\eV}{\prtSet} & \ift \hVI{=}\hFls \\ \prf{\actSIDs{\pateI}{\bVI}}{\eSembiP{\hVI}} & \otherwiset \end{xbrace}  \label{proof:evt-y-bi-3}
      \end{gather}
      and so from \eqref{proof:evt-y-bi-3} we know that the $\actt$-reduction in \eqref{proof:evt-y-bi-1} can be the result of rules \rtit{iAsy}, \rtit{iDisO} or \rtit{iDisI}. We therefore inspect each case.
      \begin{itemize}
        \item \rtit{iAsy}: By rule \rtit{iAsy}, from \eqref{proof:evt-y-bi-1} we can deduce that 
        \begin{gather}
        \exists\pV''\cdot\pV\traS{\actt}\pV'' \label{proof:evt-y-bi-4} \\
        \pVV'=\eI{\eSemS{\hBigAndUhVI}}{\pV''} \label{proof:evt-y-bi-5}
        \end{gather}
        and so by \eqref{proof:evt-y-bi-2}, \eqref{proof:evt-y-bi-5} and the \emph{inductive hypothesis} we know that 
        \begin{gather}
        \exists\pV',\trr\cdot\pV''\wtraS{\trr}\pV' \andt \pVV=\eI{\eSemS{\hBigAndUhVI}}{\pV'}. \label{proof:evt-y-bi-6}
        \end{gather}
        Finally, by \eqref{proof:evt-y-bi-4} and \eqref{proof:evt-y-bi-6} we can thus conclude that $\exists\pV',\trr\cdot\pV\wtraS{\trr}\pV'$ and also that $\pVV=\eI{\eSemS{\hBigAndUhVI}}{\pV'}$.
        \item \rtit{iDisI}: By rule \rtit{iDisI} and from \eqref{proof:evt-y-bi-1} we infer that 
        \begin{gather}
        \exists\pV''\cdot\pV\traS{\actInVB}\pV'' \label{proof:evt-y-bi-7} \\
        \eSembiP{\hBigAndUhVI}\traS{\ioact{\actdot}{\actInVB}}\eV' \label{proof:evt-y-bi-8} \\
        \pVV'=\eBI{\eV'}{\pV''} \label{proof:evt-y-bi-9}
        \end{gather}
        and from \eqref{proof:evt-y-bi-3} and by the definition of \disS we can infer that the reduction in \eqref{proof:evt-y-bi-8} occurs when the synthesised monitor inserts action $\actInV$ and then reduces back to\linebreak[4]$\eSembiP{\hBigAndUhVI}$ allowing us to infer that
        \begin{gather}
        \eV'=\eSembiP{\hBigAndUhVI}.  \label{proof:evt-y-bi-10}
        \end{gather}
        Hence, by \eqref{proof:evt-y-bi-2}, \eqref{proof:evt-y-bi-9} and \eqref{proof:evt-y-bi-10} we can apply the \emph{inductive hypothesis} and deduce that
        \begin{gather}
        \exists\pV',\trr\cdot\pV''\wtraS{\trr}\pV' \andt \pVV=\eBI{\eSembiP{\hBigAndUhVI}}{\pV'} \label{proof:evt-y-bi-11}
        \end{gather}
        so that by \eqref{proof:evt-y-bi-7} and \eqref{proof:evt-y-bi-11} we finally conclude that $\exists\pV',\trr\cdot\pV\wtraS{\actInVB\trr}\pV'$ and that $\pVV=\eBI{\eSembiP{\hBigAndUhVI}}{\pV'}$ as required, and so we are done.
        \item \rtit{iDisO}: We omit the proof for this case as it is very similar to that of case \rtit{iDisI}.
        \qedhere
      \end{itemize}
    \end{Case} 
  \end{proof}
  % \vspace{1mm}
  
\nw{The following lemma builds on \Cref{lemma:evt-y-bi}, and states that the monitor obtained from a sequence of transitions $\tr$  and a synthesised monitor $\eSembiP{\hVV}$ can be calculated using the function $\afterF{\hV}{\tr}$ and the synthesis function given in \Cref{def:synthesis-bi}.}
\begin{lem} \label[lemma]{lemma:evt-x}
  For every set of names $\Pi$, formula $\hV{\,\in\,}\SHMLnf$, state $\pV$ and trace $\tr$, if $\eI{\eSembiP{\hV}}{\pV}\wtraS{\tr}\eI{\eV'}{\pV'}$ then $\exists\hVV{\,\in\,}\SHMLnf\cdot\hVV=\afterF{\hV}{\tr}$ and $\eSembiP{\hVV}=\eV'$. 
\end{lem}

\begin{proof}
  % [Proof for \Cref{lemma:evt-x}]
   We need to prove that for every formula $\hV{\,\in\,}\SHMLnf$, if we assume that $\eI{\eSembiP{\hV}}{\pV}\wtraS{\tr}\eI{\eV'}{\pV'}$ then there must exist some formula \hVV, such that $\hVV=\afterF{\hV}{\tr}$ and $\eSembiP{\hVV}=\eV'$. 
	We proceed by induction on the length of \tr.
	
	\begin{Case}[\tr=\trE] This case holds vacuously since when $\tr{\,=\,}\trE$ then $\eV'{\,=\,}\eSembiP{\hV}$ and $\hV{\,=\,}\afterF{\hV}{\trE}$. % as required.
	\end{Case} 
	
	\begin{Case}[\tr=\acta\trr] Assume that $\eI{\eSembiP{\hV}}{\pV}\wtraS{\acta\trr}\eI{\eV'}{\pV'}$ from which by the definition $\wtraS{\tr}$ we can infer that there are $\pVV$ and $\pVV'$ such that
		\begin{gather}
		\eI{\eSembiP{\hV}}{\pV}\traSC{\actt}\pVV  \label{proof:evt-x-1} \\
		\pVV\traS{\acta}\pVV'  \label{proof:evt-x-2} \\
		\pVV'\wtraS{\trr}\eI{\eV'}{\pV'}.  \label{proof:evt-x-3}
		\end{gather}
		We now proceed by case analysis on \hV.
		\begin{itemize}
			\item $\hV{\,=\,}\hVarX$: This case does not apply since $\eSembiP{\hFls}$ and $\eSembiP{\hVarX}$ do not yield a valid monitor.
			\item $\hV{\,\in\,}\set{\hFls,\hTru}$: Since $\eSembiP{\hTru}{\,=\,}\eIden$ we know that the \actt-reductions in \eqref{proof:evt-x-1} are only possible via rule \rtit{iAsy} which means that $\pV\traSC{\actt}\pV''$ and $\pVV{\,=\,}\eI{\eSembiP{\hTru}}{\pV''}$. The latter allows us to deduce that the reduction in \eqref{proof:evt-x-2} is only possible via rule \rtit{iTrn} and so we also know that $\pV''\traSC{\acta}\pV'''$ and $\pVV'{\,=\,}\eI{\eSembiP{\hTru}}{\pV'''}$. Hence, by \eqref{proof:evt-x-3} and the \emph{inductive hypothesis} we conclude that 
			\begin{gather}
			\exists \hVV\in\SHMLnf\cdot\hVV=\afterF{\hTru}{\trr}  \label{proof:evt-x-4} \\
			\eSembiP{\hVV}=\eV'.   \label{proof:evt-x-5}
			\end{gather}
			Since from the definition of \afterFS we know that $\afterF{\hTru}{\acta\trr}$ equates to $\afterF{\afterF{\hTru}{\acta}}{\trr}$ and $\afterF{\hTru}{\acta}{\,=\,}\hTru$, from  \eqref{proof:evt-x-4} we can conclude that $\hVV=\afterF{\hTru}{\acta\trr}$ and so this case holdssince we also know \eqref{proof:evt-x-5}.
      The case for \hFls is analogous.
			\item $\hV{\,=\,}\hBigAndUhVI$ and $\bigdistinct{i\in\IndSet}\actSNVI$: Since $\hV{\,=\,}\hBigAndUhVI$, by the definition of \eSembiS{-} we know that  $\rec{\rVV}{\chBigI\begin{xbrace}{ll} \dis{\pateI}{\bVI}{\rVV}{\prtSet} & \ift \hVI{=}\hFls \\ \prf{\actSIDs{\pateI}{\bVI}}{\eSembiP{\hVI}} & \otherwiset \end{xbrace}}$ which can be unfolded into 
			\begin{gather}
			\eSembiP{\hBigAndDhVI}{\,=\,}\chBigI\begin{xbrace}{ll} \dis{\pateI}{\bVI}{\eV}{\prtSet} & \ift \hVI{=}\hFls \\ \prf{\actSIDs{\pateI}{\bVI}}{\eSembiP{\hVI}} & \otherwiset \end{xbrace}  \label{proof:evt-x-7}
			\end{gather} 
			and so by \eqref{proof:evt-x-1}, \eqref{proof:evt-x-7} and \Cref{lemma:evt-y-bi} we conclude that $\exists \pV''\cdot\pV\wtraS{\trr}\pV''$ and that
			\begin{gather}
			\pVV=\eI{\eSembiP{\hBigAndUhVI}}{\pV''}.  \label{proof:evt-x-8}
			\end{gather}
			Hence, by \eqref{proof:evt-x-7} and \eqref{proof:evt-x-8} we know that the reduction in \eqref{proof:evt-x-2} can only happen if $\exists\pV'''\cdot\pV''\traS{\acta}\pV'''$ and \acta matches an identity transformation $\prf{\actSIDs{\pateJ}{\bVJ}}\eSembiP{\hVJ}$ (for some $j\in\IndSet$) which was derived from $\hNec{\actSNVJ}\hVJ$ (where $\hVJ\neq\hFls$). We can thus deduce that
			\begin{gather}
			\pVV'=\eI{\eSembiP{\hVJ\sV}}{\pV'''}  \label{proof:evt-x-9} \\
			\mtch{\pateJ}{\acta}=\sV \andt \ceval{\bVJ\sV}{\boolT} \label{proof:evt-x-10}
			\end{gather}
			and so by \eqref{proof:evt-x-3}, \eqref{proof:evt-x-9} and the \emph{inductive hypothesis} we deduce that 
			\begin{gather}
			\exists \hVV\in\SHMLnf\cdot\hVV=\afterF{\hVJ\sV}{\trr}  \label{proof:evt-x-11} \\
			\eSembiP{\hVV}=\eV'. \label{proof:evt-x-12}
			\end{gather}
			Now since we know \eqref{proof:evt-x-10}, by the definition of \afterFS we infer that 
			\begin{gather}
			\begin{array}{rcl}
				\afterF{\hBigAndUhVI}{\acta\trr}&=&\afterF{\afterF{\hBigAndUhVI}{\acta}}{\trr}\\
				&=&\afterF{\hVJ\sV}{\trr}
			\end{array} \label{proof:evt-x-13}
			\end{gather}
			and so from \eqref{proof:evt-x-11} and \eqref{proof:evt-x-13} we conclude that 
			\begin{gather}
			 	\hVV=\afterF{\hBigAndUhVI}{\acta\trr}.  \label{proof:evt-x-14} 
			\end{gather}
			Hence, this case is done by \eqref{proof:evt-x-12} and \eqref{proof:evt-x-14}.
			\item $\hV{\,=\,}\hMaxX{\hVV}$ and $\hVarX{\,\in\,}\fv{\hVV}$: Since $\hV{\,=\,}\hMaxX{\hVV}$, by the syntactic rules of \SHMLnf we know that $\hVV{\,\notin\,}\set{\hFls,\hTru}$ since $\hVarX{\,\notin\,}\fv{\hVV}$, and that $\hVV{\,\neq\,}\hVarX$ since logical variables must be guarded, hence we know that $\hVV$ can only be of the form 
			\begin{gather}
			\hVV=\hMax{\hVarY_1}{\ldots\hMax{\hVarY_n}{\hBigAndUhVI}}.  \label{proof:evt-x-14.5} 
			\end{gather}
			where $\hMax{\hVarY_1}{\ldots\hMax{\hVarY_n}{}}$ denotes an arbitrary number of fixpoint declarations, possibly none. Hence, knowing \eqref{proof:evt-x-14.5}, by unfolding every fixpoint in $\hMaxX{\hVV}$ we reduce the formula to $$\hV=\hBigAndUhVI\subb{\subE{\hMaxX{\hMax{\hVarY_1}{\ldots\hMax{\hVarY_n}{\hBigAndUhVI}}}}{\hVarX},\ldots}$$ and so from this point onwards the proof proceeds as per that of case $\hV{\,=\,}\hBigAndUhVI$ which allows us to deduce that
			\begin{gather}
			\exists\hVV'{\,\in\,}\SHMLnf\cdot\hVV'{\,=\,}\afterF{\hBigAndUhVI\subb{\ldots}}{\acta\trr} \label{proof:evt-x-15} \\
			\eSembiP{\hVV'}=\eV'.  \label{proof:evt-x-16}
			\end{gather}
			From \eqref{proof:evt-x-14.5}, \eqref{proof:evt-x-15} and the definition of \afterFS we can therefore conclude that 
			\begin{gather}
			\exists\hVV'{\,\in\,}\SHMLnf\cdot\hVV'{\,=\,}\afterF{\hMaxX{\hVV}}{\acta\trr} \label{proof:evt-x-17} 
			\end{gather}
			and so this case holds by \eqref{proof:evt-x-16} and \eqref{proof:evt-x-17}.
		\end{itemize}
		Hence, the above cases suffice to show that the case for when $\tr=\acta\trr$ holds.  \qedhere
    % \vspace{-8mm}
	\end{Case}	
\end{proof}

\smallskip
The transparency proof following \Cref{def:tenf} is given below.

\begin{prop}[Transparency] 
  \label[proposition]{lemma:transparency-bi} For every state $\pV{\,\in\,}\Sys$ and $\hV{\,\in\,}\SHMLnf$, if $\pV{\,\in\,}\hSemS{\hV}$ then $\eBI{\eSembiP{\hV}{}}{\pV}{\,\bisim\,}\pV$. 
\end{prop}

\begin{proof}
  % [Proof for \Cref{lemma:transparency-bi} (Transparency)]
	% We need to prove that for every system \pV, if $\pV{\,\in\,}\hSemS{\hV}$ then $\eBI{\eV}{\pV}\bisim\pV$.
	%
	Since $\pV{\,\in\,}\hSemS{\hV}$ is analogous to $\pV{\,\hSatS\,}\hV$ we prove that relation $\R{\,\defeq\,}\Setdef{(\pV,\eBI{\eSembiP{\hV}}{\pV})}{\pV{\,\hSatS\,}\hV}$ is a \emph{strong bisimulation relation} that satisfies the following transfer properties:
	\begin{enumerate}[\quad(a)]
		\item if $\pV\traS{\actu}\pV'$ then $\eBI{\eSembiP{\hV}}{\pV}\traS{\actu}\pVV'$ and $(\pV',\pVV')\in\R$
		\item if $\eBI{\eSembiP{\hV}}{\pV}\traS{\actu}\pVV'$ then $\pV\traS{\actu}\pV'$ and $(\pV',\pVV')\in\R$
	\end{enumerate} 
	\noindent We prove $(a)$ and $(b)$ separately by assuming that $\pV{\,\hSatS\,}\hV$ in both cases as defined by relation \R and conduct these proofs by case analysis on \hV. 
	We now proceed to prove $(a)$ by case analysis on \hV. 
		
%	\begin{Cases}[\hV{\,\in\,}\set{\hTru,\hFls,\hVarX,\hMaxXF}] We omit showing the proof for these cases as they are very similar to the respective ones in \Cref{lemma:transparency} (a) proven in \Cref{sec:synthesis-function}.
%	\end{Cases}

	\begin{Cases}[\hV\in\Set{\hFls,\hVarX}]
		Both cases do not apply since $\nexists \pV\cdot \pV\hSatS\hFls$ and similarly since $\hVarX$ is an open-formula and so $\nexists \pV\cdot \pV\hSatS\hVarX$.
	\end{Cases}

	\begin{Case}[\hV=\hTru]
		We now assume that 
		\begin{gather}
			\pV\hSatS\hTru \label{proof:trans-bi-tt-1}\\
			\pV\traS{\actu}\pV'  \label{proof:trans-bi-tt-2}
		\end{gather}
		and since $\actu\in\set{\actt,\acta}$, we must consider both cases.
		\begin{itemize}
			\item \emph{$\actu=\actt$:} Since $\actu=\actt$, we can apply rule \rtit{iAsy} on \eqref{proof:trans-bi-tt-2} and get that
				\begin{gather}
					\eBI{\eSembiP{\hTru}}{\pV}\traS{\actt}\eBI{\eSembiP{\hTru}}{\pV'} \label{proof:trans-bi-tt-3}
				\end{gather}
				as required. Also, since we know that every process satisfies \hTru, we know that $\pV'\hSatS\hTru$, and so by the definition of \R we conclude that
				\begin{gather}
					(\pV',\eBI{\eSembiP{\hTru}}{\pV'})\in\R \label{proof:trans-bi-tt-4}
				\end{gather}
				as required. This means that this case is done by \eqref{proof:trans-bi-tt-3} and \eqref{proof:trans-bi-tt-4}.
			\item \emph{$\actu=\acta$:} Since $\eSembiP{\hTru}{\,=\,}\eIden$ encodes the `catch-all' monitor, \eIdenFull, by rules \rtit{eRec} and \rtit{eTrn} we can apply rule \rtit{iTrnI/O} and deduce that $\eIden\traS{\ioact{\acta}{\acta}}\eIden$, which we can further refine as
				\begin{gather}
					\eBI{\eSembiP{\hTru}}{\pV}\traS{\acta}\eBI{\eSembiP{\hTru}}{\pV'} \label{proof:trans-bi-tt-5}
				\end{gather}
				as required. Once again since $\pV'\hSatS\hTru$, by the definition of \R we can infer that
				\begin{gather}
					(\pV',\eBI{\eSembiP{\hTru}}{\pV'})\in\R \label{proof:trans-bi-tt-6}
				\end{gather}
				as required, and so this case is done by \eqref{proof:trans-bi-tt-5} and \eqref{proof:trans-bi-tt-6}.
		\end{itemize}
	\end{Case}

	{
		\newcommand{\hVAnd}{\ensuremath{\hV_{\hAndS}}\xspace}
		\newcommand{\formula}{\hBigAndU{i\in\IndSet}\hNec{\actSNVI}\hVI}
		\newcommand{\synForm}{\rec{\rVV}{
		\ch{
			\Big(\chBigI\begin{xbrace}{lc} \dis{\pateI}{\bVI}{\rVV}{\prtSet} &\; \text{if }\hVI=\hFls \\ \prf{\actSIDs{\pateI}{\bVI}}{\eSembiP{\hVI}} &\; \text{otherwise} \end{xbrace}\Big)
		}{
			\defmon{\hVAnd}
		}
	}}
	\begin{Case}[\hV=\formula]
		We assume that 
		\begin{gather}
			\pV\hSatS\formula \label{proof:trans-bi-nec-1}\\
			\pV\traS{\actu}\pV'  \label{proof:trans-bi-nec-2}
		\end{gather}
		and by the definition of \hSatS and \eqref{proof:trans-bi-nec-1} we have that for every index $i{\,\in\,}\IndSet$ and action $\actb{\,\in\,}\Act$,
		\begin{gather}
			\ift \pV\wtraS{\actb}\pV' \andt \mtchS{\actSNVI}{\actb}=\sV \thent \pV'\hSatS\hVI\sV. \label{proof:trans-bi-nec-3}
		\end{gather}
		Since $\actu\in\set{\actt,\acta}$, we must consider both possibilities.
		\begin{itemize}
			\item \emph{$\actu=\actt$:} Since $\actu=\actt$, we can apply rule \rtit{iAsy} on \eqref{proof:trans-bi-nec-2} and obtain
				\begin{gather}
				\eBI{\eSembiP{\formula}}{\pV}\traS{\actt}\eBI{\eSembiP{\formula}}{\pV''} \label{proof:trans-bi-nec-4}
				\end{gather}
				as required. Since \actu=\actt, and since we know that \SHML is \actt-closed, from \eqref{proof:trans-bi-nec-1}, \eqref{proof:trans-bi-nec-2} and \Cref{prop:tau-closure}, we can deduce that $\pV'\hSatS\formula$, so that by the definition of \R we conclude that
				\begin{gather}
					(\pV'',\eBI{\eSembiP{\formula}}{\pV''})\in\R \label{proof:trans-bi-nec-5}
				\end{gather}
				as required. This subcase is therefore done by \eqref{proof:trans-bi-nec-4} and \eqref{proof:trans-bi-nec-5}.	
			\item \emph{$\actu=\acta$:} Since $\actu=\acta$, from \eqref{proof:trans-bi-nec-2} we know that
				\begin{gather}
					\pV\traS{\acta}\pV' \label{proof:trans-bi-nec-6}
				\end{gather}
				and by the definition of \eSembiS{-} we can immediately deduce that
				\begin{gather}
					\eSembiP{\hVAnd} = \synForm \label{proof:trans-bi-nec-7}
				\end{gather}
				where $\hVAnd{\,\defeq\,}\formula$. Since the branches in the conjunction are all disjoint, $\bigdistinct{i\in\IndSet}\actSNVI$, we know that \emph{at most one} of the branches can match the same (input or output) action \acta. Hence, we consider two cases, namely:
				\begin{itemize}
					\item \emph{No matching branches (\ie $ \forall i\in\IndSet\cdot\mtchS{\actSNVI}{\acta}=\sVundef$):} Since none of the symbolic actions in \eqref{proof:trans-bi-nec-7} can match action \acta, we can infer that if \acta is an \emph{input}, \ie $\acta=\actInV$, then it will match the default monitor \defmon{\hVAnd} and transition via rule \rtit{iTrnI}, while if it is an \emph{output}, \ie $\acta=\actOutV$, rule \rtit{iDef} handles the underspecification. In both cases, the monitor reduces to $\eIden$. Also, notice that rules \rtit{iDisO} and \rtit{iDisI} cannot be applied since if they do, it would mean that \pV can also perform an erroneous action, which is not the case since we assume \eqref{proof:trans-bi-nec-1}. Hence, we infer that
					\begin{gather}
						\eBI{\eSembiP{\formula}}{\pV}\traS{\acta}\eBI{\eSembiP{\hTru}}{\pV'} \quad \text{(since \eIden=\eSembiP{\hTru})} \label{proof:trans-bi-nec-9}
					\end{gather}
					as required. Also, since any process satisfies \hTru, we know that $\pV'\hSatS\hTru$, and so by the definition of \R we conclude that
					\begin{gather}
						(\pV',\eBI{\eSembiP{\hTru}}{\pV'})\in\R \label{proof:trans-bi-nec-10}
					\end{gather}
					as required. This case is therefore done by \eqref{proof:trans-bi-nec-9} and \eqref{proof:trans-bi-nec-10}. \medskip
					\item \emph{One matching branch (\ie $ \exists j{\,\in\,}\IndSet\cdot\mtchS{\actSNVJ}{\acta}{\,=\,}\sV$):} From \eqref{proof:trans-bi-nec-7} we can infer that the synthesised monitor can only disable the (input or output) actions that are defined by violating modal necessities. However, from  \eqref{proof:trans-bi-nec-3} we also deduce that \pV is \emph{incapable} of executing such an action as that would contradict assumption \eqref{proof:trans-bi-nec-1}. Hence, since we now assume that $ \exists j\in\IndSet\cdot\mtchS{\actSNVJ}{\acta}=\sV$, from \eqref{proof:trans-bi-nec-7} we deduce that this action can only be transformed by an identity transformation and so by rule \rtit{eTrn} we have that 
					\begin{gather}
						\prf{\actSNVJ}{\eSembiP{\hVJ}} \traS{\ioact{\acta}{\acta}} \eSembiP{\hVJ\sV}. \label{proof:trans-bi-nec-11}
					\end{gather}
					By applying rules \rtit{eSel}, \rtit{eRec} on \eqref{proof:trans-bi-nec-11} and by \eqref{proof:trans-bi-nec-6}, \eqref{proof:trans-bi-nec-7} and \rtit{iTrnI/O} (depending on whether \acta is an input or output action) we get that
					\begin{gather}
						\eBI{\eSembiP{\formula}}{\pV}\traS{\acta}\eBI{\eSembiP{\hVJ\sV}}{\pV'} \label{proof:trans-bi-nec-12}
					\end{gather}
					as required. By \eqref{proof:trans-bi-nec-3}, \eqref{proof:trans-bi-nec-6} and since we assume that  $ \exists j\in\IndSet\cdot\mtchS{\actSNVJ}{\acta}=\sV$ we have that $\pV'\hSat\hVJ\sV$, and so by the definition of \R we conclude that
					\begin{gather}
						(\pV',\eBI{\eSembiP{\hVJ\sV}}{\pV'})\in\R \label{proof:trans-bi-nec-13}
					\end{gather}
					as required. Hence, this subcase holds by \eqref{proof:trans-bi-nec-12} and \eqref{proof:trans-bi-nec-13}.
				\end{itemize}
		\end{itemize}
	\end{Case}
	}

	{\newcommand{\formula}{\hMaxXF}
	\newcommand{\formUnfold}{\hVMaxXFSub}
	\newcommand{\formUnfoldFull}{(\hBigAndUhVI)\subb{\cdot\cdot}}
	\newcommand{\subOmega}{\subb{\SubE{\formula}{\hVarX},\SubE{(\hMax{\hVarY_{0}}{\ldots\hMax{\hVarY_{n}}{\hBigAndUhVI}})}{\hVarY_{0}},\ldots}}%,\SubE{(\hMax{\hVarY_{n}}{\hBigAndDhVI})}{\hVarY_{n}}}}
		\begin{Case}[\hV=\formula \andt \hVarX{\,\in\,}\fv{\hV}]
			Now, lets assume that 
			\begin{gather}
			\pV\traS{\actu}\pV'  \label{proof:trans-max-1}
			\end{gather}
			and that $\pV\hSatS\formula$ from which by the definition of \hSatS we have that
			\begin{gather}
			\pV\hSatS\formUnfold. \label{proof:trans-max-2}
			\end{gather}
			Since $\formUnfold{\,\in\,}\SHMLnf$, by the restrictions imposed by \SHMLnf we know that: \hV cannot be \hVarX because (bound) logical variables are required to be \emph{guarded}, and it also cannot be \hTru or \hFls since \hVarX is required to be defined in \hV, \ie $\hVarX\in\fv{\hV}$. Hence, we know that \hV can only have the following form, that is
			\begin{gather}
			\hV=\hMax{\hVarY_{0}}{\ldots\hMax{\hVarY_{n}}{\hBigAndDhVI}}  \label{proof:trans-max-3}
			\end{gather}
			and so by \eqref{proof:trans-max-2}, \eqref{proof:trans-max-3} and the definition of \hSatS we have that 
			\begin{gather}
			\pV\hSatS\formUnfoldFull \quad \wheret \label{proof:trans-max-4} \\
			\subb{\cdot\cdot}=\subOmega. \nonumber
			\end{gather}
			Since we know \eqref{proof:trans-max-1} and \eqref{proof:trans-max-4}, 
      from this point onwards the proof proceeds as in the previous case. We thus omit the details. 
      %
      % from this point onwards the proof proceeds as per the previous case. We thus omit this part of the proof and immediately deduce that
			% \begin{gather}
			% \exists\eV'\cdot \eBI{\eSembiP{\formUnfoldFull}}{\pV} \traS{\actu} \eBI{\eSembiP{\eV'}}{\pV'} \label{proof:trans-max-5} \\
			% (\pV',\eBI{\eSembiP{\eV'}}{\pV'})\in\R  \label{proof:trans-max-6}
			% \end{gather}
			% Since $\eSembiP{\formUnfoldFull}$ synthesises the \emph{unfolded equivalent} as per \eSembiP{\formUnfold}, from \eqref{proof:trans-max-5} we can conclude that
			% \begin{gather}
			% \exists\eV'\cdot \eBI{\eSembiP{\formUnfold}}{\pV} \traS{\actu} \eBI{\eSembiP{\eV'}}{\pV'} \label{proof:trans-max-7} 
			% \end{gather}
			% as required, and so this case holds by \eqref{proof:trans-max-6} and \eqref{proof:trans-max-7}.\\
		\end{Case}
	}

	% \noindent
  These cases thus allow us to conclude that $(a)$ holds. 
  We now proceed to prove $(b)$ using a similar case analysis approach. 
	
%	\begin{Cases}[\hV{\,\in\,}\set{\hFls,\hVarX,\hMaxXF}] We omit showing the proof for these cases as they are very similar to the respective ones in \Cref{lemma:transparency} (b) proven in \Cref{sec:synthesis-function}.
%	\end{Cases}
		
	\begin{Cases}[\hV\in\Set{\hFls,\hVarX}]
		Both cases do not apply since $\nexists \pV\cdot \pV\hSatS\hFls$ and similarly since $\hVarX$ is an open-formula and $\nexists \pV\cdot \pV\hSatS\hVarX$.
	\end{Cases}
		
	\begin{Case}[\hV=\hTru]
		Assume that 
		\begin{gather}
			\pV\hSatS\hTru \label{proof:trans-bi-b-tt-1}\\
			\eBI{\eSembiP{\hTru}}{\pV}\traS{\actu}\pVV'  \label{proof:trans-bi-b-tt-2}
		\end{gather}
		Since $\actu\in\set{\actt,\actInV,\actOutV}$, we must consider each case.
		\begin{itemize}
			%% -- OMIT IN DISSERTATION
			%
			\item \emph{$\actu=\actt$:} Since $\actu=\actt$, the transition in \eqref{proof:trans-bi-b-tt-2} can be performed via \rtit{iDisI}, \rtit{iDisO} or \rtit{iAsy}. We must therefore consider these cases.
			\begin{itemize}
				\item \emph{\rtit{iAsy}:} From rule \rtit{iAsy} and \eqref{proof:trans-bi-b-tt-2} we thus know that $\pVV'=\eBI{\eSembiP{\hTru}}{\pV'}$ and that $\pV\traS{\actt}\pV'$ as required. Also, since every process satisfies \hTru, we know that $\pV'\hSatS\hTru$ as well, and so we are done since by the definition of \R we know that $(\pV',\eBI{\eSembiP{\hTru}}{\pV'})\in\R$.
				\item \emph{\rtit{iDisI}:} From rule \rtit{iDisI} and \eqref{proof:trans-bi-b-tt-2} we know that: $\pVV'=\eBI{\eV'}{\pV'}$, $\pV\traS{\actInV}\pV'$ and that
				\begin{gather}
					\eSembiP{\hTru}\traS{\ioact{\actdot}{\actInVB}}\eV'. \label{proof:trans-bi-b-tt-3}
				\end{gather}
				Since $\eSembiP{\hTru}=\eIden$ we can deduce that \eqref{proof:trans-bi-b-tt-3} is \emph{false} and hence this case does not apply.
				\item \emph{\rtit{iDisO}:} The proof for this case is analogous as to that of case \rtit{iDisI}.
			\end{itemize}
			%
			%% -- OMIT IN DISSERTATION
			%
			%% -- INCLUDE IN DISSERTATION
%			\item \emph{$\actu\in\set{\actt,\actOutV}$:}  Since $\actu{\,\in\,}\set{\actt,\actOutV}$, the transition in \eqref{proof:trans-bi-b-tt-2} can be performed via rules \rtit{iDisI}, \rtit{iDisO} or \rtit{iAsy} when $\actu{\,=\,}\actt$, and via \rtit{iDef}, \rtit{iTrnO} and \rtit{iEnO} when $\actu{\,=\,}\actOutV$. We omit the proofs for these cases as they are analogous to the ones in case $\hV=\hTru$ of \Cref{lemma:soundness}. Specifically, the case for rules \rtit{iDef} and \rtit{iAsy} remain very much the same, while the cases for rules \rtit{iDisI} and \rtit{iDisO} are analogous to that of rule \rtit{iSup}, and cases \rtit{iTrnO} and \rtit{iEnO} are the same as the cases for \rtit{iTrn} and \rtit{iIns} respectively, modulo output actions.
			%% -- INCLUDE IN DISSERTATION
			%
			\item \emph{$\actu=\actInV$:} Since $\actu=\actInV$, the transition in \eqref{proof:trans-bi-b-tt-2} can be performed either via \rtit{iTrnI} or \rtit{iEnI}. We consider both cases.
			\begin{itemize}
				\item \emph{\rtit{iEnI}:} This case also does not apply since if the transition in \eqref{proof:trans-bi-b-tt-2} is caused by rule \rtit{iEnI} we would have that $\eSembiP{\hTru}\traS{\ioact{\actInV}{\actdot}}\eV$ which is \emph{false} since $\eSembiP{\hTru}=\eIden=\eIdenFull$ and rules \rtit{eRec}, \rtit{eSel} and \rtit{eTrn} state that for every \actInV, $\eIden\traS{\ioact{\actInV}{\actInV}}\eIden$, thus leading to a contradiction.
				\item \emph{\rtit{iTrnI}:} By applying rule \rtit{iTrnI} on \eqref{proof:trans-bi-b-tt-2} we know that $\pVV'=\eBI{\eV'}{\pV'}$ such that
				\begin{gather}
					\eSembiP{\hTru}\traS{\ioact{\actInV}{\actInVV}}\eV'. \label{proof:trans-bi-b-tt-5} \\
					\pV\traS{\actInVV}\pV' \label{proof:trans-bi-b-tt-4} 
				\end{gather}
				Since $\eSembiP{\hTru}=\eIden=\eIdenFull$, by applying rules \rtit{eRec}, \rtit{eSel} and \rtit{eTrn} to \eqref{proof:trans-bi-b-tt-5} we know that $\actInV=\actInVV$, $\eV'=\eIden=\eSembiP{\hTru}$, meaning that $\pVV'=\eBI{\eSembiP{\hTru}}{\pV'}$. Hence, since every process satisfies \hTru we know that $\pV'\hSatS\hTru$, so that by the definition of \R we conclude 
				\begin{gather}
					(\pV',\eBI{\eSembiP{\hTru}}{\pV'})\in\R. \label{proof:trans-bi-b-tt-6}
				\end{gather}
				Hence, we are done by \eqref{proof:trans-bi-b-tt-4} and \eqref{proof:trans-bi-b-tt-6} since we know that $\actInV=\actInVV$.
			\end{itemize}
			%
			%% -- OMIT IN DISSERTATION
			%
			\item \emph{$\actu=\actOutV$:} When $\actu=\actOutV$, the transition in \eqref{proof:trans-bi-b-tt-2} can be performed via \rtit{iDef}, \rtit{iTrnO} or \rtit{iEnO}. We omit this proof as it is very similar to that of case $\actu=\actInV$.
			%
			%% -- OMIT IN DISSERTATION
			%
		\end{itemize}
	\end{Case}
	
	{\newcommand{\formula}{\hBigAndUhVI}
		\newcommand{\synForm}{\rec{\rVV}{
			\ch{
				\Big(\chBigI\begin{xbrace}{lc} \dis{\pateI}{\bVI}{\rVV}{\prtSet}  &\; \text{if }\hVI=\hFls\!\! \\ \prf{\actSIDs{\pateI}{\bVI}}{\eSembiP{\hVI}} &\; \text{otherwise} \!\!\end{xbrace}\Big)
			}{
				\defmon{\hBigAndDhVI}
			}
		}}
		\begin{Case}[\hV=\formula]
			We now assume that 
			\begin{gather}
				\pV\hSatS\formula \label{proof:trans-bi-nec-b-0} \\
				\eBI{\eSembiP{\formula}}{\pV}\traS{\actu}\pVV'.  \label{proof:trans-bi-nec-b-0.5}
			\end{gather}
			From \eqref{proof:trans-bi-nec-b-0} and by the definition of \hSatS we can deduce that
			\begin{gather}
				\forall i\in\IndSet,\actb\in\Act\cdot \ift \pV\wtraS{\actb}\pV' \andt \mtchS{\actSNVI}{\actb}=\sV \thent \pV'\hSatS\hVI\sV \label{proof:trans-bi-nec-b-1}
			\end{gather}
			and from \eqref{proof:trans-bi-nec-b-0.5} and the definition of \eSembiS{-} we have that
			\begin{gather}
				\eBI{\Big(\synForm\Big)}{\pV'}\traS{\actu}\pVV'.  \label{proof:trans-bi-nec-b-2}
			\end{gather}
			From \eqref{proof:trans-bi-nec-b-2} we can deduce that the synthesised monitor can only disable an (input or output) action \actb 
      when its occurrence would violate a conjunct of the form $\hNec{\actSNVI}\hFls$  for some $i\in\IndSet$.
      % when this satisfies a violating modal necessity. 
      However, we also know that \pV is \emph{unable} to perform such an action as otherwise it would contradict assumption \eqref{proof:trans-bi-nec-b-1}. Hence, we can safely conclude that the synthesised monitor in \eqref{proof:trans-bi-nec-b-2} does \emph{not} disable any (input or output) actions of \pV, and so by the definition of \disS we conclude that
			\begin{gather}
				\begin{array}{l}
					\forall\actInV,\actOutV\in\Act,\pV'\in\Sys\cdot  \\ \quad
					\begin{xbrackets}{rclc}
						\pV\traS{\actInV}\pV'&\imp&\eSembiP{\formula}\ntraS{\ioact{\actdot}{\actIn{\pidV}{\vVV}}} \;\text{(for all \vVV)} & \andt \!\!\\
						\pV\traS{\actOutV}\pV'&\imp&\eSembiP{\formula}\ntraS{\ioact{\actInV}{\actdot}}
					\end{xbrackets}\!.
				\end{array}
				\label{proof:trans-bi-nec-b-3}
			\end{gather}
			Since $\actu\in\set{\actt,\actInV,\actOutV}$, we must consider each case.
			\begin{itemize}
				\item \emph{$\actu=\actt$:} Since $\actu=\actt$, from \eqref{proof:trans-bi-nec-b-0.5} we know that
				\begin{gather}
					\eBI{\eSembiP{\formula}}{\pV}\traS{\actt}\pVV'  \label{proof:trans-bi-nec-b-4}
				\end{gather}
				The \actt-transition in \eqref{proof:trans-bi-nec-b-4} can be the result of rules \rtit{iAsy}, \rtit{iDisI} or \rtit{iDisO}; we thus consider each eventuality.
				\begin{itemize}
					\item \emph{\rtit{iAsy}:} As we assume that the reduction in \eqref{proof:trans-bi-nec-b-4} is the result of rule \rtit{iAsy}, we know that $\pVV'=\eBI{\eSembiP{\formula}}{\pV'}$ and that
					\begin{gather}
						\pV\traS{\actt}\pV'  \label{proof:trans-bi-nec-b-5}
					\end{gather}
					as required. Also, since \SHML is \actt-closed, by \eqref{proof:trans-bi-nec-b-0}, \eqref{proof:trans-bi-nec-b-5} and \Cref{prop:tau-closure} we deduce that $\pV'\hSatS\formula$ as well, so that by the definition of \R we conclude that
					\begin{gather}
						(\pV',\eBI{\eSembiP{\formula}}{\pV'})\in\R  \label{proof:trans-bi-nec-b-6}
					\end{gather}
					and so we are done by \eqref{proof:trans-bi-nec-b-5} and \eqref{proof:trans-bi-nec-b-6}.
					\item \emph{\rtit{iDisI}:} By assuming that reduction \eqref{proof:trans-bi-nec-b-4} results from \rtit{iDisI}, we have that $\pVV'=\eBI{\eV'}{\pV'}$ and that
					\begin{gather}
						\eSembiP{\formula}\traS{\ioact{\actdot}{\actInV}}\eV' \label{proof:trans-bi-nec-b-7} \\
						\pV\traS{\actInV}\pV' \label{proof:trans-bi-nec-b-8}
					\end{gather}
					By \eqref{proof:trans-bi-nec-b-3} and \eqref{proof:trans-bi-nec-b-8} we can, however, deduce that for every value \vVV, we have that $\eSembiP{\formula}\ntraS{\ioact{\actdot}{\actIn{\pidV}{\vVV}}}$. This contradicts  \eqref{proof:trans-bi-nec-b-7} and so this case does not apply.		
					\item \emph{\rtit{iDisO}:} As we now assume that the reduction in \eqref{proof:trans-bi-nec-b-4} results from \rtit{iDisO}, we have that $\pVV'=\eBI{\eV'}{\pV'}$ and that
					\begin{gather}
						\pV\traS{\actOutV}\pV' \label{proof:trans-bi-nec-b-9}  \\
						\eSembiP{\formula}\traS{\ioact{\actOutV}{\actdot}}\eV'. \label{proof:trans-bi-nec-b-10} 
					\end{gather}
					Again, this case does not apply since from \eqref{proof:trans-bi-nec-b-3} and \eqref{proof:trans-bi-nec-b-9} we can deduce that $\eSembiP{\formula}\ntraS{\ioact{\actOutV}{\actdot}}$ which contradicts \eqref{proof:trans-bi-nec-b-10}.		
				\end{itemize}
				%\item \emph{$\actu\in\set{\actt,\actOutV}$:}  Since $\actu{\,\in\,}\set{\actt,\actOutV}$, the transition in \eqref{proof:trans-bi-nec-b-0.5} can be performed via rules \rtit{iDisI}, \rtit{iDisO} or \rtit{iAsy} when $\actu{\,=\,}\actt$, and via \rtit{iDef}, \rtit{iTrnO} and \rtit{iEnO} when $\actu{\,=\,}\actOutV$. We omit the proofs for these cases as they are analogous to the ones in case $\hV=\formula$ of \Cref{lemma:soundness}. Specifically, the case for rules \rtit{iDef} and \rtit{iAsy} remain very much the same, while the cases for rules \rtit{iDisI} and \rtit{iDisO} are analogous to that of rule \rtit{iSup}, and cases \rtit{iTrnO} and \rtit{iEnO} are the same as the cases for \rtit{iTrn} and \rtit{iIns} respectively, modulo output actions.
			 	%
				\item \emph{$\actu=\actInV$:} When $\actu=\actInV$, the transition in \eqref{proof:trans-bi-nec-b-2} can be performed via rules \rtit{iEnI} or \rtit{iTrnI}, we consider both possibilities.
					\begin{itemize}
					\item \emph{\rtit{iEnI}:} This case does not apply since from \eqref{proof:trans-bi-nec-b-2} and by the definition of \eSembiS{-} we know that the synthesised monitor does not include action enabling transformations.
					\item \emph{\rtit{iTrnI}:} By assuming that \eqref{proof:trans-bi-nec-b-2} is obtained from rule \rtit{iTrnI} we know that
					\begin{gather}
						\synForm{\,\traS{\ioact{\actInV}{\actInVV}}\,}\eV'  \label{proof:trans-bi-nec-b-12} \\
						\pV\traS{\actInVV}\pV' \label{proof:trans-bi-nec-b-13} \\
						\pVV'=\eBI{\eV'}{\pV'}. \label{proof:trans-bi-nec-b-13.5}
					\end{gather}
					Since from \eqref{proof:trans-bi-nec-b-3} we know that the synthesised monitor in \eqref{proof:trans-bi-nec-b-12} does not disable any action performable by \pV, and since from the definition of \eSembiS{-} we know that the synthesis is incapable of producing action replacing monitors, we can deduce that
					\begin{gather}
						\actInV=\actInVV. \label{proof:trans-bi-nec-b-14}
					\end{gather}
					With the knowledge of \eqref{proof:trans-bi-nec-b-14}, from \eqref{proof:trans-bi-nec-b-13} we can thus deduce that 
					\begin{gather}
						\pV\traS{\actInV}\pV'  \label{proof:trans-bi-nec-b-17}
					\end{gather}
					as required.
					Knowing \eqref{proof:trans-bi-nec-b-14} we can also deduce that in \eqref{proof:trans-bi-nec-b-12} the monitor transforms an action \actInV either (i) via an identity transformation that was synthesised from one of the \emph{disjoint} conjunction branches, \ie from a branch $\prf{\actSIDs{\pateJ}{\bVJ}}{\eSembiP{\hVJ}}$ for some $j{\,\in\,}\IndSet$, or else (ii) via the default monitor synthesised by \defmon{\formula}. We consider both eventualities.
					\begin{enumerate}[(i)]
						\item In this case we apply rules \rtit{eRec}, \rtit{eSel} and \rtit{eTrn} on \eqref{proof:trans-bi-nec-b-12} and deduce that
						\begin{gather}
							\exists j\in\IndSet\cdot\mtchS{\actSNVJ}{\actInV}=\sV  \label{proof:trans-bi-nec-b-15}\\
							\eV'=\eSembiP{\hVJ\sV}.  \label{proof:trans-bi-nec-b-16}
						\end{gather}
						and so from \eqref{proof:trans-bi-nec-b-17}, \eqref{proof:trans-bi-nec-b-15} and \eqref{proof:trans-bi-nec-b-1} we infer that $\pV'\hSatS\hVJ\sV$ from which by the definition of \R we have that $(\pV',\eBI{\eSembiP{\hVJ\sV}}{\pV'})\in\R$, and so from \eqref{proof:trans-bi-nec-b-13.5} and \eqref{proof:trans-bi-nec-b-16} we can conclude that
						\begin{gather}
							(\pV',\pVV')\in\R  \label{proof:trans-bi-nec-b-18}
						\end{gather}
						as required, and so this case is done by \eqref{proof:trans-bi-nec-b-17} and \eqref{proof:trans-bi-nec-b-18}.
						\item When we apply rules \rtit{eRec}, \rtit{eSel} and \rtit{eTrn} we deduce that $\eV'=\eIden$ and so by the definition of \eSembiS{-} we have that 
						\begin{gather}
							\eV'=\eSembiP{\hTru}. \label{proof:trans-bi-nec-b-19}
						\end{gather} 
						Consequently, as every process satisfies \hTru, we know that $\pV'\hSatS\hTru$ and so by the definition of \R we have that $(\pV',\eBI{\eSembiP{\hTru}}{\pV'})\in\R$, so that from \eqref{proof:trans-bi-nec-b-13.5} and \eqref{proof:trans-bi-nec-b-19} we can conclude that
						\begin{gather}
							(\pV',\pVV')\in\R  \label{proof:trans-bi-nec-b-20}
						\end{gather}
						as required. Hence this case is done by \eqref{proof:trans-bi-nec-b-17} and \eqref{proof:trans-bi-nec-b-20}.
					\end{enumerate}
				\end{itemize}
			\item \emph{$\actu=\actOutV$:} When $\actu=\actOutV$, the transition in \eqref{proof:trans-bi-nec-b-2} can be performed via \rtit{iDef}, \rtit{iTrnO} or \rtit{iEnO}. We omit the proof for this case due to its strong resemblance to that of case $\actu=\actInV$.
			\end{itemize}
		\end{Case}
	}
	
	{
		\newcommand{\formula}{\hMaxXF}
		\newcommand{\formUnfold}{\hVMaxXFSub}
		\newcommand{\formUnfoldFull}{\hBigAndUhVI\subb{\cdot\cdot}}
		\newcommand{\subOmega}{\subb{\SubE{\formula}{\hVarX},\SubE{(\hMax{\hVarY_{0}}{\ldots\hMax{\hVarY_{n}}{\hBigAndUhVI}})}{\hVarY_{0}},\ldots}}%,\SubE{(\hMax{\hVarY_{n}}{\hBigAndDhVI})}{\hVarY_{n}}}}
		\begin{Case}[\hV=\formula \andt \hVarX\in\fv{\hV}]
			Now, lets assume that 
			\begin{gather}
			\eBI{\eSembiP{\formula}}{\pV} \traS{\actu} \pVV' \label{proof:trans-max-b-1}
			\end{gather}
			and that $\pV\hSatS\formula$ from which by the definition of \hSatS we have that
			\begin{gather}
			\pV\hSatS\formUnfold. \label{proof:trans-max-b-2}
			\end{gather}
			Since $\formUnfold{\,\in\,}\SHMLnf$, by the restrictions imposed by \SHMLnf we know that: \hV cannot be \hVarX because (bound) logical variables are required to be \emph{guarded}, and it also cannot be \hTru or \hFls since \hVarX is required to be defined in \hV, \ie $\hVarX\in\fv{\hV}$. Hence, we know that \hV can only have the following form, that is
			\begin{gather}
			\hV=\hMax{\hVarY_{0}}{\ldots\hMax{\hVarY_{n}}{\hBigAndDhVI}}  \label{proof:trans-max-b-3}
			\end{gather}
			and so by \eqref{proof:trans-max-b-2}, \eqref{proof:trans-max-b-3} and the definition of \hSatS we have that 
			\begin{gather}
			\pV\hSatS\formUnfoldFull \quad \wheret \label{proof:trans-max-b-4} \\
			\subb{\cdot\cdot}=\subOmega. \nonumber
			\end{gather}
			Since $\eSembiP{\formUnfoldFull}$ synthesises the \emph{unfolded equivalent} of \eSembiP{\formula}, from \eqref{proof:trans-max-b-1} we know that 
			\begin{gather}
			\eBI{\eSembiP{\formUnfoldFull}}{\pV} \traS{\actu} \pVV'. \label{proof:trans-max-b-5}
			\end{gather}
			Hence, since we know \eqref{proof:trans-max-b-4} and \eqref{proof:trans-max-b-5}, from this point onwards the proof proceeds as per the previous case. We thus omit showing the remainder of this proof.\\
		\end{Case}
	} %\vspace{-5mm}
	
	\noindent From the above cases we can therefore conclude that $(b)$ holds as well. 
\end{proof}

\medskip

We are finally in a position to state and prove our eventual transparency results following \Cref{def:evtenf}.

  \begin{prop}[Eventual Transparency] \label[proposition]{lemma:evt-transparency-bi}
    For every input port set \prtSet, \SHMLnf formula $\hV$, system states $\pV,\pV'{\,\in\,}\Sys$, action disabling monitor $\eV'$ and trace \tr, if $\eI{\eSembiP{\hV}}{\pV}\wtraS{\tr}\eI{\eV'}{\pV'}$ and $\pV'\in\hSemS{\afterF{\hV}{\tr}}$ then $\eI{\eV'}{\pV'}\bisim\pV'$. \qed
  \end{prop}

  \begin{proof} We must prove that for every formula $\hV{\,\in\,}\SHMLnf$ if $\eSembiP{\hV}{\,=\,}\eV$ then $\evtenfdef{\eV}{\hV}$. 
    %
    % Since \SHMLnf is equally expressive as \SHML 
    We prove that for every $\hV{\,\in\,}\SHMLnf$, if $\eI{\eSembiP{\hV}}{\pV}\wtraS{\tr}\eI{\eV'}{\pV'}$ and $\pV'\hSatS\afterF{\hV}{\tr}$ then $\eI{\eV'}{\pV'}\bisim\pV'$.
    \noindent Now, assume that 
    \begin{gather}
    \eI{\eSembiP{\hV}}{\pV}\wtraS{\tr}\eI{\eV'}{\pV'}  \label{proof:evt-transparency-1} \\
    \pV'\hSatS\afterF{\hV}{\tr}		\label{proof:evt-transparency-2}
    \end{gather}
    and so from \eqref{proof:evt-transparency-1} and \Cref{lemma:evt-x} we have that 
    \begin{gather}
    \exists\hVV\in\SHMLnf\cdot\hVV=\afterF{\hV}{\tr} 	\label{proof:evt-transparency-3} \\
    \eSembiP{\afterF{\hV}{\tr}}=\eV'=\eSembiP{\hVV}.	\label{proof:evt-transparency-4} 
    \end{gather}
    Hence, knowing \eqref{proof:evt-transparency-2} and \eqref{proof:evt-transparency-3}, by \Cref{lemma:transparency-bi} (Transparency) we conclude that $\eI{\eSembiP{\afterF{\hV}{\tr}}}{\pV'}\bisim\pV'$ as required, and so we are done.
  \end{proof}

  %
  % \noindent The proofs of these propositions are given in \Cref{sec:appendix}.

\section{Transducer Optimality}
\label{sec:optimality}

Recall \Cref{def:enforcement} from \Cref{sec:enforcement}.
Through criteria such as \Cref{def:senf,def:evtenf}, it defined what it means for a monitor to adequately enforce a formula.
However, it did \emph{not} assess whether a monitor is (to some extent) the ``\emph{best}'' that one can find to enforce a property.
In order to define such a notion we must first be able to compare monitors to one another via some kind of \emph{distance measurement} that tells them apart.
One potential measurement is to assess the monitor's level of \emph{intrusiveness} when enforcing the property.

\begin{figure}
	\[
	\mc{\eV}{\txr} \defeq 
	\begin{xbrace}{cl}
	1+\mc{\eV'}{\txr'} & \ift \txr{\,=\,}\actu\txr' \andt \eI{\eV}{\trcsys{\actu\txr'}}\traS{\actu'}\eI{\eV'}{\trcsys{\txr'}} \andt \actu{\,\neq\,}\actu' \\		
	1+\mc{\eV'}{\txr} & \ift \txr{\,\in\,}\set{\actu\txr',\txrE} \andt \eI{\eV}{\trcsys{\txr}}\traS{\actu'}\eI{\eV'}{\trcsys{\txr}} \\[1mm]
	\mc{\eV'}{\txr'} & \ift \txr{\,=\,}\actu\txr' \andt \eI{\eV}{\trcsys{\actu\txr'}}\traS{\actu}\eI{\eV'}{\trcsys{\txr'}} \\
	\len{\txr} & \ift \txr{\,\in\,}\set{\actu\txr',\txrE} \andt \forall\actu'\cdot\eI{\eV}{\trcsys{\txr}}\ntraS{\actu'}
	\end{xbrace}
	\] 
  % \vspace{-5mm}
	\caption{Modification Count (\mcS).}
	\label{fig:mc}
  % \vspace{-5mm}
\end{figure}

In \Cref{fig:mc} we  define function \mcS that inductively analyses a system run, represented as an explicit trace \txr, and counts the number of modifications applied by the monitor. %applies in order to enforce the property of interest.
In each case the function reconstructs a trace system \trcsys{\txr} and instruments it with the monitor \eV in order to assess the type of transformation applied.
Specifically, in the first two cases, \mcS increments the counter whenever the monitor adapts, disables or enables an action, and then it recurses to keep on inspecting the run (\ie the suffix $\txr'$ in the first, and the same trace $\txr$ in the second) vis-a-vis the subsequent monitor state, $\eV'$.
The third case, specifies that the counter stays unmodified when the monitor applies an identity transformation, while the last case returns the length of \txr when $\eBI{\eV}{\trcsys{\txr}}$ is unable to execute further.
\begin{exa} \label[example]{ex:mc} Recall the monitors of \Cref{ex:transducers} and consider the following system run $\txr^{0}{=}\actInPidVvVA.\actInPidVvVB.\actt.\actOut{\pidV}{\vVVB}.\actOut{\pidV}{\vVVB}.\actLoggTuple{\vVB}{\vVVB}$. 
	For \eVe and \eVa, function \mcS respectively counts three enabled actions, \ie $\mc{\eVe}{\txr^{0}}{=}3$, and four adapted actions, \ie $\mc{\eVa}{\txr^{0}}{=}4$ (since $\actLoggTuple{\vVB}{\vVVB}$ remains unmodified). 
	The maximum count of 5 is attained by \eVd as it immediately blocks the first input $\actInPidVvVA$, and so none of the actions in $\txr^{0}$ can be executed by the composite system \ie $\forall\actu\cdot\eBI{\eVd}{\trcsys{\txr^{0}}}\!\ntraS{\actu}$ and so $\mc{\eVd}{\txr^{0}}{=}5$.
	Similarly, $\mc{\eVdt}{\txr^{0}}{=}4$ since \eVdt allows the first request to be made, but blocks the second erroneous one, and as a result it also forbids the execution of the subsequent actions, \ie $\forall\actu\cdot\eBI{\eVdt}{\trcsys{\txr^{0}}}\traS{\actInPidVvVA}{\!\cdot\!}\ntraS{\actu}$.
	%{(\ch{\eVcls}{\prf{\pidVOutAPat}{\bigl(\ch{\prf{\actSN{\pidVOutBPat}{\actdot}}{\eVd}}{\ch{\prf{\actSet{\actLoggTuple{\vVA}{\dvVVA}}}{\eVdt} }{\eVcls}} \bigr)}})}{\trcsys{\actReq.\actAns.\actAns.\actLog}}\ntraS{\actu}. $
	%
	Finally, \eVdet performs the least number of modifications, namely $\mc{\eVdet}{\txr^{0}}{=}2$.
	The first modification is caused when the monitor blocks the second erroneous input and internally \emph{inserts} a default input value that allows the composite system to proceed over a \actt-action.
	This contrasts with \eVd and \eVdt which fail to perform this insertion step thereby contributing to their high intrusiveness score. 
	The second modification is attained when \eVdet suppresses the redundant response.
	\exqed 
\end{exa}

\begin{figure}
  % [t]
	\[
		\trpBI{\eV} \defeq 
		\begin{xbrace}{cl}
			\emptyset & \ift \eV{\,=\,}\rV \\
			\bigunionU{i\in\IndSet}\trpBI{\eV_i}  & \ift \eV{\,=\,}\chBigU{i\in\IndSet}\eV_i \\
			\trpBI{\eV'} & \ift \eV{\,=\,}\recX{\eV'} \ort \eV{\,=\,}\prf{\actSID{\pate}{\bV}}{\eV'} \\
			\set{\DIS}{\,\cup\,}\trpBI{\eV'} & \ift\eV{\,=\,}\prf{\actSTD{\patOutV}{\bV}}{\eV'} \ort \eV{\,=\,}\prf{\actSTI{\bV}{\actInV}}{\eV'} \\
			\set{\EN}{\,\cup\,}\trpBI{\eV'} & \ift\eV{\,=\,}\prf{\actSTD{\patInV}{\bV}}{\eV'} \ort \eV{\,=\,}\prf{\actSTI{\bV}{\actOutV}}{\eV'} \\
			\set{\ADPT}{\,\cup\,}\trpBI{\eV'} & \ift\eV{\,=\,}\prf{\actSTN{\pate}{\bV}{\pate'}}{\eV'} \andt \pate'{\,\neq\,}\underline{\pate}{\,\neq\,}\actdot
		\end{xbrace}
	\]
%	\[
%		\si{\eV} \defeq 
%		\begin{xbrace}{cl}
%			\emptyset & \ift \eV{\,=\,}\rV 
%			\\
%			\bigunionU{i\in\IndSet}\si{\eV_i}  & \ift \eV{\,=\,}\chBigU{i\in\IndSet}\eV_i 
%			\\
%			\set{\pidV}{\,\cup\,}\si{\eV'} & \ift \eV{\,=\,}\prf{\actSTN{\pate}{\bV}{\actInV}}{\eV'} (\text{for all }\pidV{\in}\Port)
%			\\			
%			\si{\eV'} & \otherwiset\!\!\ift \eV{\,\in\,}\set{\recX{\eV'},\prf{\actSTN{\pate}{\bV}{\pate'}}{\eV'}} 
%		\end{xbrace}
%	\]  
%	\[
%		\fairpre{\eV}{\eVV} \defeq \trpBI{\eV}\subseteq\trpBI{\eVV} \andt \si{\eV}\subseteq\si{\eVV}
%	\]
	% \vspace{-3mm} 
	\caption{Enforcement Capabilities (\trpBIS).}
	\label{fig:ec}
  % \vspace{-5mm} 
\end{figure}

We can now use function \mcS to compare monitors to each other in order to identify the least intrusive one, \ie the monitor that applies the least amount of transformations when enforcing a specific property.
However, for this comparison to be fair, we must also compare like with like.
This means that if a monitor enforces a formula by only disabling actions, it is only fair to compare it to other monitors of the same kind.
It is reasonable to expect that monitors with more enforcement capabilities are likely to be ``better'' than those with fewer capabilities.
We determine the enforcement capabilities of a monitor via function \trpBIS of \Cref{fig:ec}. 
It inductively analyses the structure of the monitor and deduces whether it can enable, disable and adapt actions based on the type of transformation triples it defines.
For instance, if the monitor defines an output suppression triple, $\prf{\actSTD{\patOutV}{\bV}}{\eV'}$, or an input insertion branch, $\prf{\actSTI{\bV}{\actInV}}{\eV'}$, then \trpBIS determines that the monitor can disable actions \DIS, while if it defines an input suppression, $\prf{\actSTD{\patInV}{\bV}}{\eV'}$, or an output insertion branch, $\prf{\actSTI{\bV}{\actOutV}}{\eV'}$, then it concludes that the monitor can enable actions, \EN. 
Similarly, if a monitor defines a replacement transformation, it infers that the monitor can adapt actions, \ADPT.
%

%\cmt{UPDATE!! WE FOCUS ON THE HARDCODED INFO ABOUT THE SYSTEM, ie THE INPUT PORT INFORMATION WHICH CAN BE USED TO INTERACT WITH THE SUS.}
%%
%We also determine the knowledge that the monitor has about the system via the \siS function, which returns a set containing the ports on which the monitor can insert a value.
%%
%It creates this set by inspecting the structure of the monitor, and adding the port identifier to the resultant set whenever the monitor defines an input or output insertion branch.
%%
\begin{exa} Recall the monitors of \Cref{ex:transducers}. 
	With function \trpBIS we determine that $\trpBI{\eVe}{=}\set{\EN}$, $\trpBI{\eVa}{=}\set{\ADPT}$, $\trpBI{\eVd}{=}\trpBI{\eVdt}{=}\trpBI{\eVdet}{=}\set{\DIS}$. 
	Monitors may also have multiple types of enforcement capabilities. 
  For instance, 
  \begin{equation*}
    \trpBI{\recX{\ch{\prf{\actSN{\patInV}{\actdot}}{\rV}}{\prf{\actSN{\patOutV}{\actdot}}{\rV}}}}{=}\set{\EN,\DIS}.
    \tag*{\exqed}  
  \end{equation*}
\end{exa}

With these definitions we now define \emph{optimal enforcement}.
\begin{defi}[Optimal Enforcement] \label[definition]{def:opt-enf-bi} A monitor \eV is \emph{optimal} when enforcing \hV, denoted as $\optenfdef{\eV}{\hV}$, iff it \emph{enforces} \hV (\Cref{def:enforcement}) and when for every state \pV, explicit trace \txr and monitor \eVV, if $\trpBI{\eVV}{\,\subseteq\,}\trpBI{\eV}$, $\enfdef{\eVV}{\hV}$ and $\pV\traS{\txr}$ then $\mc{\eV}{\txr}{\,\leq\,}\mc{\eVV}{\txr}$. \qed
\end{defi}
\noindent \Cref{def:opt-enf-bi} states that an adequate (sound and eventually transparent) monitor \eV is \emph{optimal} for \hV, if one cannot find another adequate monitor \eVV, with the same (or fewer) enforcement capabilities, that performs \emph{fewer modifications} than \eV and is thus \emph{less intrusive}.
%
%Optimal enforcement does not exclude that a monitor with different, or additional enforcement capabilities (or system information) might be more optimal when enforcing \hV, but it ensures that one cannot find a more optimal monitor with less or the same enforcement capabilities and information as \eV.
%

\begin{exa} \label[example]{ex:oenf} %Showing that a monitor \eV implements the optimal approach for enforcing a formula \hV, \optenfdef{\eV}{\hV}, is inherently difficult. 
	%
	%This again transpires from the universal quantifications over all system states \pV, explicit traces \txr, and other adequate monitors \eVV to which \eV is compared. %, and also since \eV must enforce \hV, which we have already argued why this is difficult in \Cref{ex:sound-enf,ex:evtenf}.
	%
	Recall formula $\hV_1$ of \Cref{ex:shml-formula-bi} and monitor \eVdet of \Cref{ex:evtenf}.
	Although showing that $\optenfdef{\eVdet}{\hV_1}$ is inherently difficult, from \Cref{ex:mc} we already get the intuition that it holds since \eVdet imposes the least amount of modifications compared to the other monitors of \Cref{ex:transducers,ex:other-transducers}. 
	We further reaffirm this intuition using systems \pVbo and \pVbi from \Cref{ex:shml-formula-bi}.
	In fact, when considering the invalid runs $\txr^{1}{\defeq}{\actIn{\pidV}{\vVA}.\actt.\actOut{\pidV}{\vVVA}.\actOut{\pidV}{\vVVA}.\actLoggTuple{\vVA}{\vVVA}}$ of \pVbo, and $\txr^{2}{\defeq}{\actIn{\pidV}{\vVA}.\actIn{\pidV}{\vVB}.\actt.\actOut{\pidV}{\vVVB}.\actLoggTuple{\vVB}{\vVVB}}$ of \pVbi, one can easily deduce that no other adequate action disabling monitor can enforce $\hV_1$ with fewer modifications than those imposed by \eVdet, namely, $\mc{\eVdet}{\txr^{1}}{\,=\,}\mc{\eVdet}{\txr^{2}}{\,=\,}1$.
	Furthermore, consider the invalid traces $\txr^{1}\sub{\pidVVV}{\pidV}$ and $\txr^{2}\sub{\pidVVV}{\pidV}$ that are respectively produced by versions of \pVbo and \pVbi that interact on some port $\pidVVV$ instead of $\pidV$ (for any port $\pidVVV{\,\neq\,}\pidV$).
	Since \eVdet binds the port \pidVVV to its data binder \dvV and uses this information in its insertion branch, $\prf{\actSN{\actdot}{\dvVVVInBPat}}\rVV$, the \emph{same} modification count is achieved for these traces, as well \ie $\mc{\eVdet}{\txr^{1}\sub{\pidVVV}{\pidV}}{\,=\,}\mc{\eVdet}{\txr^{2}\sub{\pidVVV}{\pidV}}{\,=\,}1$.  \exqed
\end{exa}

\Cref{ex:oenf} describes the case where formula \hV is optimally enforced by a \emph{finite-state} and \emph{finitely-branching} monitor, \eVdet. 
In the general case, this is not always possible.
\begin{exa} \label[example]{ex:oenf-fail}
	Consider formula $\hV_2$ stating that an initial input on port \pidV followed by another input from some other port $\dvVB{\neq}\pidV$ constitutes invalid system behaviour.
	Also consider monitor $\eV_1$ where $\enfdef{\eV_1}{\hV_2}$.
	\begin{align*}
		\hV_2&\defeq\hNec{\actSet{\actIn{\pidV}{\dbVdc}}}\hNec{\actSN{\actIn{\dbVB}{\dbVdc}}{\dvVB{\neq}\pidV}}\hFls \\ \eV_1&\defeq\prf{\actSet{\actIn{\pidV}{\dbVdc}}}{\rec{\rVV}{(\ch{\prf{\actSIDs{\actdot}{\actIn{\pidVV}{\vVdef}}}\rVV}}{\prf{\actSet{\actIn{\pidV}{\dbVdc}}}\eIden})}
	\end{align*}
	When enforcing a system that generates the run $\txr^{3}{\,\defeq\,}\prf{\actInPidVvVA}\prf{\actIn{\pidVV}{\vVB}}\prf{\actOut{\pidV}{\vVVA}}\txrr^{3}$, monitor $\eV_1$ modifies the trace only \emph{once}.
	Although it disables the input $\actIn{\pidVV}{\vVB}$, it subsequently unblocks the \sus by inserting $\actIn{\pidVV}{\vVdef}$ and so trace $\txr^{3}$ is transformed into $\prf{\actInPidVvVA}\prf{\actt}\prf{\actOut{\pidV}{\vVVA}}\txrr^{3}$. 
	However, for a slightly modified version of $\txr^{3}$, \eg $\txr^{3}\sub{\pidVVV}{\pidVV}$, $\eV_1$ scores a modification count of $2+\len{\txrr^{3}}$. 
  This is the case because, although it blocks the invalid input on port \pidVVV, it fails to insert the default value that unblocks the \sus.
	A more expressive version of $\eV_1$, such as $$\eV_2{\,\defeq\,}\prf{\actSet{\actIn{\pidV}{\dbVdc}}}{\rec{\rVV}{(\ch{ \ch{\prf{\actSIDs{\actdot}{\actIn{\pidVV}{\vVdef}}}\rVV}{ \underline{\prf{\actSIDs{\actdot}{\actIn{\pidVVV}{\vVdef}}}\rVV}} }  {\prf{\actSet{\actIn{\pidV}{\dbVdc}}}\eIden})}},$$ circumvents this problem by defining an extra insertion branch (underlined), but still fails to be optimal in the case of $\txr^{3}\sub{\pidVVVV}{\pidVV}$.
	In this case, there does not exist a way to finitely define a monitor that can insert a default value on every possible input port $\dvVB{\,\neq\,}\pidV$.
	Hence, it means that the optimal monitor $\eV_{\textsf{opt}}$ for $\hV_1$ would be an \emph{infinite branching} one, \ie it requires a countably infinite summation that is not expressible in \Enf, $$\eV_{\textsf{opt}}{\defeq}\prf{\actSet{\actIn{\pidV}{\dbVdc}}}{(\rec{\rVV}{\ch{ \chBigU{\pidVV{\,\in\,}\Port\andt\pidV{\neq}\pidVV}\prf{\actSN{\actdot}{\actIn{\pidVV}{\vVdef}}}\rVV } }{\prf{\actSet{\actIn{\pidV}{\dbVdc}}}\eIden})}$$ or alternatively $$\prf{\actSet{\actIn{\pidV}{\dbVdc}}}{(\rec{\rVV}{\ch{ \chBigU{\pidVV{\,\in\,}\Port}\prf{\actSTI{\pidV{\neq}\pidVV}{\actIn{\pidVV}{\vVdef}}}\rVV } }{\prf{\actSet{\actIn{\pidV}{\dbVdc}}}\eIden})}$$ where the condition $\pidV{\neq}\pidVV$ is evaluated at runtime. \exqed
\end{exa}

Unlike \Cref{ex:oenf}, \Cref{ex:oenf-fail} presents a case where optimality can only be attained by a monitor that defines an \emph{infinite number of branches}; this is problematic since monitors are required to be \emph{finitely} described.
As it is not always possible to find a finite monitor that enforces a formula using the least amount of transformation for every possible system, this indicates that \Cref{def:opt-enf-bi} is too strict. 
We thus mitigate this issue by weakening \Cref{def:opt-enf-bi} and redefine it in terms of the set of system states \SysPrtSet, \ie the set of states that can only perform inputs using the ports specified in a finite $\prtSet{\,\subset\,}\Port$.
Although this weaker version does \emph{not} guarantee that the monitor \eV optimally enforces \hV on \emph{all} possible systems, it, however, ensures optimal enforcement for all the systems that input values via the ports specified in \prtSet.

\begin{defi}[Weak Optimal Enforcement] 
  \label[definition]{def:opt-enf-bi-finite} 
  A monitor \eV is \emph{weakly optimal} when enforcing \hV, denoted as $\optenfdefFinP{\eV}{\hV}$, iff it \emph{enforces} \hV (\Cref{def:enforcement}) and when for every state \underline{$\pV{\,\in\,}\SysPrtSet$}, explicit trace \txr and monitor \eVV, if $\trpBI{\eVV}{\,\subseteq\,}\trpBI{\eV}$, $\enfdef{\eVV}{\hV}$ and $\pV\traS{\txr}$ then $\mc{\eV}{\txr}{\,\leq\,}\mc{\eVV}{\txr}$. \qed
\end{defi}

\begin{exa}
	Monitor $\eV_1$ from \Cref{ex:oenf-fail} ensures that $\hV_2$ is optimally enforced on systems that interact on ports \pidV and \pidVV, \ie when $\prtSet{\,=\,}\set{\pidV,\pidVV}$, while monitor $\eV_2$ guarantees it when $\prtSet{\,=\,}\set{\pidV,\pidVV,\pidVVV}$. \exqed
\end{exa}

We can show that a synthesised monitor \eSembiP{\hV} obtained using the synthesis function \Cref{def:synthesis-bi} from \Cref{sec:synthesis} is also guaranteed to be \nw{weakly} optimal (as stated by  \Cref{def:opt-enf-bi-finite}) when enforcing \hV on a \sus \pV whose input ports are specified by \prtSet, \ie $\pV{\,\in\,}\SysPrtSet$. 
Since our synthesis produces only action disabling monitors, \ie $\trpBI{\eSembiP{\hV}}{\,=\,}\set{\DIS}$ for all \hV and \prtSet, we can limit ourselves to monitors pertaining to the set $\DisEnfS{\defeq}\!\Setdef{\!\!\eVV}{\text{if }\trpBI{\eVV}{\,\subseteq\,}\set{\DIS}\!\!}$.
The proof for \Cref{thm:oenf} below relies on the following lemmas.
\begin{lem} \label[lemma]{lemma:oenf-x1-bi} 
	For every $\eV{\,\in\,}\DisEnfS$ and explicit trace $\txr$, there exists some $\nV$ such that $\mc{\eV}{\txr}{\,=\,}\nV$. \qed
\end{lem}
\begin{lem}  \label[lemma]{lemma:oenf-x2-bi}
	For every action \acta and monitor
   $\eV{\,\in\,}\DisEnfS$, if it is the case that $\eV\traS{\actaa}\eV'$, $\enfdef{\eV}{\hBigAndUhVI}$ and $\mtchS{\actSNVJ}{\acta}{\,=\,}\sV$ (for some $j{\in}\IndSet$) then $\enfdef{\eV'}{\hVJ\sV}$. \qed
\end{lem}
\begin{lem} \label[lemma]{lemma:oenf-x3-bi} 
	For every monitor $\eV{\,\in\,}\DisEnfS$, whenever $\enfdef{\eV}{\hBigAndUhVI}$ and, for some $\eV'$, $\eV\traS{\ioact{\actOutVB}{\actdot}}\eV'$ then $\enfdef{\eV'}{\hBigAndUhVI}$. \qed
\end{lem}
\begin{lem} \label[lemma]{lemma:oenf-x4-bi} 
	For every monitor $\eV{\,\in\,}\DisEnfS$, whenever $\enfdef{\eV}{\hBigAndUhVI}$ and, for some $\eV'$, $\eV\traS{\ioact{\actdot}{\actInVB}}\eV'$ then $\enfdef{\eV'}{\hBigAndUhVI}$. \qed
\end{lem}

\begin{thm}[Weak Optimal Enforcement] \label[theorem]{thm:oenf} For every system 
  % state 
  $\pV{\,\in\,}\SysPrtSet$, explicit\linebreak[4]trace \txr and monitor \eV, if $\trp{\eV}{\,\subseteq\,}\trp{\eSembiP{\hV}}$, $\enfdef{\eV}{\hV}$ and $\pV\traS{\txr}$ implies\linebreak[4]$\mc{\eSembiP{\hV}}{\txr}{\,\leq\,}\mc{\eV}{\txr}$. 
  % \qed
\end{thm}
\begin{proof} Since, from \Cref{lemma:oenf-x1-bi}, we know that for every $\eV{\,\in\,}\DisEnfS$, there exists some $\nV$ such that $\mc{\eV}{\txr}{\,=\,}\nV$, we can prove that if $\enfdef{\eV}{\hV}$, $\pV{\traS{\txr}}$ and $\mc{\eSembiP{\hV}}{\txr}{\,=\,}\nV$ then $\nV{\,\leq\,}\mc{\eV}{\txr}$. 
	We proceed by rule induction on $\mc{\eSembiP{\hV}}{\txr}$.

	\begin{Case}[\mc{\eSembiP{\hV}}{\txr} \whent \txr{\,=\,}\actu\txr' \andt \eBI{\eSembiP{\hV}}{\trcsys{\actu\txr'}}\traS{\actu}\eBI{\eV'_{\hV}}{\trcsys{\txr'}}] Assume that 
		\begin{gather}
			\mc{\eSembiP{\hV}}{\actu\txr'}=\mc{\eV'_{\hV}}{\txr'}=\nV \label{proof:opt-enf-bi--a-1}
		\end{gather}
		which implies that
		\begin{gather}
			\eBI{\eSembiP{\hV}}{\trcsys{\actu\txr'}}\traS{\actu}\eBI{\eV'_{\hV}}{\trcsys{\txr'}} \label{proof:opt-enf-bi--a-2}
		\end{gather}
		and also assume that 
		\begin{gather}
			\enfdef{\eV}{\hV} \label{proof:opt-enf-bi--a-3} 
		\end{gather}
		and that $\pV\traS{\actu\txr'}$. By the rules in our model we can infer that the reduction in \eqref{proof:opt-enf-bi--a-2} can result from rule \rtit{iAsy} when $\actu{\,=\,}\actt$,  \rtit{iDef} and \rtit{iTrnO} when $\actu{\,=\,}\actOutV$, or \rtit{iTrnI} when $\actu{\,=\,}\actInV$. We consider each case individually.
		\begin{itemize}
			\item \rtit{iAsy}: By rule \rtit{iAsy} from \eqref{proof:opt-enf-bi--a-2} we know that $\actu{\,=\,}\actt$ and that
			\begin{gather}
				\eV'_{\hV}=\eSembiP{\hV}.  \label{proof:opt-enf-bi--a-5}
			\end{gather}
			Since from \eqref{proof:opt-enf-bi--a-3} we know that \eV is sound and eventual transparent, we can thus deduce that \eV does not hinder internal \actt-actions from occurring and so the composite system $\eBI{\eSembiP{\hV}}{\trcsys{\actt\txr'}}$ can always transition over \actt via rule \rtit{iAsy}, that is,
			\begin{gather}
				\eBI{\eV}{\trcsys{\actt\txr'}}\traS{\actt}\eBI{\eV}{\trcsys{\txr'}}.   \label{proof:opt-enf-bi--a-6}
			\end{gather}
			Hence, by \eqref{proof:opt-enf-bi--a-1}, \eqref{proof:opt-enf-bi--a-3} and since $\pV\traS{\actt\txr'}$ entails $\pV\traS{\actt}\pV'$ and $\pV'\traS{\txr'}$ we can apply the \emph{inductive hypothesis} and deduce that $\nV{\,\leq\,}\mc{\eV}{\txr'}$ so that by \eqref{proof:opt-enf-bi--a-6} and the definition of \mcS, we conclude that $\nV{\,\leq\,}\mc{\eV}{\actt\txr'}$ as required.			
			\item \rtit{iDef}: From \eqref{proof:opt-enf-bi--a-2} and rule \rtit{iDef} we know that $\actu{\,=\,}\actOutV$, $\eSembiP{\hV}\ntraS{\actOutV}$ and that $\eV'_{\hV}=\eIden$. Since \eIden does not modify actions, we can deduce that $\mc{\eV_{\hV}'}{\txr'}{\,=\,}0$ and so by the definition of \mcS we know that $\mc{\eSembiP{\hV}}{\actOutVB\txr'}{\,=\,}0$ as well. This means that we cannot find a monitor that performs fewer transformations, and so we conclude that $0{\,\leq\,}\mc{\eV}{\actOutVB\txr'}$ as required.		
			\item \rtit{iTrnI}: From \eqref{proof:opt-enf-bi--a-2} and rule \rtit{iTrnI} we know that $\actu{\,=\,}\actInV$ and that 
			\begin{gather}
				\eSembiP{\hV}\traS{\ioact{\actInVB}{\actInVB}}\eV'_{\hV}.  \label{proof:opt-enf-bi--a-7}
			\end{gather}
			We now inspect the cases for \hV.
			\begin{itemize}
				\item $\hV{\,\in\,}\set{\hFls,\hTru,\hVarX}$: The cases for \hFls and \hVarX do not apply since \eSembiP{\hFls} and \eSembiP{\hVarX} do not yield a valid monitor, while the case when $\hV{\,=\,}\hTru$ gets trivially satisfied since $\eSembiP{\hTru}{\,=\,}\eIden$ and $\mc{\eIden}{\actInVB\txr'}{\,=\,}0$.
				\item $\hV{\,=\,}\hBigAndUhVI$ where $\bigdistinct{i\in\IndSet}\actSNVI$: Since $\hV=\hBigAndUhVI$, by the definition of \eSembiS{-} we have that 
				\begin{align}
					\begin{array}{@{\hspace{-1mm}}r@{\hspace{-40mm}}cl}
						\eSembiP{\hV_{\hAndS}=\hBigAndUhVI} \\
						&=&\rec{\rVV}
						{
							\ch{
								\begin{xbrackets}{c}
									\chBigI
									\begin{xbrace}{lc}
									\dis{\pateI}{\bVI}{\rVV}{\prtSet}			& \ift \hVI=\hFls \\
									\prf{\actSIDs{\pateI}{\bVI}}{\eSembiP{\hVI}}		& \otherwiset
									\end{xbrace}
								\end{xbrackets}
							}{\defmon{\hBigAndDhVI}}
						}\\
						&=&
						{
							\ch{
								\begin{xbrackets}{c}
									\chBigI
									\begin{xbrace}{lc}
									\dis{\pateI}{\bVI}{\eSembiP{\hV_{\hAndS}}}{\prtSet} & \ift \hVI=\hFls \\
									\prf{\actSIDs{\pateI}{\bVI}}{\eSembiP{\hVI}}  & \otherwiset
									\end{xbrace}
								\end{xbrackets}
							}{\defmon{\hBigAndDhVI}}
						}
					\end{array}  \label{proof:opt-enf-bi--a-8}
				\end{align}
				Since normalized conjunctions are disjoint, \ie $\bigdistinct{i\in\IndSet}\actSNVI$, from \eqref{proof:opt-enf-bi--a-8} we can infer that the identity reduction in \eqref{proof:opt-enf-bi--a-7} can only happen when \actInV matches an identity branch, $\prf{\actSIDs{\pateJ}{\bVJ}}{\eSembiP{\hVJ}}$ (for some $j{\,\in\,}\IndSet$), and so we have that
				\begin{gather}
					\mtchS{\actSNVJ}{\actInV}=\sV. \label{proof:opt-enf-bi--a-9}
				\end{gather}
				Hence, knowing \eqref{proof:opt-enf-bi--a-7} and \eqref{proof:opt-enf-bi--a-9}, by rule \rtit{eTrn} we know that $\eV'_{\hV}{\,=\,}\eSembiP{\hVJ\sV}$ and so by \eqref{proof:opt-enf-bi--a-1} we can infer that 
				\begin{gather}
					\mc{\eV'_{\hV}}{\txr'}=\nV \; \wheret \eV'_{\hV}{\,=\,}\eSembiP{\hVJ\sV}. \label{proof:opt-enf-bi--a-10}
				\end{gather}
				Since from \eqref{proof:opt-enf-bi--a-8} we also know that the monitor branch $\prf{\actSIDs{\pateJ}{\bVJ}}{\eSembiP{\hVJ}}$ is derived from a non-violating modal necessity, \ie $\hNec{\actSNVJ}\hVJ$ where $\hVJ{\,\neq\,}\hFls$, we can infer that \actInV is not a violating action and so it should not be modified by any other monitor \eV, as otherwise it would infringe the eventual transparency constraint of assumption \eqref{proof:opt-enf-bi--a-3}. %for every $\pV{\,\in\,}\hSemS{\hV}$. 
				Therefore, we can deduce that 
				\begin{gather}
					\eV\traS{\ioact{\actInVB}{\actInVB}}\eV' \quad (\text{for some }\eV') \label{proof:opt-enf-bi--a-11} 
				\end{gather}
				and subsequently, knowing \eqref{proof:opt-enf-bi--a-11} and that $\txr{\,=\,}\actInVB\txr'$ and also that $\trcsys{\actInVB\txr'}{\traS{\actInV}}\trcsys{\txr'}$, by rule \rtit{iTrnI} and the definition of \mcS we infer that 
				\begin{gather}
					\mc{\eV}{\actInVB\txr'}=\mc{\eV'}{\txr'}. \label{proof:opt-enf-bi--a-12}
				\end{gather}
				As by \eqref{proof:opt-enf-bi--a-3}, \eqref{proof:opt-enf-bi--a-7}, \eqref{proof:opt-enf-bi--a-9} and \Cref{lemma:oenf-x2-bi} we know that $\enfdef{\eV'}{\hVJ\sV}$, by \eqref{proof:opt-enf-bi--a-10} and since $\pV\traS{\actInVB\txr'}$ entails that $\pV\traS{\actInV}\pV'$ and $\pV'\traS{\txr'}$, we can apply the \emph{inductive hypothesis} and deduce that $\nV{\,\leq\,}\mc{\eV'}{\txr'}$ and so from \eqref{proof:opt-enf-bi--a-12} we conclude that $\nV{\,\leq\,}\mc{\eV}{\actInVB\txr'}$ as required.
				\item $\hV{\,=\,}\hMaxX{\hV'}$ and $\hVarX{\,\in\,}\fv{\hV'}$: Since $\hV{\,=\,}\hMaxX{\hV'}$, by the syntactic restrictions of \SHMLnf we infer that $\hV'$ cannot be \hFls or \hTru since $\hVarX{\,\notin\,}\fv{\hV'}$ otherwise, and it cannot be \hVarX since every logical variable must be guarded. Hence, $\hV'$ must be of a specific form, \ie $\hMax{\hVarY_1\ldots\hVarY_n}{\hBigAndUhVI}$, and so by unfolding every fixpoint in $\hMaxX{\hV'}$ we reduce our formula to $\hV\defeq\hBigAndUhVI\subb{\subE{\hMaxX{\hV'}}{\hVarX},\ldots}$. We thus omit the remainder of this proof as it becomes identical to that of the subcase when $\hV{\,=\,}\hBigAndDhVI$.
			\end{itemize}
		\item \rtit{iTrnO}: We elide the proof for this case as it is very similar to that of \rtit{iTrnI}.
		\end{itemize}
	\end{Case}
	
	\begin{Case}[\mc{\eSembiP{\hV}}{\txr} \whent \txr{=}\actu\txr' \andt \eBI{\eSembiP{\hV}}{\trcsys{\actu\txr'}}{\traS{\actu'}}\eBI{\eV'_{\hV}}{\trcsys{\txr'}} \andt \actu'{\neq}\actu] \\
		Assume that 
		\begin{gather}
			\mc{\eSembiP{\hV}}{\actu\txr'}=1+\nVV \label{proof:opt-enf-bi--b-1} \\
			\wheret \nVV=\mc{\eV'_{\hV}}{\txr'} \label{proof:opt-enf-bi--b-2}
		\end{gather}
		which implies that 
		\begin{gather}	
			\eBI{\eSembiP{\hV}}{\trcsys{\actu\txr'}}\traS{\actu'}\eBI{\eV'_{\hV}}{\trcsys{\txr'}} \; \wheret\actu'\neq\actu \label{proof:opt-enf-bi--b-3}
		\end{gather}
		and also assume that 
		\begin{gather}
			\enfdef{\eV}{\hV} \label{proof:opt-enf-bi--b-4} 
		\end{gather}
		and that $\pV\traS{\actu\txr'}$. 
		Since we only consider action disabling monitors, the $\actu'$ reduction of \eqref{proof:opt-enf-bi--b-3} can only be achieved via rules \rtit{iDisO} or \rtit{iDisI}. We thus explore both cases. 
		\begin{itemize}
			\item \rtit{iDisI}: From \eqref{proof:opt-enf-bi--b-3} and by rule \rtit{iDisI} we have that $\actu=\actInV$ and $\actu'=\actt$ and that
			\begin{gather}
				\eSembiP{\hV}\traS{\ioact{\actdot}{\actInV}}\eV'_{\hV}. \label{proof:opt-enf-bi--b-6}
			\end{gather}
			We now inspect the cases for \hV.
			\begin{itemize}
				\item $\hV{\,\in\,}\set{\hFls,\hTru,\hVarX}$: These cases do not apply since \eSembiP{\hFls} and \eSembiP{\hVarX} do not yield a valid monitor, while $\eSembiP{\hTru}{\,=\,}\eIden$ does not perform the reduction in \eqref{proof:opt-enf-bi--b-6}.
				\item $\hV{\,=\,}\hBigAndUhVI$ where $\bigdistinct{i\in\IndSet}\actSNVI$: Since $\hV=\hBigAndUhVI$, by the definition of \eSembiS{-} we have that 
				\begin{align}
					\begin{array}{@{\hspace{-1mm}}r@{\hspace{-40mm}}cl}
					\eSembiP{\hV_{\hAndS}=\hBigAndUhVI} \\
					&=&\rec{\rVV}
					{
						\ch{
							\begin{xbrackets}{c}
							\chBigI
							\begin{xbrace}{lc}
							\dis{\pateI}{\bVI}{\rVV}{\prtSet}			& \ift \hVI=\hFls \\
							\prf{\actSIDs{\pateI}{\bVI}}{\eSembiP{\hVI}}		& \otherwiset
							\end{xbrace}
							\end{xbrackets}
						}{\defmon{\hBigAndDhVI}}
					}\\
					&=&
					{
						\ch{
							\begin{xbrackets}{c}
							\chBigI
							\begin{xbrace}{lc}
							\dis{\pateI}{\bVI}{\eSembiP{\hV_{\hAndS}}}{\prtSet} & \ift \hVI=\hFls \\
							\prf{\actSIDs{\pateI}{\bVI}}{\eSembiP{\hVI}}  & \otherwiset
							\end{xbrace}
							\end{xbrackets}
						}{\defmon{\hBigAndDhVI}}
					}
					\end{array}   \label{proof:opt-enf-bi--b-7}
				\end{align}
				Since normalized conjunctions are disjoint \ie $\bigdistinct{i\in\IndSet}\actSNVI$, and since $\pV\traS{\actu\txr'}$ where $\actu=\actInVB$, by the definition of \disS, from \eqref{proof:opt-enf-bi--b-7} we deduce that the reduction in \eqref{proof:opt-enf-bi--b-6} can only be performed by an insertion branch of the form, $\prf{\actSTI{\bVJ\Sub{\pidV}{\dvV}}{\actInV}}{\eSembiP{\hBigAndUhVI}}$ that can only be derived from a violating modal necessity $\hNec{\actSNVJ}\hFls$ (for some $j{\,\in\,}\IndSet$). Hence, we can infer that
				\begin{gather}
					\eV'_{\hV}=\eSembiP{\hBigAndUhVI} \label{proof:opt-enf-bi--b-8} \\
					\pateJ=\patInV \andt \ceval{\bVJ\Sub{\pidV}{\dvV}}{\boolT}. \label{proof:opt-enf-bi--b-9}
				\end{gather}
				Knowing \eqref{proof:opt-enf-bi--b-9} and that $\hNec{\actSNVJ}\hFls$ we can deduce that any input on port \pidV is erroneous and so for the soundness constraint of assumption \eqref{proof:opt-enf-bi--b-4} to hold, any other monitor \eV is obliged to somehow \emph{block} this input port. As we consider action disabling monitors, \ie $\eV{\,\in\,}\DisEnfS$, we can infer that monitor \eV may block this input in two ways, namely, either by not reacting to the input action, \ie $\eV\ntraS{\actInV}$, or by additionally inserting a default value \vV, \ie $\eV\traS{\ioact{\actdot}{\actInVB}}\eV'$. We explore both cases.
				\begin{itemize}
					\item $\eV\ntraS{\actInV}$: Since $\trcsys{\actInVB\txr'}\traS{\actInV}\trcsys{\txr'}$ and since $\eV\ntraS{\actInV}$, by the rules in our model we know that for every action $\actu'$, $\eBI{\eV}{\trcsys{\actInVB\txr'}}{\ntraS{\actu'}}$ and so by the definition of \mcS we have that $\mc{\eV}{\actInVB\txr'}{\,=\,}\len{\actInVB\txr'}$ meaning that by blocking inputs on \pidV, \eV also blocks (and thus modifies) every subsequent action of trace $\txr'$. Hence, this suffices to deduce that \emph{at worst} $1+\nVV$ is equal to $\len{\actInVB\txr'}$, that is $1+\nVV{\,\leq\,}\len{\actInVB\txr'}$, and so from \eqref{proof:opt-enf-bi--b-1} we can deduce that $1+\nVV{\,\leq\,}\mc{\eSembiP{\hV}}{\actu\txr'}$ as required.
					\item $\eV\traS{\ioact{\actdot}{\actInVB}}\eV'$: Since $\trcsys{\actInVB\txr'}\traS{\actInV}\trcsys{\txr'}$ and since $\eV\traS{\ioact{\actdot}{\actInVB}}\eV'$, by rule \rtit{iDisI} we know that $\eBI{\eV}{\trcsys{\actInVB\txr'}}\traS{\actt}\eBI{\eV}{\trcsys{\txr'}}$ and so by the definition of \mcS we have that 
					\begin{gather}
						\mc{\eV}{\actInVB\txr'}=1+\mc{\eV'}{\txr'}. \label{proof:opt-enf-bi--b-11}
					\end{gather}
					As by \eqref{proof:opt-enf-bi--b-4}, \eqref{proof:opt-enf-bi--b-6} and \Cref{lemma:oenf-x4-bi} we infer that $\enfdef{\eV'}{\hBigAndUhVI}$, by \eqref{proof:opt-enf-bi--b-2}, \eqref{proof:opt-enf-bi--b-11} and since $\pV{\,\traS{\actInVB\txr'}}$ entails that $\pV{\,\traS{\actInVB}\,}\pV'$ and $\pV'{\,\traS{\txr'}}$, we can apply the \emph{inductive hypothesis} and deduce that $\nVV{\,\leq\,}\mc{\eV'}{\txr'}$ and so from \eqref{proof:opt-enf-bi--b-1}, \eqref{proof:opt-enf-bi--b-2} and \eqref{proof:opt-enf-bi--b-11} we conclude that $1+\nVV{\,\leq\,}\mc{\eV}{\actInVB\txr'}$ as required.					
				\end{itemize}			
				\item $\hV{\,=\,}\hMaxX{\hV'}$ and $\hVarX{\,\in\,}\fv{\hV'}$: We omit showing this proof as it is a special case of when $\hV{\,=\,}\hBigAndUhVI$.
			\end{itemize}
			\item \rtit{iDisO}: We omit showing the proof for this subcase as it is very similar to that of case \rtit{iDisI}. Apart from the obvious differences (\eg \actOutV instead of \actInV), \Cref{lemma:oenf-x3-bi} is used instead of \Cref{lemma:oenf-x4-bi}.
		\end{itemize}
	\end{Case}

	\begin{Case}[\mc{\eSembiP{\hV}}{\txr} \whent \txr{\,\in\,}\set{\actu\txr',\txrE} \andt \eBI{\eSembiP{\hV}}{\trcsys{\actu\txr'}}\ntraS{\actu'}]
		Assume that 
		\begin{gather}
			\mc{\eSembiP{\hV}}{\txr}=\len{\txr} \quad (\text{where } \txr{\,\in\,}\set{\actu\txr',\txrE})  \label{proof:opt-enf-bi--c-1} \\
			\eBI{\eSembiP{\hV}}{\trcsys{\actu\txr'}}\ntraS{\actu'}  \label{proof:opt-enf-bi--c-2} \\
			\enfdef{\eV}{\hV} \label{proof:opt-enf-bi--c-2.1} 
		\end{gather}
		Since $\txr{\,\in\,}\set{\actu\txr',\txrE}$ we consider both cases individually.
		\begin{itemize}
			\item $\txr=\txrE:$ This case holds trivially since by \eqref{proof:opt-enf-bi--c-1}, \eqref{proof:opt-enf-bi--c-2} and the definition of \mcS, $\mc{\eSembiP{\hV}}{\txrE}=\len{\txrE}=0$.
			\item $\txr=\actu\txr':$ Since $\txr=\actu\txr'$ we can immediately exclude the cases when $\actu{\,\in\,}\set{\actt,\actOutV}$ since rules \rtit{iAsy} and \rtit{iDef} make it impossible for \eqref{proof:opt-enf-bi--c-2} to be attained in such cases. Particularly, rule \rtit{iAsy} always permits the \sus to independently perform an internal \actt-move, while rule \rtit{iDef} allows the monitor to default to \eIden whenever the system performs an unspecified output \actOutV. However, in the case of inputs, \actInV, the monitor may completely block inputs on a port \pidV and as a consequence cause the entire composite system $\eBI{\eSembiP{\hV}}{\trcsys{\actu\txr'}}$ to block, thereby making \eqref{proof:opt-enf-bi--c-2} a possible scenario. We thus inspect the cases for \hV vis-a-vis $\actu{\,=\,}\actInV$.
			\begin{itemize}
				\item $\hV{\,\in\,}\set{\hFls,\hTru,\hVarX}$: These cases do not apply since $\eSembiP{\hFls}$ and $\eSembiP{\hVarX}$ do not yield a valid monitor and since $\eSembiP{\hTru}{\,=\,}\eIden$ is incapable of attaining \eqref{proof:opt-enf-bi--c-2}.
				\item $\hV{\,=\,}\hBigAndDhVI$ where $\bigdistinct{i\in\IndSet}\actSNVI$: Since $\hV{\,=\,}\hBigAndDhVI$, by the definition of $\eSembiS{-}$ we have that 				
				\begin{align}
					\begin{array}{@{\hspace{-1mm}}r@{\hspace{-40mm}}cl}
					\eSembiP{\hV_{\hAndS}=\hBigAndUhVI} \\
					&=&\rec{\rVV}
					{
						\ch{
							\begin{xbrackets}{c}
							\chBigI
							\begin{xbrace}{lc}
							\dis{\pateI}{\bVI}{\rVV}{\prtSet}			& \ift \hVI=\hFls \\
							\prf{\actSIDs{\pateI}{\bVI}}{\eSembiP{\hVI}}		& \otherwiset
							\end{xbrace}
							\end{xbrackets}
						}{\defmon{\hBigAndDhVI}}
					}\\
					&=&
					{
						\ch{
							\begin{xbrackets}{c}
							\chBigI
							\begin{xbrace}{lc}
							\dis{\pateI}{\bVI}{\eSembiP{\hV_{\hAndS}}}{\prtSet} & \ift \hVI=\hFls \\
							\prf{\actSIDs{\pateI}{\bVI}}{\eSembiP{\hVI}}  & \otherwiset
							\end{xbrace}
							\end{xbrackets}
						}{\defmon{\hBigAndDhVI}}
					}
					\end{array}     \label{proof:opt-enf-bi--c-3}
				\end{align}
				Since $\actu=\actInV$, from \eqref{proof:opt-enf-bi--c-3} and by the definitions of \disS and \defmonS we can infer that the only case when \eqref{proof:opt-enf-bi--c-2} is possible is when the inputs on port \pidV satisfy a violating modal necessity, that is, there exists some $j{\,\in\,}\IndSet$ such that $\hNec{\actSNVJ}\hFls$ and for every $\vV{\,\in\,}\Val$, $\mtch{\pateJ}{\actInV}{\,=\,}\sV$ and $\ceval{\bVJ\sV}{\boolT}$. At the same time, the monitor is also \emph{unaware} of the port on which the erroneous input can be made, \ie
				%\begin{gather}
					$\pidV{\,\notin\,}\prtSet$. % \label{proof:opt-enf-bi--c-5}
				%\end{gather}
				%
				Hence, this case does not apply since we limit ourselves to \SysPrtSet, \ie states of system that can only input values via the ports specified in \prtSet.
				%
				%Although the monitor still blocks the invalid inputs, it does not define an insertion branch that internally provides the required input that allows the composite system to move on to the next state.
				%
				%Since from \eqref{proof:opt-enf-bi--c-4} we know that any input on port $\pidV$, $\actInV$, is erroneous, we can deduce that any other sound monitor, such as \eV (assumption \eqref{proof:opt-enf-bi--c-2.1}), must also block these invalid inputs. As we assume that $\eV{\,\in\,}\DisEnfS$ where $\DisEnfS{\,\defeq\,}\!\Setdef{\eVV}{\text{if }\trpBI{\eVV}{\,\subseteq\,}\set{\DIS} \text{ and } \si{\eVV}{\,\subseteq\,}\prtSet'}$ and $\prtSet'{\,\subseteq\,}\prtSet$, we can infer that $\pidVV{\,\notin\,}\si{\eV}$, and so by the definition of $\siS$ we know that \eV also fails to define an insertion branch that unblocks the \sus awaiting the input on the blocked port. This means that $\forall\actu'\cdot\eBI{\eV}{\trcsys{\actInVB\txr'}}\ntraS{\actu'}$ and so by \eqref{proof:opt-enf-bi--c-1} and the definition of \mcS we have that $\mc{\eV}{\actInVB\txr'}{\,=\,}\len{\actInVB\txr'}=\mc{\eSembiP{\hV}}{\actInVB\txr'}$ as required.
				%
				\item $\hV{\,=\,}\hMaxX{\hV'}$: As argued in previous cases, this 
        % subcase 
        is a special case of $\hV{\,=\,}\hBigAndUhVI$ and so we omit this part of the proof.
			\end{itemize}
		\end{itemize}
	\end{Case}
	\begin{Case}[\mc{\eSembiP{\hV}}{\txr} \whent \txr{\,\in\,}\set{\actu\txr',\txrE} \andt \eBI{\eSembiP{\hV}}{\trcsys{\txr}}\traS{\actu'}\eBI{\eV'_{\hV}}{\trcsys{\txr}}]
		As we only consider action disabling monitors, this case does not apply since the transition $\eBI{\eSembiP{\hV}}{\trcsys{\txr}}\traS{\actu'}\eBI{\eV'_{\hV}}{\trcsys{\txr}}$ can only be achieved via action enabling and rules \rtit{iEnO} and \rtit{iEnI}.\vspace*{-\baselineskip}
	\end{Case} 
  % \vspace{-4mm}
\end{proof}

%
%This result therefore ensures that one cannot find a \emph{less intrusive} action disabling (finite) monitor that has the same information portrayed by \prtSet and that performs fewer modifications to the system's behaviour at runtime than the ones we synthesise.
%
%However, it does not exclude that a monitor with different, or additional, transformation capabilities and port information might be more optimal when enforcing \hV, yet it ensures that one cannot find a more optimal monitor with the same \prtSet and transformation capabilities as \eSembiP{\eV}.
%
%

% % \section{Related Work}
% % \label{sec:related-work}
% % \input{related-work.tex}

\section{Conclusions and Related Work}
\label{sec:conclusion}
% !TEX root = journal.tex

This work extends the framework presented in the precursor to this work~\cite{Cassar2018Concur} to 
% carry out 
the setting of bidirectional enforcement where observable actions such as inputs and outputs require different treatment.
We achieve this by: 
\begin{enumerate}
  \item augmenting substantially our instrumentation relation (\Cref{fig:mod-bi-re});
  \item  refining our definition of enforcement to incorporate transparency over violating systems (\Cref{def:enforcement}); 
  \item providing a more extensive synthesis function (\Cref{def:synthesis-bi}) that is proven correct (\Cref{thm:enf}); and
  \item exploring notions of transducer optimality in terms of limited levels of intrusiveness (\Cref{def:opt-enf-bi,def:opt-enf-bi-finite} and \Cref{thm:oenf}). 
\end{enumerate}

\paragraph*{Future work.}
There are a number of possible avenues for extending our work.
One immediate step would be the implementation of the monitor operational model presented in \Cref{sec:model} together with the synthesis function described in \Cref{sec:synthesis}. 
This effort should be integrated it within the detectEr tool suite~\cite{Attard2016,Cassar2017Eaop,CassarFAAI17,FrancalanzaX20,AttardAAFIL21,AchilleosEFLX22,AcetoAAEFI22}.  
This would allow us to assess the overhead induced by our proposed bidirectional monitoring~\cite{AcetoAFI21}. 
Other tools that are closely related to detectEr's monitoring approach, such as STMonitor~\cite{BurloFS21,BurloFSTT22} and the tool by Lanotte \etal \cite{LanotteMM23} can, in principle, also adopt our theoretical framework for bidirectional enforcement in a fairly straightforward manner.

A bidirectional treatment of actions~\cite{Pinisetty2016,Pinisetty2017} can also be used to increase the precision in related runtime techniques such as monitoring with  blame-assignment~\cite{Jia2016Popl} and explainability~\cite{Explainability22,FrancalanzaC21}.
Bidirectional enforcement can also be extended to take into consideration quantitative concerns such as cost~\cite{FrancalanzaVH14,DasB0PS21}, which would impinge on our notion of monitor optimality from \Cref{sec:optimality}. 
Another possible avenue would be the development of behavioural theories for the transducer operational model presented in  \Cref{sec:model,sec:enforcement}, along the lines of the refinement preorders studied in earlier work on sequence recognisers~\cite{francalanza2016theory,Fra17:Concur,AcetoAFIL21}.
We believe the operational models presented in \Cref{sec:model} already lend themselves to such a treatment.

% \noindent
% \textsl
\paragraph*{Related work.} 
In his seminal work \cite{schneider2000}, Schneider introduced the concept of runtime enforcement and enforceability in a linear-time setting. 
Particularly, in his setting a property is deemed enforceable if its \emph{violation} can be \emph{detected} by a \emph{truncation automaton}, and prevented via system termination.
By preventing misbehaviour, these automata can only enforce safety properties.  
Ligatti \etal extended this work in \cite{Ligatti2005} via \emph{edit automata}---an enforcement mechanism capable of \emph{suppressing} and \emph{inserting} system actions. 
A property is thus enforceable if it can be expressed as an edit automaton that \emph{transforms} invalid executions into valid ones via suppressions and insertions.
%
%Edit automata are capable of enforcing instances of safety and liveness properties, along with other properties such as infinite renewal properties \cite{Ligatti2005,Bielova2011PhD}. 
%
As a means to assess the correctness of these automata, the authors introduced \emph{soundness} and \emph{transparency}.
Both settings by Schneider \cite{schneider2000} and Ligatti \etal \cite{Ligatti2005} assume a trace based view of the \sus and that every action can be freely manipulated by the monitor.
They also do not distinguish between the specification and the enforcement mechanism, as properties are encoded in terms of the enforcement model itself, \ie as edit/truncation automata.
In our prior work \cite{Cassar2018Concur}, we addressed this issue by separating the specification and verification aspects of the logic and explored the enforceability of \recHML in a unidirectional context and in relation to a definition of adequate enforcement defined in terms of soundness and transparency.
In this paper we adopt a stricter notion of enforceability that requires adherence to eventual transparency and investigate the enforceability of \SHML formulas in the context of bidirectional enforcement.
%

%\cmt{Bielova's Work on Predictability --> which relates to optimality but still under the general assumption.}
%
Bielova and Massacci \cite{Bielova2011PhD,Bielova2011Predictability} remark that, on their own, soundness and transparency fail to specify the extent in which a transducer should modify invalid runtime behaviour and thus introduce a \emph{predictability} criterion.
A transducer is said to be \emph{predictable} if one can predict the edit-distance between an invalid execution and a valid one.
With this criterion, adequate monitors are further restricted by setting an upper bound on the number of transformations that a monitor can apply to correct invalid traces.
Although this is similar to our notion of optimality, we however use it to compare an adequate (sound and eventual transparent) monitor to \emph{all} the other adequate monitors and determine whether it is the least intrusive monitor that can enforce the property of interest.
In \cite{Konighofer2017} K{\"o}nighofer \etal present a synthesis algorithm similar to our own that produces action replacement monitors called \emph{shields} from safety properties encoded as automata-based specifications. 
Although their shields can analyse both the inputs and outputs of a reactive system, they still perform unidirectional enforcement since they only modify the data associated with the system's output actions. %whenever it deviates from the specified behaviour. 
By definition, shields should adhere to correctness and minimum deviation which are, in some sense, analogous to soundness and transparency respectively. 
In \cite{Pinisetty2016,Pinisetty2017}, Pinisetty \etal conduct a preliminary investigation of RE in a bidirectional setting.
They, however, model the behaviour of the \sus as a trace of input and output pairs, \aka \emph{reactions}, and focus on enforcing properties by modifying the payloads exchanged by these reactions.
This way of modelling system behaviour is, however, quite restrictive as it only applies to synchronous reactive systems that output a value in reaction to an input.
This differs substantially from the way we model systems as LTSs, particularly since we can model more complex systems that may opt to collect data from multiple inputs, or supply multiple outputs in response to an input. 
The enforcement abilities studied in \cite{Pinisetty2016,Pinisetty2017} are also confined to action replacement that only allows the monitor to modify the data exchanged by the system in its reactions, and so the monitors in \cite{Pinisetty2016,Pinisetty2017} are unable to disable and enable actions.
Due to their trace based view of the system, their correctness specifications do not allow for defining correct system behaviour in view of its different execution branches.
This is particularly useful when considering systems whose inputs may lead them into taking erroneous computation branches that produce invalid outputs.
Moreover, since their systems do not model communication ports, their monitors cannot influence directly the control structure of the \sus, \eg by opening, closing or rerouting data through different ports.

Finally, Lanotte \etal~\cite{LanotteMM20,LanotteMM23} employ similar synthesis techniques and correctness criteria to ours (\Cref{def:senf,def:tenf}) to generate enforcement monitors for a timed setting. 
They apply their process-based approach to build tools that enforce data-oriented security properties. 
Although their implementations handle the enforcement of first-order properties, the theory on which it is based does not, nor does it investigate logic enforceability.

% ------The Old version from FORTE-------
%
% As we discussed already in the Introduction, most work on RE assumes a trace-based view of the \sus\cite{schneider2000,Ligatti2005,Ligatti2009}, where few distinguish between actions with different control profiles (\eg inputs versus outputs).
% %
% Although shields~\cite{Konighofer2017} can analyse both  input and output actions, they still perform unidirectional enforcement and only modify the data associated with the output actions. 
% %
% The closest to our work is that by Pinisetty \etal~\cite{Pinisetty2017}, who consider bidirectional RE, modelling the system as a trace of input-output pairs. 
% %
% However, their enforcement is limited to replacements of payloads and their setting is too restrictive to model enforcements such as action rerouting and the closing of ports.
% %
% Finally, Lanotte \etal~\cite{LanotteMM20} employ similar synthesis techniques and correctness criteria to ours (\Cref{def:senf,def:tenf}) to generate enforcement monitors for a timed setting. 

\bibliographystyle{alphaurl}
\bibliography{refs}

\end{document}